\documentclass[11pt]{amsart}
 \usepackage{amsmath,amssymb,amsfonts,amsthm,amscd,textcomp,times,booktabs}
 \usepackage[dvipdfmx]{graphicx}
  \usepackage{braket}
  \usepackage{enumerate}
  \usepackage{color,float}
  \usepackage{bm}
  \usepackage{multirow}
  \usepackage{arydshln}
  \usepackage[all]{xy}
  \usepackage[normalem]{ulem}
\newtheorem{theorem}{Theorem}

\newtheorem{proposition}[theorem]{Proposition}
\newtheorem{lemma}[theorem]{Lemma}
\newtheorem{definition}[theorem]{Definition}
\newtheorem{corollary}[theorem]{Corollary}
\newtheorem{remark}[theorem]{Remark}

\newtheorem{construction}[theorem]{Construction}
\numberwithin{equation}{section}
\numberwithin{theorem}{section}

\newcommand{\R}{\ensuremath{\mathbb{R}}}
\newcommand{\Lt}{\ensuremath{\mathrm{L}}}
\newcommand{\Rt}{\ensuremath{\mathrm{R}}}
\newcommand{\norm}[1]{\ensuremath{\left| #1 \right|}}
\newcommand{\ora}[1]{\ensuremath{\overrightarrow{#1}}}
\newcommand{\inn}{\ensuremath{\mathrm{in}}}
\newcommand{\out}{\ensuremath{\mathrm{out}}}
\newcommand{\conv}{\ensuremath{\mathrm{conv}}}
\newcommand{\new}{\ensuremath{\mathrm{new}}}
\newcommand{\ave}{\ensuremath{\mathrm{ave}}}

\newcommand{\proj}{\ensuremath{\mathrm{pr}}}
\DeclareMathOperator{\Arctan}{Arctan}
\begin{document}
\title[New efficient 3D gadgets in origami extrusions]%
{New efficient flat-back 3D gadgets in origami extrusions compatible with the conventional pyramid-supported 3D gadgets}

\author{Mamoru Doi}
\address{11-9-302 Yumoto-cho, Takarazuka, Hyogo 665-0003, Japan}
\email{doi.mamoru@gmail.com}
\maketitle
\noindent{\bfseries Abstract. }
An origami extrusion is a folding of a $3$D object in the middle of a flat piece of paper, using $3$D gadgets which create faces with solid angles.
Our main concern is to make origami extrusions of polyhedrons using $3$D gadgets with simple outgoing pleats, where a simple pleat
is a pair of a mountain fold and a valley fold which are parallel to each other.
In this paper we present a new type of $3$D gadgets with simple outgoing pleats in origami extrusions and their construction.
Our $3$D gadgets are downward compatible with the conventional pyramid-supported gadgets developed by Calros Natan as a generalization of the cube gadget, 
in the sense that in many cases we can replace the conventional gadgets with the new ones with the same outgoing pleats while the converse is not always possible.
We can also change angles of the outgoing pleats under certain conditions.
Unlike the conventional pyramid-supported $3$D gadgets, the new ones have flat back sides above the ambient paper, 
and thus we can make flat-foldable origami extrusions.
Furthermore, since our new $3$D gadgets are less interfering with adjacent gadgets than the conventional ones, 
we can use wider pleats at one time to make the extrusion higher.
For example, we prove that the maximal height of the prism of any convex polygon (resp. any triangle)
that can be extruded with our new gadgets is more than $4/3$ times (resp. $\sqrt{2}$ times) of that with the conventional ones. 
We also present explicit constructions of division/repetition and negative versions of the new $3$D gadgets.

\section{Introduction}
There are many studies on folding $3$D objects from a flat piece of paper.
In particular, E. Demaine, M. Demaine and Mitchell proved in \cite{DDM} that any polyhedron can be folded from a piece of paper,
and a practical software package for designing the crease pattern of any polyhedron was developed by E. Demaine and Tachi \cite{DT}.

In this paper we are concerned with \emph{origami extrusions}, which are a different type of $3$D origami.
An origami extrusion is a folding of a $3$D object in the middle of a flat piece of paper, using \emph{$3$D gadgets} which create faces with solid angles.
Among many kinds of $3$D gadgets known in origami extrusions, the \emph{cube gadget}, four copies of which are shown in Figure $\ref{fig:cube_conv}$,
was independently discovered by David Huffman and other paper folders, which creates the top and two vertical faces of a cube,
where the side faces created are supported by a triangular pyramid from inside 
(see \cite{BDDO}, in which this cube gadget is referred to as the $\Arctan (1/2)$ gadget).
In \cite{BDDO}, Benbernou, E. Demaine, M. Demaine, and Ovadyr studied  origami extrusions of polycubes introducing three kinds of
cube gadgets including the above one.
Also, E. Demaine, M. Demaine and Ku extruded in \cite{DDK} any orthogonal maze introducing six kinds of vertex gadgets.
Notably, Carlos Natan generalized in \cite{Natan} the pyramid-supported cube gadget to one which creates a top face parallel to the ambient paper, 
and two side faces in a general position which share a ridge.
In each of Natan's $3$D gadgets, a triangular pyramid supports from inside the two side faces created, 
and the side faces do not attach to each other directly, the ridge between which we can see the inside space.

A common feature of the above $3$D gadgets is that they have \emph{simple} outgoing pleats, 
where a simple pleat is a pair of a mountain fold and a valley fold which are parallel to each other.
In origami extrusions, the use of $3$D gadgets with simple outgoing pleats
is of great advantage to successive extrusions because the paper remains flat outside the extruded object after the folding.
See also Cheng's paper \cite{Cheng} for examples of successive extrusions.
Simple outgoing pleats have another advantage in making \emph{origami $3$D tessellations} 
in which adjacent extrusions share some of their outgoing pleats.

Our main concern is to make origami extrusions of polyhedrons using $3$D gadgets with simple outgoing pleats.
The purpose of this paper is to present a new type of $3$D gadgets in origami extrusions, which improve the conventional $3$D gadgets by Natan
in several respects, where we use the term `the conventional $3$D gadgets' to indicate the above pyramid-supported ones. 

Firstly, our new $3$D gadgets are \emph{downward compatible} with the conventional ones. 
This means that we can replace the conventional gadgets with
the new ones with the same outgoing pleats in many cases (see Theorem $\ref{thm:downward_compatibility}$), 
but the converse replacement is not always possible.
For example, the crease pattern (abbreviated as `CP') in Figure $\ref{fig:cube_conv}$ 
can be replaced by that in Figure $\ref{fig:cube_new}$ using our new cube gadgets.
\begin{figure}[htbp]
  \begin{center}
    \begin{tabular}{c}
\addtocounter{theorem}{1}
      \begin{minipage}{0.5\hsize}
        \begin{center}
          \includegraphics[width=\hsize]{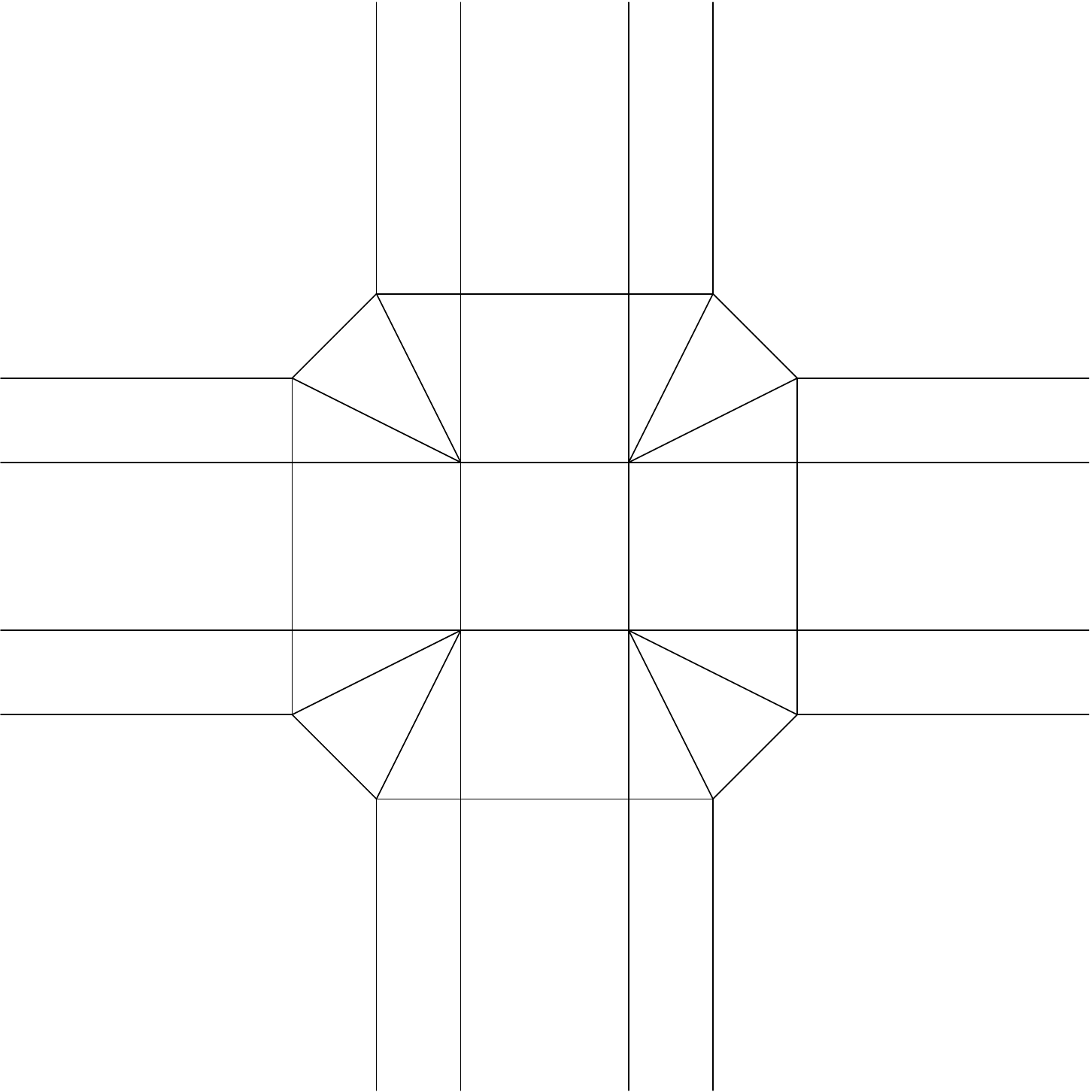}
        \end{center}
    \caption{CP of the cube extruded with the conventional cube gadgets}
    \label{fig:cube_conv}
      \end{minipage}
\addtocounter{theorem}{1}
      \begin{minipage}{0.5\hsize}
        \begin{center}
          \includegraphics[width=\hsize]{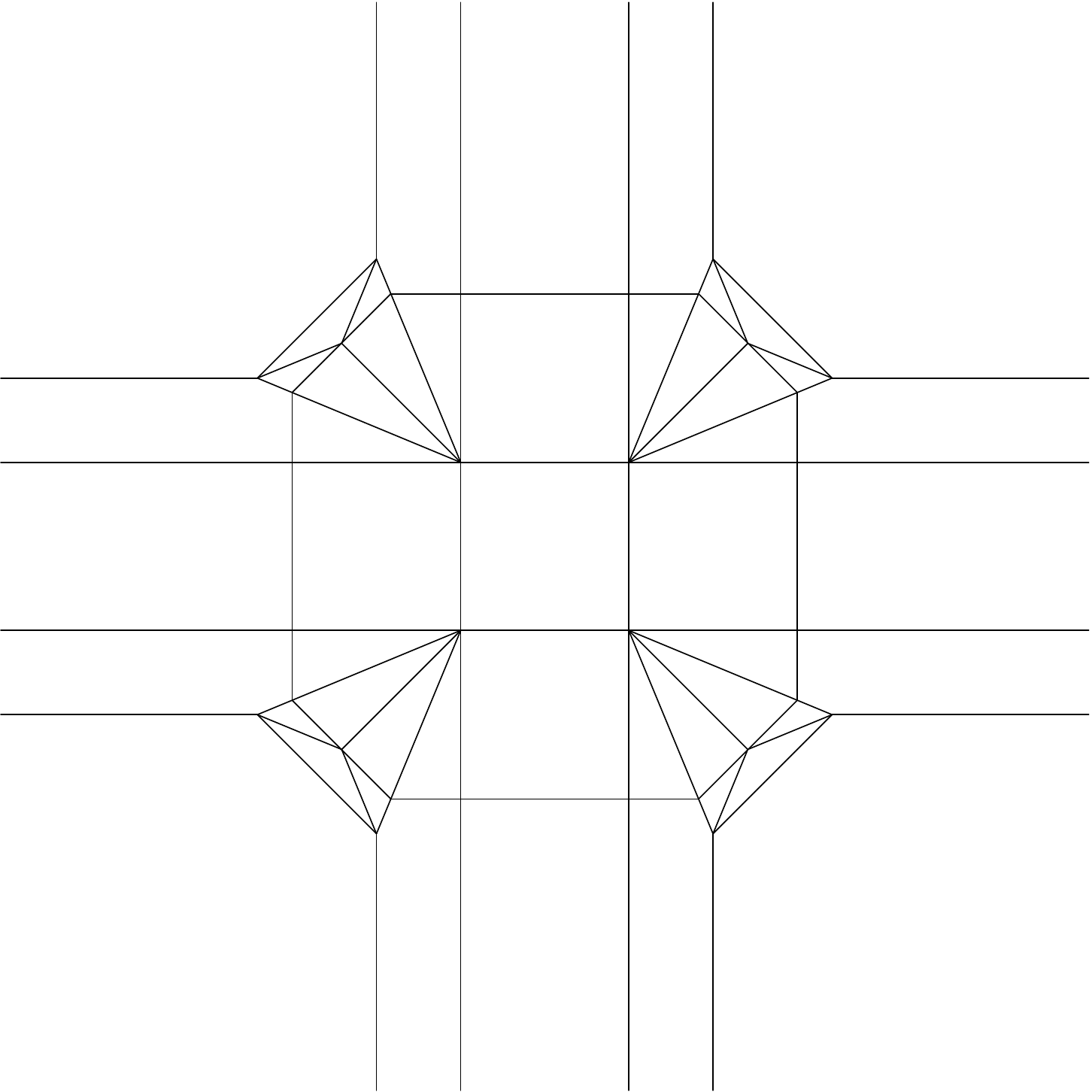}
        \end{center}
    \caption{CP of the cube extruded with the new cube gadgets}
    \label{fig:cube_new}
      \end{minipage}
    \end{tabular}
  \end{center}
\end{figure}

Secondly, the new $3$D gadgets are less interfering with adjacent gadgets, and thus more effieient than the conventional ones.
Observe that in the resulting extrusion of Figure $\ref{fig:cube_conv}$, the triangular pyramid inside each cube gadget touches those of the adjacent gadgets, 
while the `ears' and `tongue' (see the third paragraph in Section $\ref{sec:5}$) of each new cube gadget of Figure 
$\ref{fig:cube_new}$ does not interfere with adjacent gadgets. 
Indeed, we can further make the extruded square prism higher by widening the outgoing pleats. 
The crease pattern of the highest square prism, $\sqrt{2}$ times as high as the cube, is shown in Figure $\ref{fig:cube_new_max}$, 
in which we cannot replace our new cube gadget with the conventional one. 

Thirdly, our new $3$D gadgets have flat back sides above the ambient paper, and no gap between the side faces to see the inside unlike the conventional ones. 
If we add some creases to the crease pattern of Figure $\ref{fig:cube_new}$ as in Figure $\ref{fig:flat-foldable}$, 
we can extrude a cube which is flat-foldable by a twist with the added dotted creases.
Also, we can make an extrusion with curved creases (see Section $\ref{sec:9}$).

Fourthly, we can easily change angles of the two outgoing pleats associated with a new $3$D gadget independently 
under a certain condition although the crease pattern changes completely, which is impossible or at least difficult for the conventional gadgets.
\begin{figure}[htbp]
  \begin{center}
    \begin{tabular}{c}
\addtocounter{theorem}{1}
      \begin{minipage}{0.5\hsize}
        \begin{center}
          \includegraphics[width=\hsize]{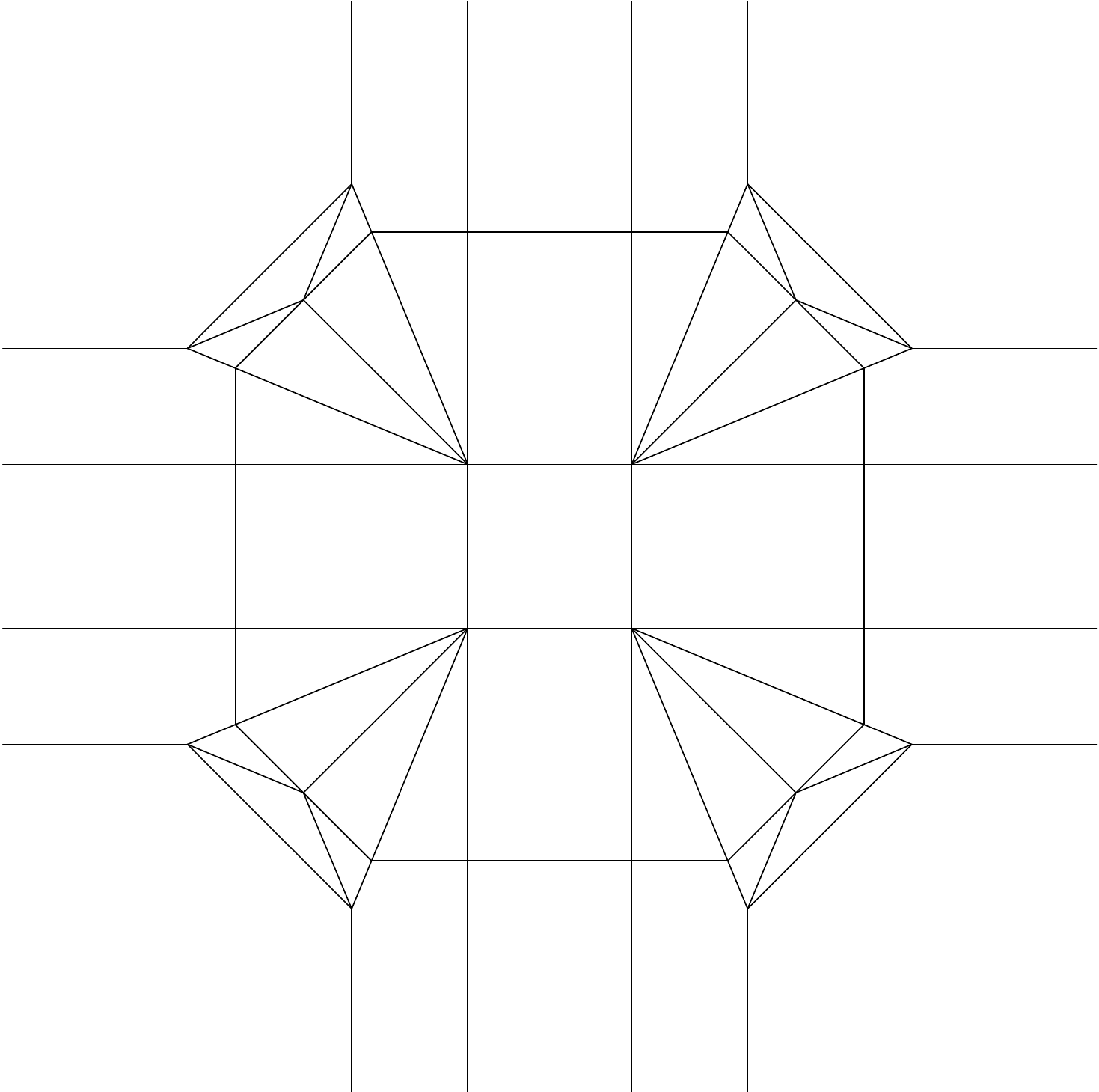}
        \end{center}
    \caption{CP of the square prism of maximal height extruded with the new cube gadgets}
    \label{fig:cube_new_max}
      \end{minipage}
\addtocounter{theorem}{1}
      \begin{minipage}{0.5\hsize}
        \begin{center}
          \includegraphics[width=\hsize]{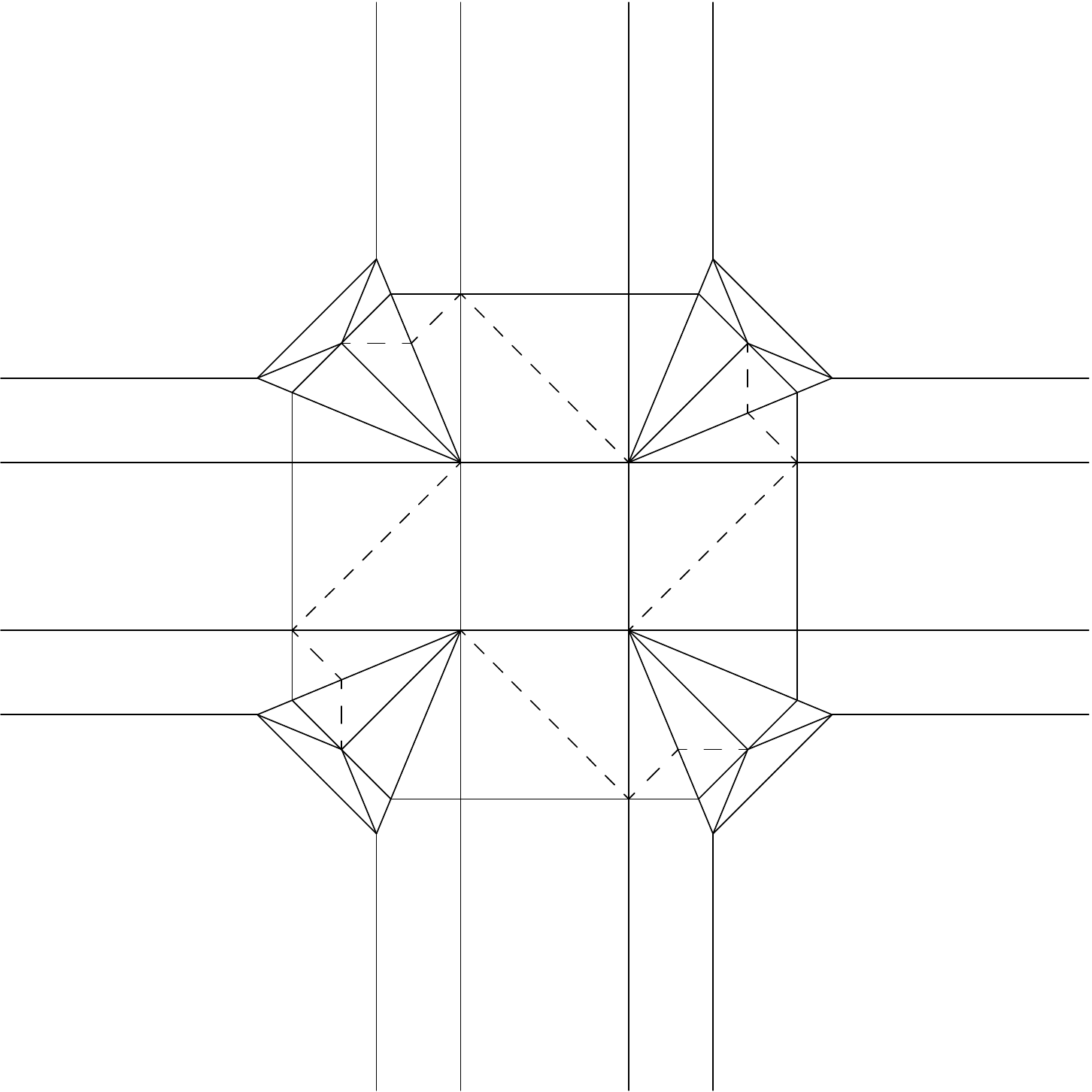}
        \end{center}
    \caption{CP of a flat-foldable cube extruded with the new cube gadgets}
    \label{fig:flat-foldable}
      \end{minipage}
    \end{tabular}
  \end{center}
\end{figure}

This paper is organized as follows. 
Section $\ref{sec:2}$ reviews the construction of the conventional $3$D gadgets developed by Natan, fixing our settings and notation throughout this paper.
In Section $\ref{sec:3}$, we give an explicit construction of our new gadgets in a general setting. 
In Section $\ref{sec:4}$ we discuss the validity of the conditions given in the construction of our new gadgets.
We also check the constructibility and foldablity of the resulting crease pattern. 
In Section $\ref{sec:5}$ we introduce the notion of interference coefficients for both the conventional and the $3$D new gadgets. 
Also, we give formulas in Theorems $\ref{thm:height_conv}$ and $\ref{thm:height_new}$ 
for calculating the maximal height of the resulting extrusion from initially given data.
Section $\ref{sec:6}$ is devoted to the proof of the downward compatibility theorem of the new gadgets with the conventional gadgets. 
Section $\ref{sec:7}$ focuses on extruding prisms of convex polygons and refine the formulas given in Theorem $\ref{thm:height_conv}$.
In particular, we prove in Corollary $\ref{cor:comparison_height_prism}$ that the maximal height of the prism of any convex polygon (resp. any triangle)
that can be extruded with our new gadgets is more than $4/3$ times (resp. $\sqrt{2}$ times) of that with the conventional ones. 
In Section $\ref{sec:8}$ we consider repetition of a $3$D gadget to make the extrusion higher with the same interference distances.
This is also regarded as division of a gadget to make the interference distances smaller.
In Section $\ref{sec:9}$ we give some remarks on our constructions of the new $3$D gadgets and further applications such as negative $3$D gadgets and
extrusions with curved creases.
Finally, Section $\ref{sec:conclusion}$ gives our conclusion.
Throughout this paper, we assume that the paper we fold is `ideal', and thus its thickness can be ignored.

As mentioned above, we are concerned with extruding polyhedrons on a flat piece of paper.
In particular, a polyhedron $\Delta$ we want to extrude is as follows.
\begin{enumerate}[(a)]
\item The bottom and the top faces $P_0,P_h$ are included in the $xy$-plane $H_0=\set{z=0}$ and $H_h=\set{z=h},h>0$ respectively.
\item The bottom face $P_0$ is a convex polygon $B_1 B_2\dots B_N$ for some $N\geqslant 3$, where 
the subscripts of $B$'s are numbered in the counterclockwise direction and $B_0=B_N, B_{N+1}=B_1$ for the later convenience.
\item Each side face $T_i$ for $i=1,\dots ,N$ is a trapezoid $A_i B_i B_{i+1} A_{i+1}$ with $A_i,A_{i+1}$ in $H_h$, and thus edge $A_iA_{i+1}$ 
is parallel to edge $B_i B_{i+1}$, where $A_0=A_N,A_{N+1}=A_1$ and we allow the case $A_i=A_{i+1}$, that is, $T_i$ is a triangle.
\item Outside $\Delta$, the paper is flat and there are only simple pleats, 
where a simple pleat is a pair of a mountain fold and a valley fold which are parallel to each other.
\end{enumerate}
Throughout this paper, we will use $A_1,A_2,\dots$ for the vertices of the top face $P_h$, and $B_1,B_2,\dots$ for those of the bottom face $P_0$.
In principle, from condition (d), we can successively extrude polyhedrons satisfying the above conditions (a)--(d) 
so that the top or a side face of each polyhedron includes the bottom face of the previous polyhedron.
However, more layers of polyhedrons make the crease pattern much more complicated, 
and so we need to adjust widths and angles of the outgoing pleats appropiately.
For $i=1,\dots N$, let $\alpha_i$  to be the inner angle $\angle B_{i-1}B_i B_{i+1}$ of $P_0=B_1\dots B_N$ at vertex $B_i$, and 
$\beta_{i,\Lt}=\pi -\angle A_iB_i B_{i-1},\beta_{i,\Rt}=\pi -\angle A_iB_iB_{i+1}$, where the subscripts $\Lt$ and $\Rt$ stand for
`left' and `right' respectively. 
We have 
\begin{equation}\label{def:angles}
\angle A_{i-1}A_i A_{i+1}=\alpha_i,\quad\angle A_{i-1}A_i B_i=\beta_{i,\Lt},\quad\text{and}\quad\angle B_i A_i A_{i+1}=\beta_{i,\Rt}
\end{equation}
if $A_{i-1}\neq A_i$ and $A_i\neq A_{i+1}$. 
For consistency, we consider $\eqref{def:angles}$ as the definition of the angles on the left-hand side when $A_{i-1}=A_i$ or $A_i=A_{i+1}$.
Note that the above $\alpha_i,\beta_{i,\Lt},\beta_{i,\Rt}$ satisfy the following conditions:
\begin{enumerate}[(i)]
\item $\alpha_i <\beta_{i,\Lt}+ \beta_{i,\Rt}$, $\beta_{i,\Lt}<\alpha_i+ \beta_{i,\Rt}$ and $\beta_{i,\Rt}<\alpha_i+ \beta_{i,\Lt}$, and
\item $\alpha_i +\beta_{i,\Lt} +\beta_{i,\Rt} <2\pi$ for all $i=1,\dots ,N$.
\end{enumerate}
Condition $\mathrm{(i)}$ is necessary for each $A_i$ to form a non-flat solid angle.
Condition $\mathrm{(ii)}$ enables us to develop $\Delta$ without dividing any face of $\Delta$, preserving angles between edges of $\Delta$ in $P_0$ and $P_h$.

\section{The conventional $3$D gadgets developed by Natan}\label{sec:2}
In this section we review the conventional $3$D gadgets developed by Carlos Natan \cite{Natan}, which generalizes the \emph{cube gadget} given in Figure 
$\ref{fig:cube_conv}$.
As a generalization of the cube gadget, let us consider a \emph{local model} of an above-mentioned polyhedron $\Delta$ as follows.
\begin{enumerate}[(A)]
\item The object we want to extrude in the middel of the paper (which we suppose to be in the $xy$-plane $H_0=\set{z=0}$) with a $3$D gadget 
is a part of a polyhedron such that its top face is bounded by two rays $j_\Lt$ and $j_\Rt$ starting from $A$ in $H_h=\set{z=h}$,
and its bottom face is bounded by two rays $k_\Lt$ and $k_\Rt$ with $k_\sigma$ parallel to $j_\sigma$ for $\sigma =\Lt ,\Rt$, 
starting from $B$ in $H_0=\set{z=0}$.
Suppose the inner angle $\alpha$ of the top face at $A$ (and so the inner angle of the bottom face at $B$) satisfies $0<\alpha <\pi$.
\item There are only simple pleats outside the extruded object.
\end{enumerate}
Let $\beta_\sigma =\angle BAj_\sigma =\pi -\angle ABk_\sigma$ for $\sigma =\Lt ,\Rt$ and $\gamma =2\pi -\alpha -\beta_\Lt -\beta_\Rt$.
Then we develop the top and side faces on the paper as in Figure $\ref{fig:development_conv}$, 
where $\ell_\Lt$ and $\ell_\Rt$ are the mountain folds of the outgoing pleats, and $A B_\Lt$ and $A B_\Rt$ assemble to form ridge $AB$.
\begin{figure}[htbp]
  \begin{center}
\addtocounter{theorem}{1}
          \includegraphics[width=0.75\hsize]{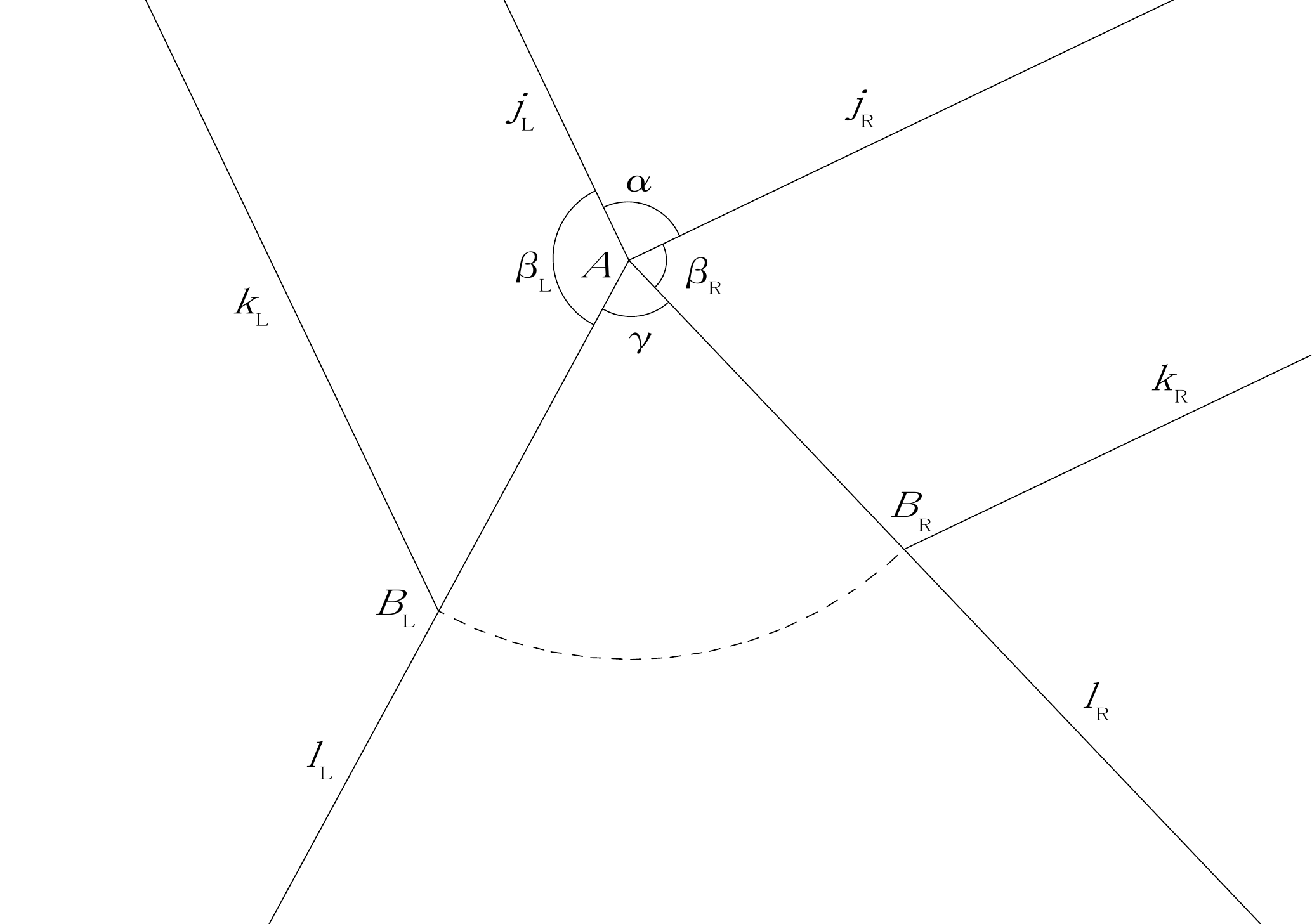}
    \caption{Delolopment to which we apply a conventional $3$D gadget}
    \label{fig:development_conv}
\end{center}
\end{figure}
\begin{lemma}\label{lem:lambda}
The height $h$ of the resulting extrusion is given by
\begin{equation*}
h=\lambda (\alpha ,\beta_\Lt ,\beta_\Rt)\norm{A B},
\end{equation*}
where $\norm{AB}$ denotes the length of ridge $AB$, and $\lambda (\alpha ,\beta_\Lt ,\beta_\Rt)$ is given by
\begin{equation}\label{lambda}
\lambda (\alpha ,\beta_\Lt ,\beta_\Rt )=\left( 1-\frac{\cos^2\beta_\Lt +\cos^2\beta_\Rt-2\cos\alpha\cos\beta_\Lt\cos\beta_\Rt}{\sin^2\alpha}\right)^{1/2}.
\end{equation}
\end{lemma}
\begin{proof}
Since $\lambda =\lambda (\alpha ,\beta_\Lt ,\beta_\Rt )$ is scale-invariant, 
we may assume $\norm{AB}=1$ without loss of generality, in which case we have $h=\lambda$.
For $\sigma =\Lt ,\Rt$, let $\ora{u_\sigma}$ be the unit direction vector of ray $k_\sigma$, and $\ora{v_\sigma}=\ora{A B_\sigma}$.
Also, let $\ora{w}$ be the vector $\ora{A B}$ of the resulting extrusion.
Then $\ora{w}$ is obtained by a rotation of $\ora{u_\sigma}$ around $\ora{v_\sigma}$ for both $\sigma =\Lt ,\Rt$.
Thus we have
\begin{align}
\ora{u_\Rt}\cdot\ora{w}&=\ora{u_\Rt}\cdot\ora{v_\Rt},\label{product_uv_R}\\
\ora{u_\Lt}\cdot\ora{w}&=\ora{u_\Lt}\cdot\ora{v_\Lt},\label{product_uv_L}\\
\norm{\ora{u_\Lt}}&=\norm{\ora{u_\Rt}}=\norm{\ora{v_\Lt}}=\norm{\ora{v_\Rt}}=\norm{\ora{w}}=1\label{norm_uvw}.
\end{align}
If we take $\ora{u_\Rt}$ to be the unit vector of the $x$-axis, from $\eqref{norm_uvw}$ we have
\begin{equation*}
\begin{aligned}\label{vector_uv}
\ora{u_\Rt}&=(1,0,0),&\ora{v_\Rt}&=(\cos\beta_\Rt ,-\sin\beta_\Rt ,0),\\
\ora{u_\Lt}&=(\cos\alpha ,\sin\alpha ,0),&\ora{v_\Lt}&=(\cos (\alpha +\beta_\Lt ),\sin (\alpha +\beta_\Lt ),0).
\end{aligned}
\end{equation*}
Letting $\ora{w}=(x,y,z)$ and substituting $\eqref{vector_uv}$ into $\eqref{product_uv_R}$, we have 
\begin{equation}\label{component_w}
x=\cos\beta_\Rt .
\end{equation}
Then substituting $\eqref{component_w}$ into $\eqref{product_uv_L}$ gives that
\begin{equation*}
\cos\alpha\cos\beta_\Rt +y\sin\alpha =\cos\beta_\Lt .
\end{equation*}
Hence we have
\begin{align}
\lambda (\alpha ,\beta_\Lt ,\beta_\Rt )^2&=h^2=\norm{z}^2=1-x^2-y^2\\
&=1-\cos^2\beta_\Rt -\left(\frac{\cos\beta_\Lt -\cos\alpha\cos\beta_\Rt}{\sin\alpha}\right)^2\\
&=1-\frac{\cos^2\beta_\Lt +\cos^2\beta_\Rt-2\cos\alpha\cos\beta_\Lt\cos\beta_\Rt}{\sin^2\alpha},
\end{align}
which completes the proof.
\end{proof}
\begin{construction}\label{const:conv}\rm
Consider a development as shown in Figure $\ref{fig:development_conv}$, for which we require the following conditions.
\begin{enumerate}[(i)]
\item $\alpha <\beta_\Lt + \beta_\Rt$, $\beta_\Lt <\alpha +\beta_\Rt$ and $\beta_\Rt <\alpha+ \beta_\Lt$.
\item $\alpha +\beta_\Lt +\beta_\Rt <2\pi$.
\item $\alpha +\beta_\Lt +\beta_\Rt >\pi$.
\end{enumerate}
Then the crease pattern of Natan's $3$D gadget is constructed as follows, where all procedures are done for both $\sigma =\Lt ,\Rt$.
\begin{enumerate}
\item Draw a perpendicular to $\ell_\sigma$ through $B_\sigma$ for both $\sigma =\Lt ,\Rt$, letting $C$ be the intersection point.
\item Draw a perpendicular bisector $m_\sigma$ to segment $B_\sigma C$.
\item Determine a point $D_\sigma$ on $m_\sigma$ so that $\angle AB_\sigma D_\sigma=\pi -\beta_\sigma$, and restrict $m_\sigma$ to a ray 
starting from $D_\sigma$ and going in the same direction as $\ell_\sigma$.
\item The desired crease pattern is shown as the solid lines in Figure $\ref{fig:gadget_conv}$, 
and the assignment of mountain folds and valley folds is given in Table $\ref{tbl:assignment_conv}$.
\end{enumerate}
\end{construction}
\begin{figure}[htbp]
  \begin{center}
\addtocounter{theorem}{1}
          \includegraphics[width=0.75\hsize]{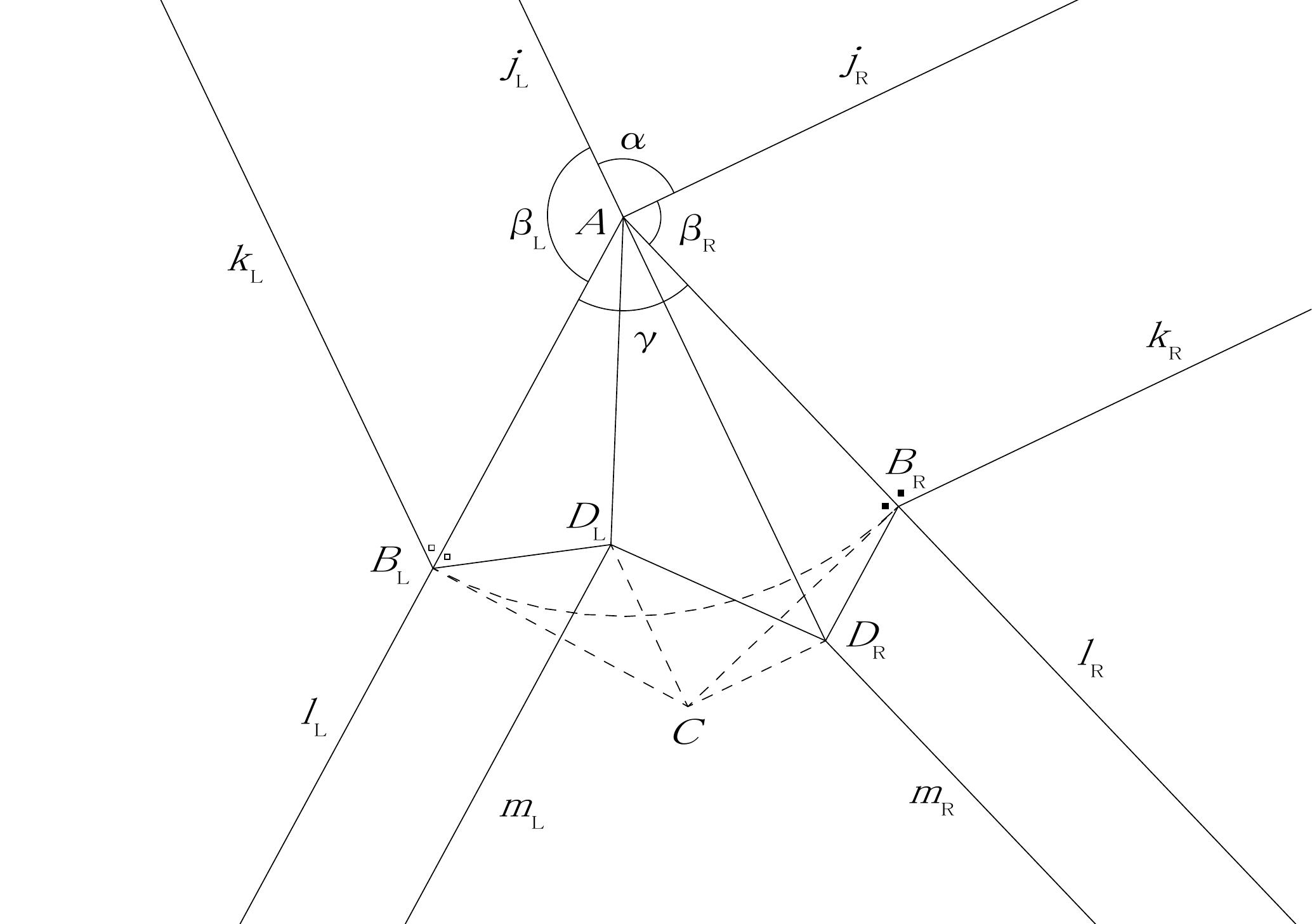}
    \caption{CP of the conventional $3$D gadget}
    \label{fig:gadget_conv}
\end{center}
\end{figure}
\renewcommand{\arraystretch}{1.5}
\addtocounter{theorem}{1}
\begin{table}[h]
\begin{tabular}{c|c}
mountain folds&$j_\sigma ,\ell_\sigma ,AB_\sigma ,B_\sigma D_\sigma$\\ \hline
valley folds&$k_\sigma ,m_\sigma ,AD_\sigma ,D_\Lt D_\Rt$
\end{tabular}\vspace{0.5cm}
\caption{Assignment of mountain folds and valley folds to the conventional $3$D gadget}
\label{tbl:assignment_conv}
\end{table}
Conditions $\mathrm{(i)}$ and $\mathrm{(ii)}$ are the same as given in the end of Introduction, and these together imply $\alpha <\pi$ as desired.
Condition $\mathrm{(iii)}$ is necessary for the existence of the point $C$ inside angle $\gamma$ in Figure $\ref{fig:development_conv}$. 
Note that for this gadget $B_\Lt$ and $B_\Rt$ move to $C$ and 
$\triangle A B_\Lt D_\Lt ,\triangle A B_\Rt D_\Rt ,\triangle A D_\Lt D_\Rt$ and $\triangle C D_\Lt D_\Rt$
form a triangular pyramid to support from inside the top and side faces. 

We rewrite conditions $\mathrm{(i)}$--$\mathrm{(iii)}$ in terms of $\beta_\Lt ,\beta_\Rt$ and $\gamma$ as the following for the later convenience, 
which is clear by a straightforward calculation.
\begin{lemma}\label{lem:condition_gamma}
Conditions $\mathrm{(i)}$--$\mathrm{(iii)}$ of Construction $\ref{const:conv}$ are equivalent to the following.
\begin{enumerate}
\item[$(\mathrm{i}')$] $\beta_\Lt +\beta_\Rt +\gamma /2>\pi$, $\beta_\Lt +\gamma /2<\pi$ and $\beta_\Rt +\gamma /2<\pi$.
\item[$(\mathrm{ii}')$] $\gamma >0$.
\item[$(\mathrm{iii}')$] $\gamma <\pi$.
\end{enumerate}
\end{lemma}
\begin{remark}\rm
Let $P$ be the intersection point of the extensions of $m_\Lt$ and $m_\Rt$. 
Then $P$ is the circumcenter of $\triangle AB_\Lt B_\Rt$, so that $\angle AB_\sigma P=\gamma /2$. 
If $\pi -\beta_\sigma \leqslant\gamma /2$, then segment $AD_\sigma$ intersects (an extension of) $m_{\sigma'}$ before it intersects $m_\sigma$,
where $\sigma'$ denotes the other side of $\sigma$. 
However, this is impossible. 
Indeed, if $\pi -\beta_\sigma\leqslant\gamma /2$, then we have
\begin{equation*}
\pi =2\pi -\pi \geqslant2\pi -\left(\beta_\sigma +\frac{\gamma}{2}\right) =\alpha +\beta_{\sigma'}+\frac{\gamma}{2},
\end{equation*}
and thus
\begin{equation*}
\alpha +\beta_{\sigma'}\leqslant\pi -\frac{\gamma}{2}\leqslant\beta_\sigma ,
\end{equation*}
which does not satisfy condition $\mathrm{(i)}$ of Construction $\ref{const:conv}$.
\end{remark}
\section{Constructing the crease patterns of the new $3$D gadgets}\label{sec:3}
To construct the crease pattern of our new $3$D gadget, we begin with a development as in Figure $\ref{fig:development_new}$, 
which is similar to Figure $\ref{fig:development_conv}$ but different 
in that we allow changes of angles $\delta_\sigma\geqslant 0$ of the outgoing pleats for $\sigma =\Lt ,\Rt$.
In our example, we are taking $\delta_\Lt >0$ and $\delta_\Rt=0$ for instruction.
\begin{figure}[htbp]
  \begin{center}
\addtocounter{theorem}{1}
          \includegraphics[width=0.75\hsize]{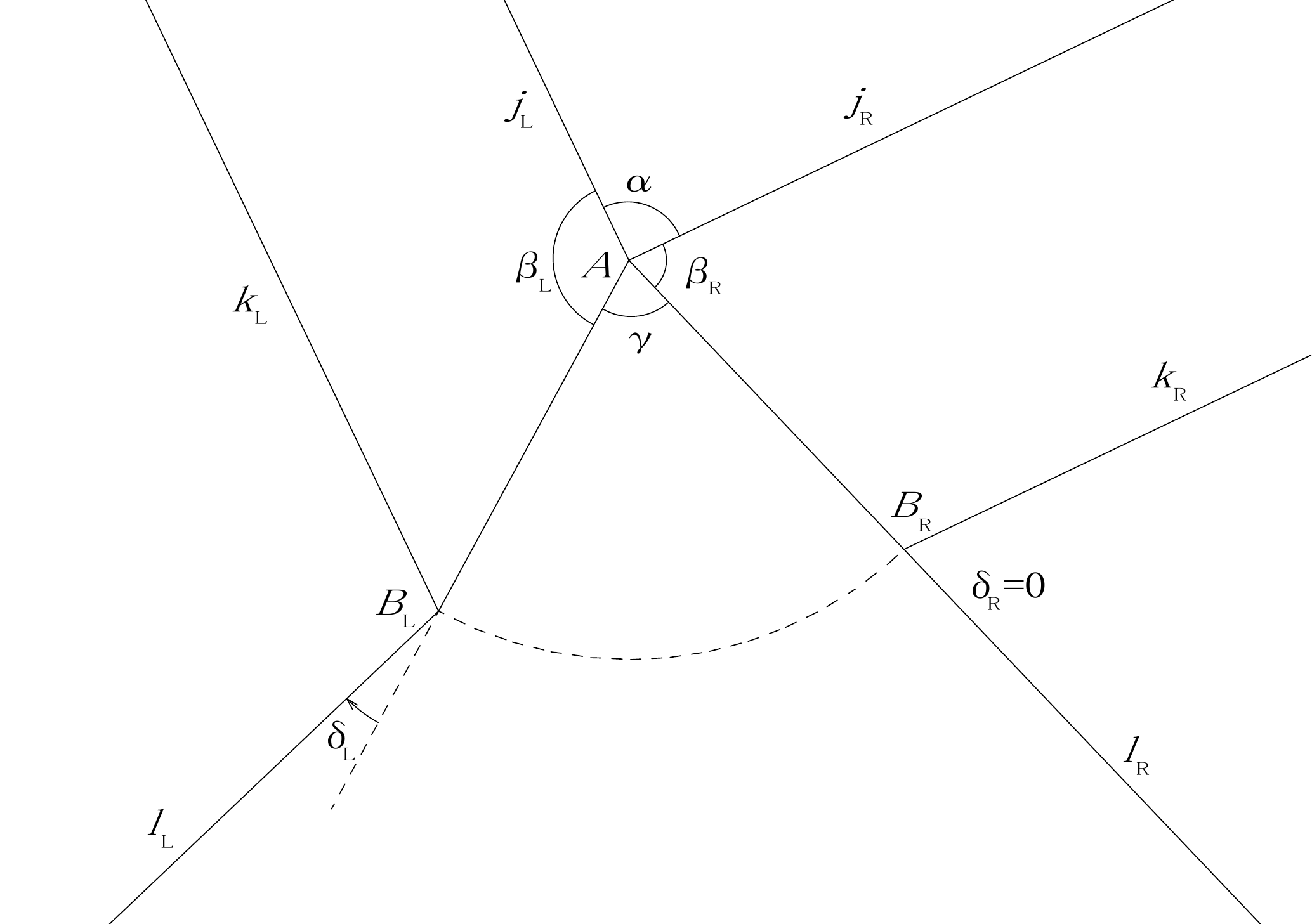}
    \caption{Development to which we apply a new gadget}
    \label{fig:development_new}
\end{center}
\end{figure}
\begin{construction}\label{const:new}\rm
Consider a development as shown in Figure $\ref{fig:development_new}$, for which we require the following conditions.
\begin{enumerate}[(i)]
\item $\alpha <\beta_\Lt+ \beta_\Rt$, $\beta_\Lt <\alpha+ \beta_\Rt$ and $\beta_\Rt <\alpha+ \beta_\Lt$.
\item $\alpha +\beta_\Lt +\beta_\Rt <2\pi$.
\item[(iii.a)] $\delta_\Lt ,\delta_\Rt\geqslant 0$, where we take clockwise (resp. counterclockwise) direction as positive for $\sigma =\Lt$ (resp. $\sigma=\Rt$).
\item[(iii.b)] $\delta_\sigma <\pi /2$ for $\sigma =\Lt ,\Rt$. 
\item[(iii.c)] $\alpha +\beta_\Lt +\beta_\Rt -\delta_\Lt -\delta_\Rt >\pi$, or equivalently, $\gamma +\delta_\Lt +\delta_\Rt <\pi$.
\setcounter{enumi}{3}
\item $\beta_\sigma+\gamma_\sigma /2\geqslant\pi /2$ for $\sigma =\Lt ,\Rt$, 
where we exclude the simultaneous equalities $\beta_\Lt +\gamma_\Lt /2 =\beta_\Rt +\gamma_\Rt /2=\pi /2$, and $\gamma_\Lt ,\gamma_\Rt$ are determined by
\begin{equation}\label{tan_gamma_LR}
\begin{aligned}
\tan\gamma_\Lt &=
\frac{1-\cos\gamma +\sin\gamma\tan\delta_\Rt}{\sin\gamma +\cos\gamma\tan\delta_\Rt +\tan\delta_\Lt},\\
\tan\gamma_\Rt &=
\frac{1-\cos\gamma +\sin\gamma\tan\delta_\Lt}{\sin\gamma +\cos\gamma\tan\delta_\Lt +\tan\delta_\Rt}.
\end{aligned}
\end{equation}
\end{enumerate}
Validity of these conditions are discussed in Section $\ref{sec:4}$. 
The first inequality of condition $\mathrm{(i)}$ is derived from condition $\mathrm{(iv)}$ by Lemma $\ref{lem:condition_gamma}$, and thus abundant.
However, we will leave condition $\mathrm{(i)}$ as it is to compare the conditions needed in Constructions $\ref{const:conv}$ and $\ref{const:new}$.
The crease pattern of our new $3$D gadget is constructed as follows, where all procedures are done for both $\sigma =\Lt ,\Rt$.
\begin{enumerate}
\item Draw a perpendicular to $\ell_\sigma$ through $B_\sigma$ for both $\sigma =\Lt ,\Rt$, letting $C$ be the intersection point. 
\item Let $D$ be the intersection point of segment $AC$ and the circular arc $B_\Lt B_\Rt$ with center $A$.
\item Let $E_\sigma$ be the circumcenter of $\triangle B_\sigma CD$. 
Also, let $m_\sigma$ be a ray parallel to and going in the same direction as $\ell_\sigma$, starting from $E_\sigma$. 
Thus $m_\sigma$ is a perpendicular bisector to segment $B_\sigma C$ and $AE_\sigma$ bisects $\angle B_\sigma AC$.
\item Let $F$ be the intersection point of segments $AC$ and $E_\Lt E_\Rt$.
\item Determine a point $G_\sigma$ on segment $A E_\sigma$ so that $\angle A B_\sigma G_\sigma=\pi -\beta_\sigma$.
\item If $\delta_\sigma>0$, then determine a point $H_\sigma$ on segment $A E_\sigma$ so that $\angle E_\sigma B_\sigma H_\sigma =\delta_\sigma$.
\item The crease pattern is shown as the solid lines in Figure $\ref{fig:gadget_new}$, 
and the assignment of mountain folds and valley folds is given in Table $\ref{tbl:assignment_new}$, 
where we have two ways of assigning mountain folds and valley folds if both $\delta_\sigma >0$ and $\beta_\sigma +\gamma_\sigma >\pi /2$ hold.
\end{enumerate}
\end{construction}
The location of point $H_\sigma$ in Construction $\ref{const:new}$, $(7)$ is derived from the foldability condition around point $B_\sigma$,
which we will check in the next section.

\begin{figure}[htbp]
  \begin{center}
\addtocounter{theorem}{1}
          \includegraphics[width=0.75\hsize]{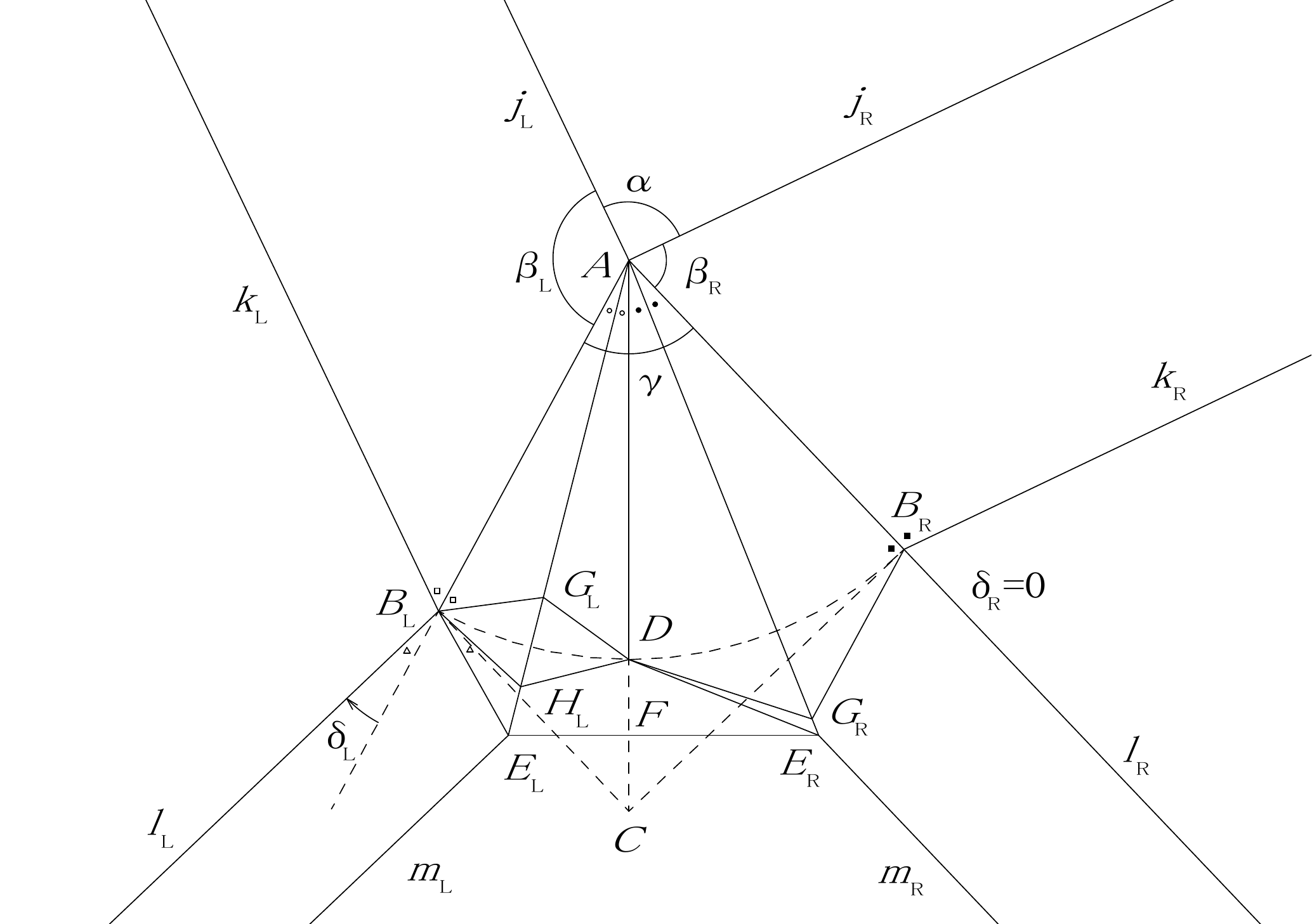}
    \caption{CP of the new gadget}
    \label{fig:gadget_new}
\end{center}
\end{figure}
\addtocounter{theorem}{1}
\begin{table}[h]
\begin{tabular}{c|c|c|c:c|c|c}
&common&\multicolumn{3}{c|}{$\beta_\sigma +\gamma_\sigma /2>\pi /2$ and}&\multicolumn{2}{c}{$\beta_\sigma +\gamma_\sigma /2=\pi /2$ and}\\
\cline{3-7}
&creases&$\delta_\sigma=0$&\multicolumn{2}{c|}{$\delta_\sigma>0$}&$\delta_\sigma=0$&$\delta_\sigma>0$\\
\hline
mountain&$j_\sigma ,\ell_\sigma ,$&$B_\sigma G_\sigma ,DE_\sigma$&\multicolumn{2}{c|}{$DH_\sigma$}
&$B_\sigma E_\sigma =B_\sigma G_\sigma$&$B_\sigma E_\sigma ,$\\ \cline{4-5}
folds&$AB_\sigma ,AD$&&$B_\sigma E_\sigma ,$&$B_\sigma H_\sigma ,$&&$DG_\sigma =DH_\sigma$\\
&&&$B_\sigma G_\sigma$&$E_\sigma H_\sigma$&&\\
\hline
valley&$k_\sigma ,m_\sigma ,$&$AE_\sigma ,DG_\sigma$&\multicolumn{2}{c|}{$DG_\sigma$}&$AE_\sigma ,$&$AE_\sigma ,$\\ \cline{4-5}
folds&$E_\Lt E_\Rt$&&$AE_\sigma ,$&$AH_\sigma ,$&$DE_\sigma =DG_\sigma$&$B_\sigma G_\sigma =B_\sigma H_\sigma$\\
&&&$B_\sigma H_\sigma$&$B_\sigma G_\sigma$&&
\end{tabular}\vspace{0.5cm}
\caption{Assignment of mountain folds and valley folds to the new gadget}
\label{tbl:assignment_new}
\end{table}
\section{Constructibility and foldability of the crease pattern}\label{sec:4}
In this section we discuss the validity of conditions $\mathrm{(i)}$--$\mathrm{(iv)}$ given in Construction $\ref{const:new}$ 
to ensure that the construction is possible.
Also, we check the foldability of the resulting crease pattern. 

Conditions $\mathrm{(i)}$ and $\mathrm{(ii)}$ are the same as in Construction $\ref{const:conv}$.
For condition $\mathrm{(iii.a)}$, we suppose the converse, i.e., $\delta_\sigma<0$.
Then $\angle k_\sigma B_\sigma \ell_\sigma=\pi -\beta_\sigma -\delta_\sigma$ is greater than 
$\angle \ell_\sigma B_\sigma G_\sigma =\pi -\beta_\sigma +\delta_\sigma$. 
Since no crease can be added inside $\angle k_\sigma B_\sigma \ell_\sigma$,
we cannot make the flat-foldablility hold around $B_\sigma$ however we add creases inside $\angle \ell_\sigma B_\sigma G_\sigma$.
Condition $\mathrm{(iii.b)}$ is necessary for $\tan\gamma_\sigma >0$ to be well-defined.
Thanks to condition $\mathrm{(iii.c)}$, we can find a point $C$ inside $\angle B_\Lt A B_\Rt =\gamma$ in Construction $\ref{const:new}$.
If $\delta_\Lt =\delta_\Rt =0$, then conditions $\mathrm{(iii.a)}$--$\mathrm{(iii.c)}$ together 
are equivalent to condition $\mathrm{(iii)}$ of Construction $\ref{const:conv}$. 
To consider condition $\mathrm{(iv)}$, we derive $\eqref{tan_gamma_LR}$.
\begin{lemma}\label{lem:tan_gamma_sigma}
The angles $\gamma_\Lt$ and $\gamma_\Rt$ dividing $\gamma$ are determined by
\begin{equation*}
\begin{aligned}
\tan\gamma_\Lt &=
\frac{1-\cos\gamma +\sin\gamma\tan\delta_\Rt}{\sin\gamma +\cos\gamma\tan\delta_\Rt +\tan\delta_\Lt},\\
\tan\gamma_\Rt &=
\frac{1-\cos\gamma +\sin\gamma\tan\delta_\Lt}{\sin\gamma +\cos\gamma\tan\delta_\Lt +\tan\delta_\Rt}.
\end{aligned}
\end{equation*}
In particular, we have 
\begin{equation*}
\gamma_\Lt =\gamma_\Rt=\frac{\gamma}{2}\quad\text{if }\delta_\Lt =\delta_\Rt,
\end{equation*}
which incudes the case of $\delta_\Lt =\delta_\Rt=0$.
\end{lemma}
\begin{proof}
We may assume $\norm{AB}=1$ without loss of generality.
For $\sigma =\Lt ,\Rt$, let $\ora{u_\sigma}=\ora{A B_\sigma}$, and
let $\ora{n_\sigma}$ be the unit direction vector of $\ora{B_\sigma C}$, which is normal to ray $\ell_\sigma$.
Also, let $\ora{v}=\ora{AC}$. 
Then $\ora{v}$ is written as
\begin{equation}\label{v_pq}
\ora{v}=\ora{u_\Lt}+p_\Lt\ora{n_\Lt}=\ora{u_\Rt}+p_\Rt\ora{n_\Rt}
\end{equation}
for some $p_\Lt ,p_\Rt >0$.
Taking $\ora{u_\Lt}$ be the unit direction vector of the $x$-axis, we have
\begin{equation}\label{component_un}
\begin{aligned}
\ora{u_\Lt}&=(1,0),&\ora{n_\Lt}&=(\sin\delta_\Lt ,\cos\delta_\Lt ),\\
\ora{u_\Rt}&=(\cos\gamma ,\sin\gamma ),&\ora{n_\Rt}&=(\sin (\gamma +\delta_\Rt ),-\cos (\gamma +\delta_\Rt )).
\end{aligned}
\end{equation}
Setting $\ora{v}=(x,y)$ and substituting $\eqref{component_un}$ into $\eqref{v_pq}$ gives that
\begin{equation*}
(p_\Lt ,p_\Rt )M=(\cos\gamma -1,\sin\gamma ),\quad\text{where }M=
\xymatrix{
{\begin{pmatrix}\sin\delta_\Lt&\cos\delta_\Lt\\
-\sin (\gamma +\delta_\Rt )&\cos (\gamma +\delta_\Rt )\end{pmatrix}}
}.
\end{equation*}
Thus we have
\begin{equation}\label{p_sigma}
\begin{aligned}
(p_\Lt ,p_\Rt )&=(\cos\gamma -1,\sin\gamma )M^{-1}\\
&=\frac{1}{\sin (\gamma +\delta_\Lt +\delta_\Rt )}
(\cos\delta_\Rt -\cos (\gamma +\delta_\Rt ),\cos\delta_\Lt -\cos (\gamma +\delta_\Lt )),\quad\text{where}\\
M^{-1}&=\frac{1}{\sin (\gamma +\delta_\Lt +\delta_\Rt )}
\xymatrix{
{\begin{pmatrix}\cos (\gamma +\delta_\Rt )&-\cos\delta_\Lt\\
\sin (\gamma +\delta_\Rt )&\sin\delta_\Lt\end{pmatrix}}
}.
\end{aligned}
\end{equation}
Substituting $\eqref{component_un}$ and $\eqref{p_sigma}$ into $\eqref{v_pq}$ gives that
\begin{align*}
x&=1+\frac{\sin\delta_\Lt}{\sin (\gamma +\delta_\Lt +\delta_\Rt )}(\cos\delta_\Rt -\cos (\gamma +\delta_\Rt )),\\
y&=\frac{\cos\delta_\Lt}{\sin (\gamma +\delta_\Lt +\delta_\Rt )}(\cos\delta_\Rt -\cos (\gamma +\delta_\Rt )),
\end{align*}
from which we deduce that
\begin{equation}\label{expr:tan_gamma_L}
\begin{aligned}
\tan\gamma_\Lt &=\frac{y}{x}=
\frac{\cos\delta_\Lt (\cos\delta_\Rt -\cos (\gamma +\delta_\Rt ))}
{\sin (\gamma +\delta_\Lt +\delta_\Rt )+\sin\delta_\Lt (\cos\delta_\Rt -\cos (\gamma +\delta_\Rt ))}\\
&=\frac{\cos\delta_\Lt (\cos\delta_\Rt -\cos (\gamma +\delta_\Rt ))}
{\cos\delta_\Lt\sin (\gamma +\delta_\Rt ) +\sin\delta_\Lt\cos\delta_\Rt}\\
&=\frac{\cos\delta_\Lt (\cos\delta_\Rt -\cos\gamma\cos\delta_\Rt +\sin\gamma\sin\delta_\Rt )}
{\sin\gamma\cos\delta_\Lt\cos\delta_\Rt +\cos\gamma\cos\delta_\Lt\sin\delta_\Rt +\sin\delta_\Lt\cos\delta_\Rt}\\
&=\frac{1-\cos\gamma +\sin\gamma\tan\delta_\Rt}{\sin\gamma +\cos\gamma\tan\delta_\Rt +\tan\delta_\Lt}.
\end{aligned}
\end{equation}
Also, interchanging $\delta_\Lt$ and $\delta_\Rt$ in $\eqref{expr:tan_gamma_L}$ gives an expression for $\tan\gamma_\Rt$.
The latter part of the lemma is trivial. 
\end{proof}
\begin{lemma}
In Construction $\ref{const:new}$, we have 
\begin{equation*}
\gamma_\sigma <\frac{\pi}{2},\quad\beta_\sigma >\delta_\sigma\quad\text{for }\sigma =\Lt ,\Rt .
\end{equation*}
Thus $\ell_\sigma$ never overlaps with $k_\sigma$.
\end{lemma}
\begin{proof}
We may assume $\sigma =\Rt$, $B_\Rt =(0,0), C=(0,-1)$ and $\ell_\Rt =\set{y=0,x\leqslant 0}$ without loss of generality.
Then $A$ must lie in the range $\set{x<0,y\geqslant 0}$, which gives that $0<\gamma_\Rt =\angle B_\Rt AC<\pi /2$.
Let $A=(-a,b)$ with $a\geqslant 0,b>0$ and let $P=(0,-p)$ with $0<p<1$ be a point on segment $B_\Rt C$ such that $AP$ bisects $\angle B_\Rt AC$. 
If we set $\phi_\Rt =\beta_\Rt +\gamma_\Rt /2 -\pi /2$, then $\phi_\Rt\geqslant 0$ by condition $\mathrm{(iv)}$ 
and the direction vector of $k_\Rt$ is given by rotating $\ora{AP}=(a,-p-b)$ by angle $\pi /2 +\phi_\Rt$, and so by rotating $(p+b,a)$ by angle $\phi_\Rt$. 
Hence $\ell_\Rt$ never overlaps with $k_\Rt$ and we have $\beta_\Rt -\delta_\Rt =\tan^{-1}((p+b)/a)+\phi_\Rt >0$.
\end{proof}
Note that we have $\beta_\sigma -\delta_\sigma\to +0$ as $b,\phi_\sigma\to +0$ in the above proof.
We may also include $\delta_\sigma <\beta_\sigma$ in condition $\mathrm{(iii.b)}$ from the beginning.
\begin{proposition}
For $\sigma =\Lt ,\Rt$ we have
\begin{equation*}
\angle C B_\sigma E_\sigma =\frac{\gamma_\sigma}{2},
\end{equation*}
and thus
\begin{equation*}
\angle A B_\sigma E_\sigma =\frac{\pi}{2}+\frac{\gamma_\sigma}{2}+\delta_\sigma .
\end{equation*}
\end{proposition}
\begin{proof}
Let $\theta_\sigma =\angle C B_\sigma E_\sigma$ and $\eta_\sigma =\angle A E_\sigma F$. 
Since $\triangle B_\sigma D E_\sigma$ and $\triangle C D E_\sigma$ are equilateral, so is $\triangle B_\sigma C E_\sigma$, 
and it follows that $\angle B_\sigma E_\sigma C=2\eta_\sigma$ and $2\theta_\sigma +2\eta_\sigma =\pi$. 
Also, we see from $\angle A F E_\sigma=\pi /2$ that $\gamma_\sigma /2 +\eta_\sigma =\pi /2$. 
Consequently we have $\theta_\sigma =\gamma_\sigma /2$, which completes the proof.
\end{proof}
In our construction, triangles $\triangle A B_\Lt G_\Lt$ and $\triangle A B_\Rt G_\Rt$ overlap on the left and the right face respectively,
and thus in the crease pattern no crease is allowed to pass across these triangles. 
If $\angle A B_\sigma G_\sigma > \angle A B_\sigma E_\sigma$, which is equivalent to 
\begin{equation}\label{cross_1}
\pi -\beta_\sigma >\frac{\pi}{2}+\frac{\gamma_\sigma}{2}+\delta_\sigma ,
\end{equation}
then $m_\sigma$ passes across $\triangle A B_\sigma G_\sigma$.
Also, if $\delta_\sigma>0$ and $\angle A B_\sigma G_\sigma>\angle A B_\sigma H_\sigma$, which is equivalent to
\begin{equation}\label{cross_2}
\pi -\beta_\sigma >\frac{\pi}{2}+\frac{\gamma_\sigma}{2},
\end{equation}
then segment $B_\sigma H_\sigma$ passes across $\triangle A B_\sigma G_\sigma$.
Hence we need the negation of $\eqref{cross_1}$ and $\eqref{cross_2}$, that is,
\begin{equation}\label{cond_iv}
\beta_\sigma +\frac{\gamma_\sigma}{2}\geqslant\frac{\pi}{2}.
\end{equation}
Conversely, $\eqref{cond_iv}$ ensures that point $G_\sigma$ lies on segment $A E_\sigma$ in Construction $\ref{const:new}$, $(5)$ 
and that $H_\sigma$ lies on segment $G_\sigma E_\sigma$ for $\delta_\sigma>0$ in Construction $\ref{const:new}$, $(6)$ for $\sigma =\Lt ,\Rt$.
This is why we need condition $\mathrm{(iv)}$ of Construction $\ref{const:new}$.
Note in particular that $\eqref{cond_iv}$ becomes
\begin{equation*}
\beta_\sigma +\frac{\gamma}{4}\geqslant\frac{\pi}{2}\quad\text{if }\delta_\Lt =\delta_\Rt ,
\end{equation*}
which is independent of $\delta_\sigma$.
\begin{proposition}\label{prop:flat_extrusion}
If $\beta_\sigma +\gamma_\sigma /2=\pi /2$ 
for both $\sigma=\Lt ,\Rt$, then $\alpha =\beta_\Lt +\beta_\Rt$, i.e., the resulting extrusion is flat (of height zero),
which does not satisfy condition $\mathrm{(i)}$ of Construction $\ref{const:new}$.
\end{proposition}
\begin{proof}
Define $\alpha_\sigma$ for $\sigma =\Lt ,\Rt$ by $\alpha_\sigma =\pi -\beta_\sigma -\gamma_\sigma$, so that we have $\alpha =\alpha_\Lt +\alpha_\Rt$.
It is sufficient to prove that if $\beta_\sigma +\gamma_\sigma /2=\pi /2$, 
then $\alpha_\sigma =\beta_\sigma$. 
But from the definition of $\alpha_\sigma$, we easily have
\begin{equation*}
\alpha_\sigma=\pi -\beta_\sigma -\gamma_\sigma=\pi -\beta_\sigma -(\pi -2\beta_\sigma)=\beta_\sigma ,
\end{equation*}
which completes the proof.
\end{proof}
Thus for a non-flat extrusion a special case $E_\sigma=G_\sigma$ for $\delta_\sigma=0$ or $G_\sigma=H_\sigma$ for $\delta_\sigma>0$ 
can occur only on either the left or the right side. 
This is why we exclude the case where $\beta_\sigma +\gamma_\sigma /2=\pi /2$ for both $\sigma =\Lt ,\Rt$ 
in condition $\mathrm{(iv)}$ of Construction $\ref{const:new}$.

Now we shall check the foldability of the crease pattern given in Construction $\ref{const:new}$.
For this purpose, we divide the crease pattern into two parts by polygonal chain $k_\Lt B_\Lt G_\Lt D G_\Rt B_\Rt k_\Rt$.
Then we can see easily that the upper part forms the top and the side faces with flat back sides, where
$\triangle A B_\sigma G_\sigma$ and $\triangle ADG_\sigma$ overlap with the side face $j_\sigma AB_\sigma k_\sigma$ for $\sigma =\Lt ,\Rt$.
Thus it remains to check that the lower part is flat-foldable to form the base of the extrusion. 
We show in Table $\ref{tbl:adjacent_vertex_ray}$ the adjacent vertices of or the rays starting from each vertex in the lower part,
where $\sigma'$ denotes the other side of $\sigma$, and for vertex $D$, we take into account both contributions from $\sigma =\Lt ,\Rt$.
Since vertices $E_\sigma ,H_\sigma$ are interior points of the lower part, we can check the flat-foldability around them
by Kawasaki's theorem \cite{Kawasaki} (Murata-Kawasaki's theorem may be more precise) that the crease pattern is (locally) flat-foldable 
if and only if the alternative sum of the angles around each vertex vanishes.
Note that $k_\sigma$ and $B_\sigma G_\sigma$ around boundary point $B_\sigma$, 
and $G_\sigma B_\sigma$ and $G_\sigma D$ around boundary point $G_\sigma$ overlap with each other, 
so that around $B_\sigma$ and $G_\sigma$, the alternative sums of the angles contained in the lower part must also vanish .
For boundary point $D$, $DG_\Lt$ and $DG_\Rt$, which overlap with $k_\Lt$ and $k_\Rt$ respectively, form an angle $\alpha$.
Thus around $D$, the alternative sum of the angles contained in the lower part must be $\alpha$ if we take the angle which contains $E_\Lt E_\Rt$ as positive.
In view of the above, the flat-foldability around each vertex in the lower part is checked as follows.
\addtocounter{theorem}{1}
\begin{table}[h]
\begin{tabular}{c|c|c|c|c}
&\multicolumn{2}{c|}{$\delta_\sigma=0$ and $\beta_\sigma +\gamma_\sigma /2$}
&\multicolumn{2}{c}{$\delta_\sigma>0$ and $\beta_\sigma +\gamma_\sigma /2$}\\
\cline{2-5}
&$>\pi/2$&$=\pi/2$
&$>\pi/2$&$=\pi/2$\\
\hline
$B_\sigma$&\multicolumn{2}{c|}{$k_\sigma,\ell_\sigma,G_\sigma$}&$k_\sigma,\ell_\sigma,E_\sigma,H_\sigma,G_\sigma$
&$k_\sigma,\ell_\sigma,E_\sigma,G_\sigma=H_\sigma$\\ \hline
$D$&$E_\sigma,G_\sigma$&$E_\sigma=G_\sigma$&$G_\sigma,H_\sigma$&$G_\sigma=H_\sigma$\\ \hline
$E_\sigma$&$m_\sigma,D,E_{\sigma'},G_\sigma$&\multirow{2}{*}{$m_\sigma,B_\sigma,D,E_{\sigma'}$}&$m_\sigma,B_\sigma,E_{\sigma'},H_\sigma$
&$m_\sigma,B_\sigma,E_{\sigma'},G_\sigma=H_\sigma$\\ 
\cline{1-2}\cline{4-5}
$G_\sigma$&$B_\sigma,D,E_\sigma$&&$B_\sigma,D,H_\sigma$&\multirow{2}{*}{$B_\sigma,D,E_\sigma$}\\
\cline{1-4}
$H_\sigma$&\multicolumn{2}{c|}{---}&$B_\sigma,D,E_\sigma,G_\sigma$&
\end{tabular}
\caption{Adjacent vertices of, or rays starting from each vertex in the lower part}
\label{tbl:adjacent_vertex_ray}
\end{table}\\
$\bullet$ {\it Flat-foldability around $B_\sigma$.}
If $\delta_\sigma =0$, then the alternative sum of the angles around $B_\sigma$ in the lower part is calculated as
\begin{equation*}
\angle k_\sigma B_\sigma \ell_\sigma -\angle \ell_\sigma B_\sigma G_\sigma=\beta_\sigma -\beta_\sigma=0.
\end{equation*}
If $\delta_\sigma>0$, then we have
\begin{align*}
\angle k_\sigma B_\sigma \ell_\sigma&=\beta_\sigma -\delta_\sigma,\\
\angle \ell_\sigma B_\sigma E_\sigma&=\angle \ell_\sigma B_\sigma C -\angle E_\sigma B_\sigma C =\frac{\pi}{2}-\frac{\gamma_\sigma}{2},\\
\angle E_\sigma B_\sigma H_\sigma&=\delta_\sigma,\\
\angle H_\sigma B_\sigma G_\sigma&=\angle A B_\sigma H_\sigma -\angle A B_\sigma G_\sigma -\angle E_\sigma B_\sigma H\\
&=\left(\frac{\pi}{2}+\frac{\gamma_\sigma}{2}+\delta_\sigma\right)-(\pi -\beta_\sigma)-\delta_\sigma =\beta_\sigma+\frac{\gamma_\sigma}{2}-\frac{\pi}{2}.
\end{align*}
Thus the alternative sum is given by
\begin{equation*}
\begin{cases}
\angle k_\sigma B_\sigma \ell_\sigma -\angle \ell_\sigma B_\sigma E_\sigma +\angle E_\sigma B_\sigma H_\sigma -\angle H_\sigma B_\sigma G_\sigma =0
&\text{if }\beta_\sigma+\gamma_\sigma /2>\pi /2,\\
\angle k_\sigma B_\sigma \ell_\sigma -\angle \ell_\sigma B_\sigma E_\sigma +\angle E_\sigma B_\sigma H_\sigma =0
&\text{if }\beta_\sigma+\gamma_\sigma /2=\pi /2.
\end{cases}
\end{equation*}
$\bullet$ {\it Flat-foldability around $D$.}
As in the proof of Proposition $\ref{prop:flat_extrusion}$, we define $\alpha_\sigma =\pi -\beta_\sigma -\gamma_\sigma$, so that $\alpha_\Lt +\alpha_\Rt =\alpha$.
We divide $\angle E_\Rt D E_\Lt$ as
\begin{equation*}
\angle E_\Rt D E_\Lt =\angle F D E_\Lt +\angle F D E_\Rt ,
\end{equation*}
and consider the contributions to the alternative sum of angles around $D$ from both sides separately.
If $\delta_\sigma=0$, then we have
\begin{align*}
\angle F D E_\sigma&=\angle D A E_\sigma +\angle A E_\sigma D =\angle B_\sigma A E_\sigma +\angle A E_\sigma B_\sigma\\
&=\pi -\angle A B_\sigma E_\sigma =\pi -\left(\frac{\pi}{2}+\frac{\gamma_\sigma}{2}\right)
=\frac{\pi}{2}-\frac{\gamma_\sigma}{2},\\
\angle E_\sigma D G_\sigma&=\angle A D E_\sigma -\angle A D G_\sigma =\angle A B_\sigma E_\sigma -\angle A B_\sigma G_\sigma\\
&=\left(\frac{\pi}{2}+\frac{\gamma_\sigma}{2}\right) -(\pi -\beta_\sigma )
=\beta_\sigma+\frac{\gamma_\sigma}{2}-\frac{\pi}{2},
\end{align*}
and thus
\begin{equation*}
\begin{cases}
\angle F D E_\sigma -\angle E_\sigma D G_\sigma =\pi -\beta_\sigma -\gamma_\sigma =\alpha_\sigma&\text{if }\beta_\sigma+\gamma_\sigma /2>\pi /2,\\
\angle F D E_\sigma =\alpha_\sigma&\text{if }\beta_\sigma+\gamma_\sigma /2=\pi /2.
\end{cases}
\end{equation*}
If $\delta_\sigma>0$, then we have
\begin{align*}
\angle F D H_\sigma&=\angle F D E_\sigma +\angle E_\sigma D H_\sigma =
(\pi -\angle A B_\sigma E_\sigma) +\angle E_\sigma B_\sigma H_\sigma \\
&=\left\{\pi -\left(\frac{\pi}{2}+\frac{\gamma_\sigma}{2}+\delta_\sigma\right)\right\} +\delta_\sigma
=\frac{\pi}{2}-\frac{\gamma_\sigma}{2},\\
\angle H_\sigma D G_\sigma&=\angle H_\sigma B_\sigma G_\sigma =\angle A B_\sigma E_\sigma-\angle A B_\sigma G_\sigma -\angle E_\sigma B_\sigma H_\sigma\\
&=\left(\frac{\pi}{2}+\frac{\gamma_\sigma}{2}+\delta_\sigma\right) -(\pi -\beta_\sigma) -\delta_\sigma
=\beta_\sigma +\frac{\gamma_\sigma}{2}-\frac{\pi}{2},
\end{align*}
and thus
\begin{equation*}
\begin{cases}
\angle F D H_\sigma -\angle H_\sigma D G_\sigma =\pi -\beta_\sigma -\gamma_\sigma =\alpha_\sigma&\text{if }\beta_\sigma+\gamma_\sigma /2>\pi /2,\\
\angle F D H_\sigma =\alpha_\sigma&\text{if }\beta_\sigma+\gamma_\sigma /2=\pi /2.
\end{cases}
\end{equation*}
Consequently, in all cases 
the alternative sum is given by
\begin{equation*}
\alpha_\Lt +\alpha_\Rt =\alpha ,
\end{equation*}
which coincide with the angle formed by $k_\Lt$ and $k_\Rt$ as desired. \\
$\bullet$ {\it Flat-foldability around $E_\sigma$.} 
Suppose $\delta_\sigma=0$. 
Then we have
\begin{align*}
\angle m_\sigma E_\sigma G_\sigma&=\angle m_\sigma E_\sigma B_\sigma +\angle B_\sigma E_\sigma G_\sigma\\
\angle m_\sigma E_\sigma F&=\angle m_\sigma E_\sigma C +\angle C E_\sigma F =\angle m_\sigma E_\sigma B_\sigma+\angle C E_\sigma F\\
\angle F E_\sigma D&=\angle C E_\sigma F\\
\angle D E_\sigma G_\sigma&=\angle B_\sigma E_\sigma G_\sigma,
\end{align*}
which gives that 
\begin{equation*}
\begin{cases}
\angle m_\sigma E_\sigma G_\sigma-\angle m_\sigma E_\sigma F+\angle F E_\sigma D-\angle D E_\sigma G_\sigma=0&\text{if }\beta_\sigma +\gamma_\sigma /2>\pi /2, \\
\angle m_\sigma E_\sigma B_\sigma-\angle m_\sigma E_\sigma F+\angle F E_\sigma D=0&\text{if }\beta_\sigma +\gamma_\sigma /2=\pi /2. 
\end{cases}
\end{equation*}
Next suppose $\delta_\sigma>0$. 
Then we have
\begin{align*}
\angle m_\sigma E_\sigma F&=\angle m_\sigma E_\sigma C +\angle C E_\sigma F=\angle B_\sigma E_\sigma m_\sigma +\angle D E_\sigma F,\\
\angle F E_\sigma G_\sigma&= \angle D E_\sigma F +\angle G_\sigma E_\sigma D =\angle D E_\sigma F +\angle G_\sigma E_\sigma B_\sigma,
\end{align*}
which gives that
\begin{equation*}
\begin{cases}
\angle B_\sigma E_\sigma m_\sigma -\angle m_\sigma E_\sigma F +\angle F E_\sigma G_\sigma -\angle G_\sigma E_\sigma B_\sigma =0
&\text{if }\beta_\sigma +\gamma_\sigma /2>\pi /2, \\
\angle B_\sigma E_\sigma m_\sigma -\angle m_\sigma E_\sigma F +\angle F E_\sigma G_\sigma =0
&\text{if }\beta_\sigma +\gamma_\sigma /2=\pi /2.
\end{cases}
\end{equation*}
$\bullet$ {\it Flat-foldability around $G_\sigma$ and $H_\sigma$.} 
This is clear from the symmetry of $B_\sigma$ and $D$ with respect to $A E_\sigma$.
\section{Interference coefficients and maximization of heights}\label{sec:5}
In this section we consider a polyhedron as described in the end of Introduction,
for which the bottom face $B_1\dots B_N$ and parameters $\alpha_i,\beta_{i,\sigma}, \delta_{i,\sigma}$ (and so $\gamma_i,\alpha_{i,\sigma},\gamma_{i,\sigma}$)
are all fixed for $i=1,\dots ,N$ and $\sigma =\Lt ,\Rt$.
Also, we use the convention that $A_{i\pm N}=A_i$ and $B_{i\pm N}=B_i$.
For $i=1,\dots N$, let $\triangle T_i B_i B_{i+1}$ be a triangle with $\angle T_i B_i B_{i+1}=\pi -\beta_{i,\Rt},\angle T_i B_{i+1} B_i=\pi -\beta_{i+1,\Lt}$
if $\beta_{i,\Rt}+\beta_{i+1,\Lt}<\pi$.
Then the maximal height $h_{\max}$ of a polyhedron with these data is given by
\begin{equation*}
h_{\max}=\begin{cases}
\displaystyle\min_{\substack{i=1,\dots ,N\text{ with}\\ \beta_{i,\Rt}+\beta_{i+1,\Lt}<\pi}}
\lambda_i\norm{T_i B_i}&\text{if }\beta_{i,\Rt}+\beta_{i+1,\Lt}<\pi\text{ for some }i,\\
\infty&\text{if }\beta_{i,\Rt}+\beta_{i+1,\Lt}\geqslant\pi\text{ for all }i,
\end{cases}
\end{equation*}
where $\norm{T_i B_i}$ is given by $\norm{T_i B_i}=\norm{B_i B_{i+1}}\sin\beta_{i+1,\Lt}/\sin (\beta_{i,\Rt} +\beta_{i+1,\Lt})$ by the sine theorem, and
$\lambda_i =\lambda (\alpha_i ,\beta_{i,\Lt},\beta_{i,\Rt})$ is given by $\eqref{lambda}$.
For $0<h\leqslant h_{\max}$, we denote by $\Delta_h$ the polyhedron of height $h$ with the above data, and let $A_i=A_i(h), i=1,\dots ,N$ be the corresponding
vertices of the top face of $\Delta_h$. 
The length of ridge $A_i B_i$ is given by $\norm{A_i B_i}=h/\lambda_i$.

If conditions $\mathrm{(i)}$--$\mathrm{(iii)}$ or conditions $\mathrm{(i)}$--$\mathrm{(iv)}$ hold for each $i=1,\dots N$ and $\sigma =\Lt ,\Rt$ 
in Constructions $\ref{const:conv}$ or $\ref{const:new}$ respectively,
then a crease pattern of $\Delta_h$ is constructible and locally foldable as shown in Section $\ref{sec:4}$.
However, it may not be globally foldable because interference of adjacent gadgets may occur.
Let us see how $3$D gadgets interfere with adjacent gadgets.

For conventional gadgets, interference with an adjacent conventional gadget occurs in the way that the two triangular pyramids supporting from inside collide.
Thus in Figure $\ref{fig:gadget_conv}$ the lengths of $BD_\sigma$ of two adjacent gadgets which share the same edge are involved in the interference.

In the case of new $3$D gadgets, interference occurs in a more complicated way.
Consider the pleats formed by shadowed kites $AB_\sigma E_\sigma D$, which we call the \emph{ears} of the gadget, 
and $CI_\Lt DI_\Rt$, which we call the thrusting part of the \emph{tongue} (kite $CE_\Lt DE_\Rt$) of the gadget,
in Figure $\ref{fig:gadget_new_interference}$ may interfere with adjacent pleats, 
where $I_\sigma$ is the intersection point of $E_\Lt E_\Rt$ and a line through $C$ parallel to $j_\sigma$ for $\sigma =\Lt ,\Rt$.
Observe that the tip $G_\sigma$ of each ear swings to ${G'}_\sigma$, and $C$ swings to $B_\sigma$ with $CI_\sigma$ overlapping with $k_\sigma$.
Then interference occurs in a bottom edge in the way that the tongue of the \emph{inner} (first folded) pleat of one gadget 
collides with one of the ears of the \emph{outer} (second folded) pleat of the other gadget.
Thus in Figure $\ref{fig:gadget_new_interference}$ the length of $BG_\sigma$ of the inner gadget and that of $CI_\sigma$ of the outer gadget 
which share the same edge are involved in the interference.

To see if interference occurs or not, we introduce three kinds of \emph{interference coefficients} 
$\kappa_\conv ,$ $\kappa_\inn,$ $\kappa_\out$ for each vertex $B_{i,\sigma}$ 
and two kinds of interference coefficients $\kappa_\conv ,\kappa_\new$ for each bottom edge $B_i B_{i+1}$ with $i=1,\dots N$ and $\sigma =\Lt ,\Rt$.
Before defining the coefficients, we calculate the following lengths. 
\begin{figure}[htbp]
  \begin{center}
\addtocounter{theorem}{1}
          \includegraphics[width=0.75\hsize]{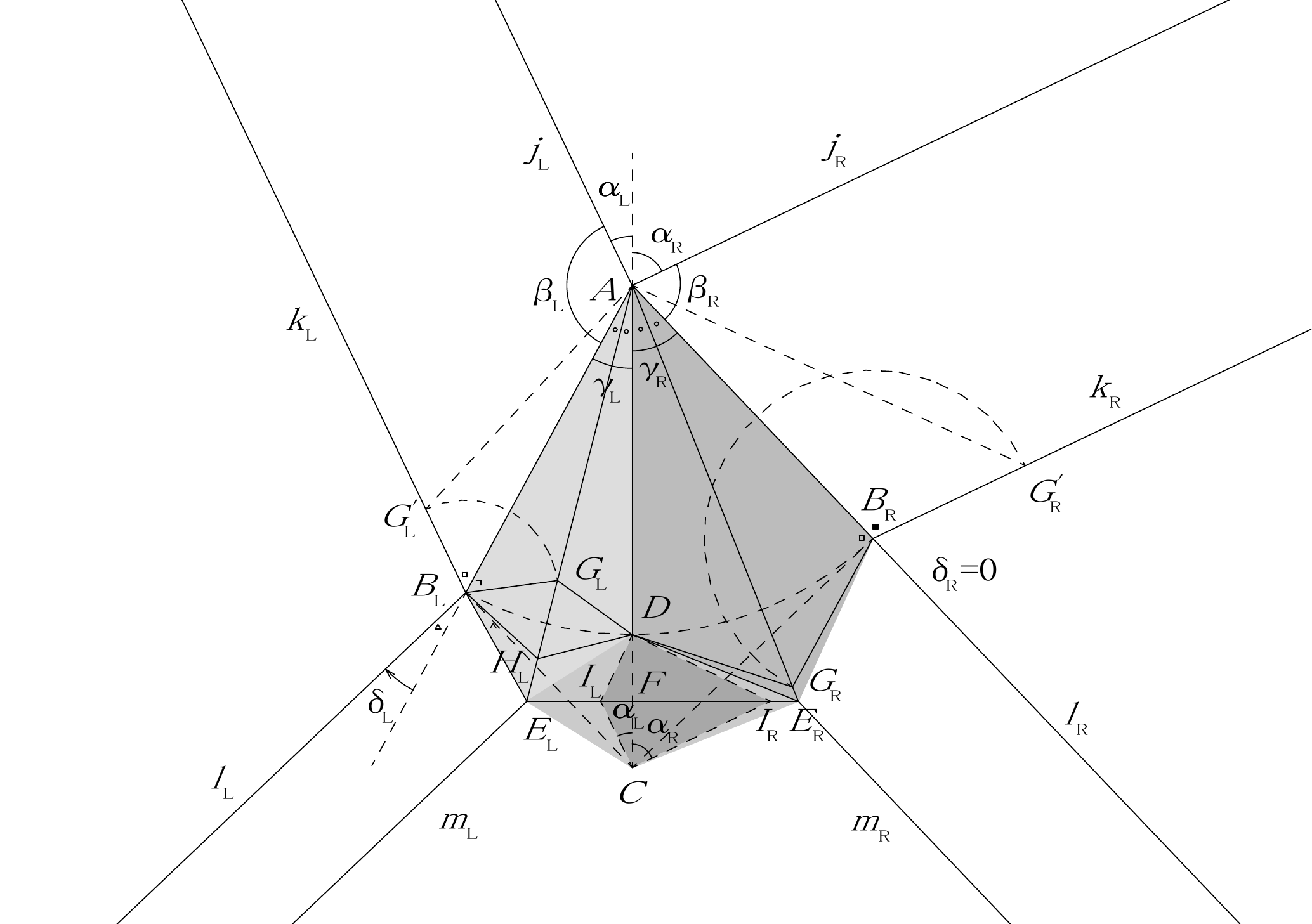}
    \caption{Kites which may cause interference with adjacent gadgets}
    \label{fig:gadget_new_interference}
\end{center}
\end{figure}
\begin{proposition}\label{prop:calc_length_new}
In the resulting crease pattern Figure $\ref{fig:gadget_new}$ in Construction $\ref{const:new}$, assume $\norm{AB}=1$. 
Then we have
\begin{equation}\label{length}
\begin{aligned}
\norm{AC}&=\frac{1}{\cos\gamma_\sigma (1-\tan\gamma_\sigma\tan\delta_\sigma )},\\
\norm{B_\sigma C}&=\frac{\cos\delta_{\sigma'}-\cos (\gamma+\delta_{\sigma'})}{\sin (\gamma +\delta_\Lt +\delta_\Rt )},\\
\norm{B_\sigma G_\sigma}&=\frac{1}{\sin\beta_\sigma /\tan (\gamma_\sigma /2)-\cos\beta_\sigma},\\
\norm{CI_\sigma}&=\frac{1}{2\cos\gamma_\sigma (\sin\beta_\sigma\tan\gamma_\sigma -\cos\beta_\sigma)}
\left(\frac{1}{\cos\gamma_\sigma (1-\tan\gamma_\sigma\tan\delta_\sigma )}-1\right) ,
\end{aligned}
\end{equation}
where $\tan\gamma_\sigma$ is given by $\eqref{tan_gamma_LR}$, and $\cos\gamma_\sigma$ and $\tan (\gamma_\sigma /2)$ are given by
\begin{equation}\label{cos_gamma_sigma}
\cos\gamma_\sigma =\frac{1}{\sqrt{1+\tan^2\gamma_\sigma}},\quad\tan\frac{\gamma_\sigma}{2}=\frac{\sqrt{\tan^2\gamma_\sigma +1}-1}{\tan\gamma_\sigma}
\end{equation}
in terms of $\tan\gamma_\sigma$. 
Thus using $\eqref{tan_gamma_LR}$ and $\eqref{cos_gamma_sigma}$, we can express the lengths in $\eqref{length}$
in terms of $\beta_\sigma ,\gamma$, $\delta_\sigma$ and not $\gamma_\sigma$.
In particular, if $\delta_\Lt =\delta_\Rt =0$, then $\gamma_\sigma =\gamma /2$,
and $\eqref{length}$ is simplified as
\begin{equation}\label{length_new_delta=0}
\begin{aligned}
\norm{AC}&=\frac{1}{\cos (\gamma /2)},\\
\norm{B_\sigma C}&=\tan\frac{\gamma}{2},\\
\norm{B_\sigma G_\sigma}&=\frac{1}{\sin\beta_\sigma /\tan (\gamma /4)-\cos\beta_\sigma},\\
\norm{C I_\sigma}&=-\frac{1-\cos (\gamma /2)}{2\cos (\beta_\sigma +\gamma /2)\cos (\gamma /2)}.
\end{aligned}
\end{equation}
Furthermore, if $\beta_\sigma =\pi /2$, then we have
\begin{equation}\label{length_new_prism}
\begin{aligned}
\norm{B_\sigma G_\sigma}&=\tan\frac{\gamma}{4},\\
\norm{C I_\sigma}&=\frac{1-\cos (\gamma /2)}{\sin\gamma}=\frac{1}{2}\left(\tan\frac{\gamma}{2}-\tan\frac{\gamma}{4}\right) .
\end{aligned}
\end{equation}
\end{proposition}
\begin{proof}
To obtain an expression of $\norm{AC}$, we use the sine theorem for $\triangle AB_\sigma C$. 
Then we have
\begin{align*}
\frac{\norm{AB}}{\sin\angle ACB_\sigma}&=\frac{\norm{AC}}{\sin\angle AB_\sigma C},\quad\text{where}\\
\sin\angle AB_\sigma C&=\sin\left(\frac{\pi}{2}+\delta_\sigma\right) =\cos\delta_\sigma ,\\
\sin\angle ACB_\sigma &=\sin\left\{\pi -\gamma_\sigma -\left(\frac{\pi}{2}+\delta_\sigma\right)\right\}
=\sin\left(\frac{\pi}{2}-\gamma_\sigma -\delta_\sigma\right)\\
&=\cos (\gamma_\sigma +\delta_\sigma )=\cos\gamma_\sigma\cos\delta_\sigma -\sin\gamma_\sigma\sin\delta_\sigma ,
\end{align*}
so that
\begin{equation}\label{AC}
\norm{AC}=\frac{\cos\delta_\sigma}{\cos\gamma_\sigma\cos\delta_\sigma -\sin\gamma_\sigma\sin\delta_\sigma}
=\frac{1}{\cos\gamma_\sigma (1-\tan\gamma_\sigma\tan\delta_\sigma )}.
\end{equation}
Recalling the proof of Lemma $\ref{lem:tan_gamma_sigma}$, we have
\begin{equation*}
\norm{BC_\sigma}=p_\sigma =\frac{\cos\delta_{\sigma'}-\cos (\gamma+\delta_{\sigma'})}{\sin (\gamma +\delta_\Lt +\delta_\Rt )}.
\end{equation*}
By the sine theorem for $\triangle AB_\sigma G_\sigma$, we have
\begin{align*}
\frac{\norm{AB}}{\sin\angle AG_\sigma B_\sigma}&=\frac{\norm{B_\sigma G_\sigma}}{\sin\angle AB_\sigma G_\sigma},\quad\text{where}\\
\sin\angle AG_\sigma B_\sigma &=\sin\left(\pi-\left(\pi -\beta_\sigma\right) -\frac{\gamma_\sigma}{2}\right)\\
&=\sin\left(\beta_\sigma -\frac{\gamma_\sigma}{2}\right)
=\sin\beta_\sigma\cos\frac{\gamma_\sigma}{2}-\cos\beta_\sigma\sin\frac{\gamma_\sigma}{2},\\
\sin\angle AB_\sigma G_\sigma &=\sin\frac{\gamma_\sigma}{2}.
\end{align*}
Thus we have
\begin{equation*}
\norm{B_\sigma G_\sigma}=\frac{\sin (\gamma_\sigma /2)}{\sin\beta_\sigma\cos (\gamma_\sigma /2)-\cos\beta_\sigma\sin (\gamma_\sigma /2)}
=\frac{1}{\sin\beta_\sigma /\tan (\gamma_\sigma /2)-\cos\beta_\sigma}.
\end{equation*}
For the length of segment $CI_\sigma$, we have
\begin{align}
\norm{CI_\sigma}&=\frac{\norm{CF}}{\cos\angle FCI_\sigma}=\frac{\norm{AC}-1}{2\cos\alpha_\sigma},\quad\text{where}\label{CI}\\
\cos\alpha_\sigma &=-\cos (\beta_\sigma +\gamma_\sigma )=\sin\beta_\sigma\sin\gamma_\sigma -\cos\beta_\sigma\cos\gamma_\sigma\label{cos_alpha_sigma}\\
&=\cos\gamma_\sigma (\sin\beta_\sigma\tan\gamma_\sigma -\cos\beta_\sigma ) .\nonumber
\end{align}
Then substituting $\eqref{cos_alpha_sigma}$ and $\eqref{AC}$ into $\eqref{CI}$ gives that
\begin{equation*}
\norm{CI_\sigma}
=\frac{1}{2\cos\gamma_\sigma (\sin\beta_\sigma\tan\gamma_\sigma -\cos\beta_\sigma)}
\left(\frac{1}{\cos\gamma_\sigma (1-\tan\gamma_\sigma\tan\delta_\sigma )}-1\right) .
\end{equation*}
The expression in $\eqref{length_new_delta=0}$ for $\delta_\Lt =\delta_\Rt =0$ follows easily. 
For the last equality in $\eqref{length_new_prism}$, we set $t=\tan (\gamma /4)$, so that $\cos (\gamma /2)=(1-t^2)/(1+t^2),\sin (\gamma /2)=2t/(1+t^2)$.
Then we have
\begin{align*}
\frac{1-\cos (\gamma /2)}{\sin\gamma}&=\frac{1-\cos (\gamma /2)}{2\sin (\gamma /2)\cos (\gamma /2)}\\
&=\frac{1}{2}\cdot\frac{1+t^2}{2t}\cdot\frac{1+t^2}{1-t^2}\left( 1-\frac{1-t^2}{1+t^2}\right) \\
&=\frac{1}{2}\left(\frac{2t}{1-t^2}-t\right) =\frac{1}{2}\left(\tan\frac{\gamma}{2}-\tan\frac{\gamma}{4}\right) .
\end{align*}
This completes the proof of Proposition $\ref{prop:calc_length_new}$.
\end{proof}
\begin{proposition}\label{prop:calc_length_conv}
In the resulting crease pattern Figure $\ref{fig:gadget_conv}$ in Construction $\ref{const:conv}$, assume $\norm{AB}=1$. 
Then we have
\begin{equation*}
\norm{B_\sigma D_\sigma}=\frac{\tan (\gamma /2)}{2\sin\beta_\sigma}.
\end{equation*}
In particular, if $\beta_\sigma =\pi /2$, the we have
\begin{equation*}
\norm{B_\sigma D_\sigma}=\frac{1}{2}\tan\frac{\gamma}{2}.
\end{equation*}
\end{proposition}
\begin{proof}
This is clear from $\norm{B_\sigma C}=2\norm{B_\sigma D_\sigma}\sin\beta_\sigma$ and 
the expression for $\norm{B_\sigma C}$ in $\eqref{length_new_delta=0}$ for $\delta_\Lt =\delta_\Rt =0$.
\end{proof}
\begin{definition}\label{def:interference_conv}\rm
Consider a development of a polyhedron $\Delta_h$ as described in the beginning of this section. 
Suppose conditions $\mathrm{(i)}$--$\mathrm{(iii)}$ of Construction $\ref{const:conv}$ hold around all ridges $A_i B_i$.
For a vertex $B_{i,\sigma}$ in the development of $\Delta_h$ with $i=1,\dots N$ and $\sigma =\Lt ,\Rt$, we define 
an \emph{interference coefficient} $\kappa_\conv (B_{i,\sigma})$ \emph{of the conventional kind} to be the ratio 
\begin{equation*}
\frac{\norm{B_{i,\sigma}D_{i,\sigma}}}{h} =\frac{\norm{B_{i,\sigma}D_{i,\sigma}}}{\lambda_i\norm{A_i B_i}}
\end{equation*}
of the length of segment $B_{i,\sigma}D_{i,\sigma}$ to the height $h$ of the extrusion, where $\lambda_i =\lambda (\alpha_i ,\beta_{i,\Lt},\beta_{i,\Rt})$.

Also, for an edge $B_i B_{i+1}$ of the bottom face (note that $B_{i+1}$ is located to the right of $B_i$), 
we define an \emph{interference coefficient} $\kappa_\conv (B_i B_{i+1})$ \emph{of the conventional kind} as
\begin{equation*}
\kappa_\conv (B_i B_{i+1})=\kappa_\conv (B_{i,\Rt})+\kappa_\conv (B_{i+1,\Lt}).
\end{equation*}
In particular, when we extrude a prism of $B_1,\dots B_N$, in which case we have $\beta_{i,\sigma}=\pi /2$ and $\lambda_ i=1$ 
for all $i=1,\dots N$ and $\sigma =\Lt ,\Rt$, we will denote $\kappa_\conv (B_{i,\sigma})$ and $\kappa_\conv (B_i B_{i+1})$
by $\kappa_\conv^\perp (B_i)$ and $\kappa_\conv^\perp (B_i B_{i+1})$ respectively.
\end{definition}

\begin{definition}\label{def:interference_new}\rm
Consider a development of a polyhedron $\Delta_h$ as described in the beginning of this section. 
Suppose conditions $\mathrm{(i)}$--$\mathrm{(iv)}$ of Construction $\ref{const:new}$ hold around all ridges $A_i B_i$.
In the resulting crease pattern in Construction $\ref{const:new}$, let $C_i,D_i,\dots G_{i,\sigma},H_{i,\sigma}$ 
be the points corresponding to ridge $A_i B_i$ for $i=1,\dots N$.
For a vertex $B_{i,\sigma}$ in the development of $\Delta_h$,
we define an \emph{interference coefficient} $\kappa_\inn (B_{i,\sigma})$ \emph{of the inner pleat} to be the ratio 
\begin{equation*}
\frac{\norm{C_i I_{i,\sigma}}}{h}=\frac{\norm{C_i I_{i,\sigma}}}{\lambda_i\norm{A_i B_i}}.
\end{equation*}
Similarly, we define an \emph{interference coefficient} $\kappa_\out (B_{i,\sigma})$ \emph{of the outer pleat} to be the ratio 
\begin{equation*}
\frac{\norm{B_{i,\sigma}G_{i,\sigma}}}{h}=\frac{\norm{B_{i,\sigma}G_{i,\sigma}}}{\lambda_i\norm{A_i B_i}}.
\end{equation*}

Now we can define an \emph{interference coefficient} $\kappa_{\new}(B_i B_{i+1})$ \emph{of the new kind} as
\begin{equation*}
\kappa_{\new}(B_i B_{i+1})=\min\{\kappa_\inn (B_{i,\Rt})+\kappa_\out (B_{i+1,\Lt}),\kappa_\out (B_{i,\Rt})+\kappa_\inn (B_{i+1,\Lt})\} .
\end{equation*}
In particular, when we extrude a prism of $B_1,\dots B_N$, we will denote $\kappa_\inn (B_{i,\sigma}),\kappa_\out (B_{i,\sigma})$ and $\kappa_\new (B_i B_{i+1})$
by $\kappa_\inn^\perp (B_i),\kappa_\out^\perp (B_i)$ and $\kappa_\new^\perp (B_i B_{i+1})$ respectively.
\end{definition}

An interference coefficient $\kappa_\conv (B_i B_{i+1})$ in Definition $\ref{def:interference_conv}$
means the total length which the supporting pyramids of the conventional gadgets 
would occupy in edge $B_i B_{i+1}$ of polyhedron $\Delta_1$, i.e., $\Delta_h$ with $h=1$.
For a polyhedron $\Delta_h$ with a general $h$ with $h\leqslant h_{\max}$, the total length in edge $B_i B_i+1$
occupied by the conventional gadgets would be $h\cdot\kappa_\conv (B_i B_{i+1})$, and if $h\cdot\kappa_\conv (B_i B_{i+1})\leqslant \norm{B_i B_{i+1}}$, 
no interference of the conventional gadgets occurs in edge $B_i B_{i+1}$ of $\Delta_h$, so that the gadgets are actually foldable.
Conversely, $h\cdot\kappa_\conv (B_i B_{i+1})>\norm{B_i B_{i+1}}$, then the supporting pyramids of the conventional gadgets 
would collide in edge $B_i B_{i+1}$ of $\Delta_h$, in which case Construction $\ref{const:conv}$ fails.

Similarly, an interference coefficient $\kappa_\new (B_i B_{i+1})$ in Definition $\ref{def:interference_new}$
means the total length which the tongue of one new gadget 
and one of the ears of the other gadget would occupy in edge $B_i B_{i+1}$ of polyhedron $\Delta_1$ if we choose the inner and the outer pleats appropiately.
Also, we see that if $h\leqslant h_{\max}$, then interference occurs in edge $B_i B_{i+1}$ of $\Delta_h$
if and only if $h\cdot\kappa_\new (B_i B_{i+1})>\norm{B_i B_{i+1}}$.

Summarizing the above, we have the following theorems for the maximal heights of the extrusion with the conventional and the new $3$D gadgets.
\begin{theorem}\label{thm:height_conv}
The maximum value $h_\conv$ of the height $h$ of polyhedron $\Delta_h$ that can be extruded with the conventional $3$D gadgets at one time is given by
\begin{equation*}
h_\conv =\min\left\{\min_{i=1,\dots ,N}\frac{\norm{B_i B_{i+1}}}{\kappa_\conv (B_i B_{i+1})},h_{\max}\right\} .
\end{equation*}
\end{theorem}
\begin{theorem}\label{thm:height_new}
The maximum value $h_\new$ of the height $h$ of polyhedron $\Delta_h$ that can be extruded with the new $3$D gadgets at one time is given by
\begin{equation*}
h_\new =\min\left\{\min_{i=1,\dots ,N}\frac{\norm{B_i B_{i+1}}}{\kappa_\new (B_i B_{i+1})},h_{\max}\right\} 
\end{equation*}
if we choose the order of the outgoing pleats appropiately.
\end{theorem}
\begin{remark}\rm
We defined interference coefficients as ratios to the height $h$ of the extrusion, which is a global variable in the extrusion.
However, dealing with $h$ is a little troublesome because it does not appear explicitly in the development of the polyhedron we want to extrude
and we have to compute each $\lambda_i$ in terms of the angles around ridge $A_i B_i$.
Instead of using $h$, we can also use the length of ridge $A_i B_i$ for any fixed $i$ in the definitions.
This is particularly useful when all $\norm{A_i B_i}$ are the same, 
which occurs for example when $\alpha_i =\alpha$ and $\beta_{i,\sigma}=\beta$ for all $i$ and $\sigma$.
Then we obtain the maximal ridge length of the extrusion instead of the maximal height.
\end{remark}
\begin{remark}\rm
Even if the ears and the tongue of a new gadget does not interfere with adjacent gadgets,
there may be another kind of interference: the outermost point $E_\sigma$ of an outgoing pleat 
may pass across the valley fold $m_{\sigma'}$ of the adjacent gadget.
However, this interference is \emph{avoidable} because we can fold back or sink with a crease parallel to the pleat 
through $G_\sigma$ if $\gamma_\sigma /2 >\delta_\sigma$ and $H_\sigma$ if $\gamma_\sigma /2\leqslant\delta_\sigma$,
where $H_\sigma$ is located closer to $\ell_\sigma$ than $G_\sigma$ if $\gamma_\sigma /2\leqslant\delta_\sigma$. 

If $\delta_\sigma >0$, then
we can fold back or sink the tongue of a new gadget with a crease through $H_\sigma E_{\sigma'}$ if $\delta_{\sigma'}=0$ and
$H_\sigma H_{\sigma'}$ if $\delta_{\sigma'}>0$, to reduce the inner interference coefficient $\kappa_\inn (B_\sigma )$.
But this is not considered here from an aesthetic viewpoint.
\end{remark}
\section{Downward compatibility of the new gadgets}\label{sec:6}
\begin{theorem}\label{thm:comparison_kappa}
Suppose $\delta_\Lt =\delta_\Rt =0$. 
Let $\norm{B_\sigma G_\sigma},\norm{C I_\sigma}$ be as given in Proposition $\ref{prop:calc_length_new}$, 
and $\norm{B_\sigma D_\sigma}$ be as given in Proposition $\ref{prop:calc_length_conv}$.
Then we have inequalities
\begin{equation}\label{ineq_1}
\norm{B_\sigma G_\sigma}\leqslant\norm{B_\sigma D_\sigma},\quad\norm{C I_\sigma}\leqslant\norm{B_\sigma D_\sigma},
\end{equation}
and hence
\begin{equation}\label{ineq_2}
\kappa_\out (B_\sigma )\leqslant\kappa_\conv (B_\sigma ),\quad\kappa_\inn (B_\sigma )\leqslant\kappa_\conv (B_\sigma ),
\end{equation}
where equality holds if $\beta_\sigma +\gamma /4=\pi /2$ for all inequalities.
\end{theorem}
\begin{proof}
We may assume $\norm{AB}=1$.
For the first inequality in $\eqref{ineq_1}$, we calculate as
\begin{align*}
\norm{B_\sigma D_\sigma}-\norm{B_\sigma G_\sigma}&=\frac{\tan\gamma /2}{2\sin\beta_\sigma}-\frac{\sin (\gamma /4)}{\sin (\beta_\sigma -\gamma /4)}\\
&=\sin\frac{\gamma}{4}\cdot\left(\frac{\cos (\gamma /4)}{\cos (\gamma /2)\sin\beta_\sigma}-\frac{1}{\sin (\beta_\sigma -\gamma /4)}\right)\\
&=\frac{\sin (\gamma /4)}{\cos (\gamma /2)\sin\beta_\sigma\sin (\beta_\sigma -\gamma /4)}
\left\{\cos\frac{\gamma}{4}\sin\left(\beta_\sigma -\frac{\gamma}{4}\right)-\cos\frac{\gamma}{2}\sin\beta_\sigma\right\} ,
\end{align*}
where the coefficient outside the parentheses in the last line is positive. 
Using the product-to-sum and sum-to-product formulas, we deform the terms in the parentheses in the last line as
\begin{align*}
\cos\frac{\gamma}{4}\sin\left(\beta_\sigma -\frac{\gamma}{4}\right)&-\cos\frac{\gamma}{2}\sin\beta_\sigma\\
&=\frac{1}{2}\left\{\sin\beta_\sigma +\sin\left(\beta_\sigma -\frac{\gamma}{2}\right)\right\}
-\frac{1}{2}\left\{\sin\left(\beta_\sigma +\frac{\gamma}{2}\right)+\sin\left(\beta_\sigma -\frac{\gamma}{2}\right)\right\}\\
&=\frac{1}{2}\left\{\sin\beta_\sigma -\sin\left(\beta_\sigma -\frac{\gamma}{2}\right)\right\}
=-\cos\left(\beta_\sigma +\frac{\gamma}{4}\right)\sin\frac{\gamma}{4}\geqslant 0
\end{align*}
because of $\pi /2\leqslant\beta_\sigma +\gamma /4<\pi$ and $0<\gamma <\pi$ by Lemma $\ref{lem:condition_gamma}$,
where equality holds for $\beta_\sigma +\gamma /4=\pi /2$.
This proves the first inequality in $\eqref{ineq_1}$.

For the second inequality in $\eqref{ineq_1}$, we have
\begin{equation}\label{BD-CI}
\begin{aligned}
\norm{B_\sigma D_\sigma}&-\norm{C I_\sigma}=\frac{\tan\gamma /2}{2\sin\beta_\sigma}+\frac{1-\cos (\gamma /2)}{2\cos (\beta_\sigma +\gamma /2)\cos (\gamma /2)}\\
&=\frac{\sin (\gamma /4)}{\cos (\gamma /2)}\left(\frac{\cos (\gamma /4)}{\sin\beta_\sigma}-\frac{\sin (\gamma /4)}{-\cos (\beta_\sigma +\gamma /2)}\right)\\
&=\frac{\sin (\gamma /4)}{\cos (\gamma /2)}\cdot\frac{1}{-\cos (\beta_\sigma +\gamma /2)\sin\beta_\sigma}
\left\{\sin\frac{\gamma}{4}\sin\beta_\sigma +\cos\left(\beta_\sigma +\frac{\gamma}{4}\right)\cos\frac{\gamma}{4}\right\} ,
\end{aligned}
\end{equation}
where the coefficient outside the parentheses in the last line is positive. 
Now we define $\theta$ and $\phi_\sigma$ by
\begin{equation*}
\theta =\frac{\gamma}{4} ,\quad\phi_\sigma =\frac{\pi}{2}-\beta_\sigma -\frac{\gamma}{4}.
\end{equation*}
Then we have
\begin{equation*}
\beta_\sigma +\frac{\gamma}{2}=\pi -\alpha_\sigma =\frac{\pi}{2}+\theta +\phi_\sigma <\pi ,
\end{equation*}
so that $\theta$ and $\phi_\sigma$ move in the following range:
\begin{equation}\label{range_theta_phi}
0<\theta <\frac{\pi}{4},\quad 0\leqslant\phi_\sigma ,\quad 0<\theta +\phi_\sigma <\frac{\pi}{2}.
\end{equation}
The terms in the parentheses in the last line in $\eqref{BD-CI}$ become
\begin{align*}
\sin\frac{\gamma}{4}\sin\beta_\sigma +\cos\left(\beta_\sigma +\frac{\gamma}{4}\right)\cos\frac{\gamma}{4}
&=\sin\theta\cos (\theta -\phi_\sigma )-\sin (\theta +\phi_\sigma )\cos\theta\\
&=\sin\theta\sin\phi_\sigma (\sin\theta +\cos\theta )\geqslant 0,
\end{align*}
where equality holds for $\phi_\sigma =0$, i.e., $\beta_\sigma +\gamma /4=\pi /2$.
This proves the second inequality in $\eqref{ineq_1}$.
\end{proof}
Now we can state the \emph{downward compatibility theorem} of the new $3$D gadgets with the conventional ones as follows.
\begin{theorem}\label{thm:downward_compatibility}
Any conventional $3$D gadget can be replaced with our new $3$D gadget with the same outgoing pleats without affecting any other conventional $3$D gadget unless
\begin{equation}\label{exception}
\text{either}\quad\beta_\Lt +\frac{\gamma}{4}<\frac{\pi}{2}\quad\text{or}\quad\beta_\Rt +\frac{\gamma}{4}<\frac{\pi}{2}.
\end{equation}
holds. 
(It is impossible for both of the above inequalities to hold simultaneously.)
\end{theorem}
\begin{proof}
Comparing the necessary conditions in Construction $\ref{const:conv}$ with those in Construction $\ref{const:new}$ for $\delta_\Lt =\delta_\Rt =0$, 
we see that an exceptional case where we cannot replace a conventional $3$D gadget 
with our new one occurs if condition $\mathrm{(iv)}$ of Construction $\ref{const:new}$ fails, i.e., if
\begin{equation*}
\beta_\Lt +\frac{\gamma}{4}<\frac{\pi}{2}\quad\text{or}\quad\beta_\Rt +\frac{\gamma}{4}<\frac{\pi}{2}.
\end{equation*}
Suppose $\beta_\Lt +\gamma /4<\pi /2$. 
Then using $\alpha_\Lt +\beta_\Lt +\gamma /2=\pi$ gives $\beta_\Lt <\alpha_\Lt$.
If $\beta_\Rt +\gamma /4<\pi /2$ also holds, then we have $\beta_\Rt <\alpha_\Rt$, so that $\beta_\Lt +\beta_\Rt <\alpha$,
which does not satisfy condition $\mathrm{(i)}$ of both constructions.

Thus unless $\eqref{exception}$ holds, we can replace a conventional gadget with our new one if we ignore other conventional gadgets.
However, the interference coefficients of the new gadgets are always smaller than or equal to those of the conventional ones 
as shown in Theorem $\ref{thm:comparison_kappa}$, the replacement does not affect any other existing conventional gadget.

Also, in both of Constructions $\ref{const:conv}$ and $\ref{const:new}$ with $\delta_\Lt =\delta_\Rt =0$, the location of point $C$ is the same.
Since the widths of the resulting outgoing pleats are given by $\norm{BC_\Lt}/2$ and $\norm{BC_\Rt}/2$ in both constructions,
the outgoing pleats are the same if we replace the conventional gadget with our new one.
This completes the proof of Theorem $\ref{thm:downward_compatibility}$.
\end{proof}
\begin{corollary}\label{cor:comparison_height}
Let $B_1,\dots ,B_N$ and $\alpha_i,\beta_{i,\sigma},\gamma_i$ be a convex polygon and angles for $i=1,\dots ,N$ and $\sigma =\Lt ,\Rt$.
Suppose conditions $\mathrm{(i)}$--$\mathrm{(iii)}$ hold for all $i$ and $\sigma$. 
Then the maximum value $h_\new$ of the height $h$ of $\Delta_h$ that can be extruded by Construction $\ref{const:new}$ is always greater than or equal to
the maximum value $h_\conv$ of that by Construction $\ref{const:conv}$:
\begin{equation}\label{comparison_height}
h_\new\geqslant h_\conv .
\end{equation}
In particular, if $\beta_{i,\sigma}+\gamma_i /4>\pi /2$ for all $i=1,\dots ,N$ and $\sigma =\Lt ,\Rt$, and $h_\conv <h_{\max}$, 
then we can exclude the equality in $\eqref{comparison_height}$, i.e., $h_\new >h_\conv$. 
\end{corollary}
The latter part of Corollary $\ref{cor:comparison_height}$ applies to the case of prisms, 
which we will deal with in more detail in the next section and refine inequality $\eqref{comparison_height}$ in Corollary $\ref{cor:comparison_height_prism}$.
\section{Numerical comparison of the efficiency for prisms}\label{sec:7}
In this section we restrict ourselves to extruding a \emph{prism} $\Delta_h$ of a convex polygon $B_1\dots B_N$. 
Thus we assume $\beta_{i,\sigma}=\pi /2$ for all $i=1,\dots ,N$ and $\sigma =\Lt ,\Rt$, from which it follows that $\lambda_i=1, h_{\max}=\infty$.
Note that conditions $\mathrm{(i)}$--$\mathrm{(iii)}$ of Construction $\ref{const:conv}$ and 
conditions $\mathrm{(i)}$--$\mathrm{(iv)}$ of Construction $\ref{const:new}$ hold automatically. 

We already know from Corollary $\ref{cor:comparison_height}$ that the new gadgets are more efficient than the conventional ones.
To compare the efficiency numerically, we calculate the ratio $h_\new^\perp /h_\conv^\perp $ of the maximal heights.
For $N=3$, we have the following formula for $h_\conv^\perp$.
\begin{theorem}
For any triangle,
the maximal height of its prism that can be extruded with the conventional $3$D gadgets at one time is the diameter of the incircle of the triangle.
By Heron's formula, for a triangle with side lengths a,b,c, the maximal height $h_\conv^\perp$ is given by
\begin{equation*}
h_\conv^\perp =\sqrt{\frac{(a+b-c)(b+c-a)(c+a-b)}{a+b+c}}.
\end{equation*}
\end{theorem}
\begin{proof}
Consider a triangle $\triangle ABC$ with $\norm{AB}=c,\norm{BC}=a,\norm{CA}=b$. 
Let $r$ be the radius of the incircle of $\triangle ABC$.
Suppose the tangency points of the incircle divide sides $AB,BC,CA$ into lengths of $s$ and $t$, $t$ and $u$, and $u$ and $s$ respectively,
so that $s+t=c,t+u=a,u+s=b$.
Then we have
\begin{equation*}
\tan\frac{\angle CAB}{2}=\frac{r}{s}, \quad\tan\frac{\angle ABC}{2}=\frac{r}{t}, \quad\tan\frac{\angle BCA}{2}=\frac{r}{u}. 
\end{equation*}
Thus the interference coefficients of the vertices are calculated as
\begin{align*}
\kappa_\conv^\perp (A)=\frac{1}{2}\tan\left(\frac{\pi}{2}-\frac{\angle CAB}{2}\right) =\frac{1}{2}\left(\tan\frac{\angle CAB}{2}\right)^{-1}=\frac{s}{2r}, 
\end{align*}
and similarly $\kappa_\conv^\perp (B)=t/2r, \kappa_\conv^\perp (C)=u/2r$. 
It follows that
\begin{equation*}
\kappa_\conv^\perp (AB)=\frac{s+t}{2r}=\frac{c}{2r},\quad\kappa_\conv^\perp (BC)=\frac{s+u}{2r}=\frac{a}{2r},
\quad\kappa_\conv^\perp (CA)=\frac{t+u}{2r}=\frac{b}{2r}.
\end{equation*}
Consequently, we have
\begin{equation}\label{kappa_conv_triangle}
\frac{\norm{AB}}{\kappa_\conv^\perp (AB)}=\frac{\norm{BC}}{\kappa_\conv^\perp (BC)}=\frac{\norm{CA}}{\kappa_\conv^\perp (CA)}=2r,
\end{equation}
which gives that $h_\conv^\perp =2r$.
\end{proof}
Similarly, we can calculate $h_\conv^\perp$ for $N>3$ as follows. 
\begin{theorem}
Let $P=B_1\dots B_N$ be a convex polygon. 
For $i=1,\dots ,N$, define $k_{i,i\pm 1}$ to be the ray starting from $B_i$ and passing through $B_{i\pm 1}$ respectively.
Let $r_{i,i+1}$ be the radius of the circle tangent to side $B_i B_{i+1}$ of $P$ and rays $k_{i,i-1},k_{i+1,i+2}$.
Then we have 
\begin{equation}\label{kappa_conv_r}
\kappa_\conv^\perp =\frac{\norm{B_i B_{i+1}}}{2r_{i,i+1}},
\end{equation}
where $2r_{i,i+1}$ is given by
\begin{equation}\label{expr:2r}
2r_{i,i+1}=\norm{B_i B_{i+1}}\cdot\frac{2\sin (\alpha_i /2)\sin (\alpha_{i+1}/2)}{\sin ((\alpha_i +\alpha_{i+1})/2)}.
\end{equation}
Consequently we have
\begin{equation*}
h_\conv^\perp =\min_{i=1,\dots ,N}2r_{i,i+1}.
\end{equation*}
\end{theorem}
\begin{proof}
Derivation of $\eqref{kappa_conv_r}$ is almost the same as that of $\eqref{kappa_conv_triangle}$.
Thus it remains to compute $r_{i,i+1}$.

Suppose $\alpha_i +\alpha_{i+1}\neq\pi$ and let $T_{i,i+1}$ be the intersection point of $k_{i,i-1}$ and $k_{i+1,i+2}$. 
Then the circle tangent to $B_i B_{i+1}, k_{i,i-1},k_{i+1,i+2}$ is the incircle of $\triangle B_i B_{i+1} T_{i,i+1}$ if $\alpha_i+\alpha_{i+1}<\pi$,
and the excircle $\triangle B_i B_{i+1} T_{i,i+1}$ tangent to $B_i B_{i+1}$ if $\alpha_i+\alpha_{i+1}>\pi$.
By the sine theorem, we have
\begin{align}
\frac{\norm{B_i T_{i,i+1}}}{\sin\alpha'_i}&=\frac{\norm{B_{i+1}T_{i,i+1}}}{\sin\alpha'_{i+1}}=\frac{\norm{B_i B_{i+1}}}{\sin (\alpha'_i+\alpha'_{i+1})},
\quad\text{where}\\
(\alpha'_i,\alpha'_{i+1})&=\begin{cases}(\alpha_i,\alpha_{i+1})&\text{if }\alpha_i+\alpha_{i+1}<\pi ,\\
(\pi -\alpha_i ,\pi -\alpha_{i+1})&\text{if }\alpha_i+\alpha_{i+1}>\pi .\end{cases}
\end{align}
Now we set $P,Q,R,S$ by
\begin{align*}
P&=\phantom{-}\sin\alpha'_i +\sin\alpha'_{i+1}+\sin (\alpha'_i+\alpha'_{i+1}),\\
Q&=-\sin\alpha'_i +\sin\alpha'_{i+1}+\sin (\alpha'_i+\alpha'_{i+1}),\\
R&=\phantom{-}\sin\alpha'_i -\sin\alpha'_{i+1}+\sin (\alpha'_i+\alpha'_{i+1}),\\
S&=\phantom{-}\sin\alpha'_i +\sin\alpha'_{i+1}-\sin (\alpha'_i+\alpha'_{i+1}).
\end{align*}
Setting $s_1=\sin (\alpha'_i /2),s_2=\sin (\alpha'_{i+1}/2),c_1=\cos (\alpha'_i/2),c_2=\cos (\alpha'_{i+1}/2)$ gives that
\begin{align*}
P&=\phantom{-}\sin\alpha'_i (1+\cos\alpha'_{i+1})+\sin\alpha'_{i+1}(1+\cos\alpha'_i)=4c_1c_2(s_1c_2+s_2c_1),\\
Q&=-\sin\alpha'_i (1-\cos\alpha'_{i+1})+\sin\alpha'_{i+1}(1+\cos\alpha'_i)=4c_1s_2(c_1c_2-s_1s_2),\\
R&=\phantom{-}\sin\alpha'_i (1+\cos\alpha'_{i+1})-\sin\alpha'_{i+1}(1-\cos\alpha'_i)=4c_2s_1(c_1c_2-s_1s_2),\\
S&=\phantom{-}\sin\alpha'_i (1-\cos\alpha'_{i+1})+\sin\alpha'_{i+1}(1-\cos\alpha'_i)=4s_1s_2(s_1c_2+c_2c_1)
\end{align*}
by a straightforward calculation.
Then by the formulas for the inradius and exradii, we have
\begin{equation*}
2r_{i,i+1}=\begin{cases}
\displaystyle\frac{\norm{B_i B_{i+1}}}{\sin (\alpha'_i+\alpha'_{i+1})}\sqrt{\frac{QRS}{P}}&\text{if }\alpha_i+\alpha_{i+1}<\pi ,\\
\displaystyle\frac{\norm{B_i B_{i+1}}}{\sin (\alpha'_i+\alpha'_{i+1})}\sqrt{\frac{PQR}{S}}&\text{if }\alpha_i+\alpha_{i+1}>\pi ,
\end{cases}
\end{equation*}
where 
\begin{align*}
\sin (\alpha'_i+\alpha'_{i+1})&=\begin{cases}
\displaystyle\phantom{-}2\sin\frac{\alpha_i+\alpha_{i+1}}{2}\cos\frac{\alpha_i+\alpha_{i+1}}{2}&\text{if }\alpha_i+\alpha_{i+1}<\pi ,\\
\displaystyle -2\sin\frac{\alpha_i+\alpha_{i+1}}{2}\cos\frac{\alpha_i+\alpha_{i+1}}{2}&\text{if }\alpha_i+\alpha_{i+1}>\pi ,\end{cases}\\
\sqrt{\frac{QRS}{P}}&=4\norm{s_1s_2(c_1c_2-s_1s_2)}\\
&=4\sin\frac{\alpha_i}{2}\sin\frac{\alpha_{i+1}}{2}\cos\frac{\alpha_i+\alpha_{i+1}}{2}\quad\text{if }\alpha_i+\alpha_{i+1}<\pi ,\\
\sqrt{\frac{PQR}{S}}&=4\norm{c_1c_2(c_1c_2-s_1s_2)}\\
&=-4\sin\frac{\alpha_i}{2}\sin\frac{\alpha_{i+1}}{2}\cos\frac{\alpha_i+\alpha_{i+1}}{2}\quad\text{if }\alpha_i+\alpha_{i+1}>\pi .
\end{align*}
Thus in both cases $2r_{i,i+1}$ is written as in $\eqref{expr:2r}$.
If $\alpha_i+\alpha_{i+1}=\pi$, then $2r_{i,i+1}$ is the distance between two rays $k_{i,i-1}$ and $k_{i+1,i+2}$, and so
\begin{equation*}
2r_{i,i+1}=\norm{B_i B_{i+1}}\sin\alpha_i =\norm{B_i B_{i+1}}\sin\alpha_{i+1},
\end{equation*}
which is equal to the expression in $\eqref{expr:2r}$ for $\alpha_i+\alpha_{i+1}=\pi$.
\end{proof}
\begin{theorem}\label{thm:comparison_kappa_perp}
Let $P=B_1\dots B_N$ be a convex polygon. 
For any side $B_i B_{i+1}$ of $P$, we have an inequality
\begin{align*}
\kappa_\new^\perp (B_i B_{i+1})&<\frac{3}{4}\kappa_\conv^\perp (B_i B_{i+1}),\quad\text{where }
\frac{\kappa_\new^\perp (B_i B_{i+1})}{\kappa_\conv^\perp (B_i B_{i+1})}\to\frac{3}{4}\text{ as }\alpha_i ,\alpha_{i+1}\to\pi .
\end{align*}
In particular, if $N=3$, i.e., if $P$ is a triangle, then for any $i$ we have
\begin{align*}
\kappa_\new^\perp (B_i B_{i+1})<\frac{1}{\sqrt{2}}\kappa_\conv^\perp (B_i B_{i+1}),\quad\text{where }
\frac{\kappa_\new^\perp (B_i B_{i+1})}{\kappa_\conv^\perp (B_i B_{i+1})}\to\frac{1}{\sqrt{2}}\text{ as }\alpha_i ,\alpha_{i+1}\to\frac{\pi}{2}.
\end{align*}
\end{theorem}
\begin{proof}
Set $s=\tan (\gamma_i /4)$ and $t=\tan (\gamma_{i+1}/4)$. 
Then it follows from $0<\alpha_i=\pi -\gamma_i<\pi$ for all $i$ that $s,t\in (0,1)$. 
The interference coefficients 
$\kappa_\conv^\perp (B_i B_{i+1})$ and $\kappa_\new^\perp (B_i B_{i+1})$ are written as
\begin{align*}
\kappa_\conv^\perp (B_i B_{i+1})&=\kappa_\conv^\perp (B_i)+\kappa_\conv^\perp (B_{i+1})\\
&=\frac{1}{2}\left(\tan\frac{\gamma_i}{2}+\frac{\gamma_{i+1}}{2}\right)
=\frac{s}{1-s^2}+\frac{t}{1-t^2}=\frac{(s+t)(1-st)}{(1-s^2)(1-t^2)},\\
\kappa_\new^\perp (B_i B_{i+1})&=\min\{\kappa_\inn^\perp (B_i)+\kappa_\out^\perp (B_{i+1}),\kappa_\out^\perp (B_i)+\kappa_\inn^\perp (B_{i+1})\}\\
&=\min\left\{\frac{1}{2}\left(\tan\frac{\gamma_i}{2}-\tan\frac{\gamma_i}{4}\right)+\tan\frac{\gamma_{i+1}}{4},
\tan\frac{\gamma_i}{4}+\frac{1}{2}\left(\tan\frac{\gamma_{i+1}}{2}-\tan\frac{\gamma_{i+1}}{4}\right)\right\}\\
&=\min\left\{\frac{s}{1-s^2}-\frac{s}{2}+t,s+\frac{t}{1-t^2}-\frac{t}{2}\right\} .
\end{align*}
Since the expression of $\kappa_\new^\perp (B_i B_{i+1})$ is rather complicated, 
we instead consider an \emph{averaged coefficient} $\kappa_\ave^\perp (B_i B_{i+1})$ defined as
\begin{align*}
\kappa_\ave^\perp (B_i B_{i+1})&=\frac{1}{2}
\left\{(\kappa_\inn^\perp (B_i)+\kappa_\out^\perp (B_{i+1}))+(\kappa_\out^\perp (B_i)+\kappa_\inn^\perp (B_{i+1}))\right\}\\
&=\frac{1}{4}\left(\tan\frac{\gamma_i}{2}+\tan\frac{\gamma_{i+1}}{2}+\tan\frac{\gamma_i}{4}+\tan\frac{\gamma_{i+1}}{4}\right)\\
&=\frac{1}{2}\left(\frac{s}{1-s^2}+\frac{t}{1-t^2}+\frac{1}{2}(s+t)\right) =\frac{1}{2}\kappa_\conv (B_i B_{i+1})+\frac{1}{4}(s+t).
\end{align*}
Then we have an inequality
\begin{equation}\label{kappa_ave_conv}
\frac{\kappa_\new^\perp}{\kappa_\conv^\perp}\leqslant\frac{\kappa_\ave^\perp}{\kappa_\conv^\perp}=\frac{1}{2}+\frac{(1-s^2)(1-t^2)}{4(1-st)}
=\frac{3}{4}-\frac{st}{2}-\frac{(s-t)^2}{2(1-st)}.
\end{equation}
Let $S=(0,1)\times (0,1)$ be a product set and $S_u=S\cap\set{st=u}$ for $u\in (0,1)$.
Then since $\kappa_\ave^\perp /\kappa_\conv^\perp$ attains maximum on $S_u$ at $s=t=\sqrt{u}$ , we have
\begin{equation}\label{ineq:sup_S}
\sup_S\frac{\kappa_\new^\perp}{\kappa_\conv^\perp}\leqslant\sup_S\frac{\kappa_\ave^\perp}{\kappa_\conv^\perp}
=\sup_{u\in (0,1)}\sup_{S_u}\frac{\kappa_\ave^\perp}{\kappa_\conv^\perp}
=\sup_\Delta\frac{\kappa_\ave^\perp}{\kappa_\conv^\perp}=\sup_\Delta\frac{\kappa_\new^\perp}{\kappa_\conv^\perp},
\end{equation}
where $\Delta =\set{(s,t)\in S|s=t}$ is a diagonal set, and the last equality holds because 
$\kappa_\new^\perp /\kappa_\conv^\perp=\kappa_\ave^\perp /\kappa_\conv^\perp$ on $\Delta$.
On the other hand, it follows clearly from $\Delta\subset S$ that
\begin{equation}\label{ineq:sup_Delta}
\sup_\Delta\frac{\kappa_\new^\perp}{\kappa_\conv^\perp}\leqslant\sup_S\frac{\kappa_\new^\perp}{\kappa_\conv^\perp}.
\end{equation}
Combining $\eqref{ineq:sup_S}$ and $\eqref{ineq:sup_Delta}$ gives that
\begin{equation*}
\sup_S\frac{\kappa_\new^\perp}{\kappa_\conv^\perp}=\sup_\Delta\frac{\kappa_\new^\perp}{\kappa_\conv^\perp}
=\sup_{s\in (0,1)}\left(\frac{3}{4}-\frac{s^2}{2}\right) =\frac{3}{4}.
\end{equation*}
Thus $\kappa_\new^\perp /\kappa_\conv^\perp$ attains its supremum value $3/4$ as $s=t\to +0$ along $\Delta$.
By the continuity of $\kappa_\new^\perp /\kappa_\conv^\perp$ on $S$,  
$\kappa_\new^\perp /\kappa_\conv^\perp$ attains its supremum value $3/4$ as $s,t\to +0$ also on $S$, 
where $s,t\to +0$ corresponds to $\alpha_i ,\alpha_{i+1}\to\pi -0$.

Next we shall find the supremum value of $\kappa_\new^\perp /\kappa_\conv^\perp$ if $P=B_1\dots B_N$ is a triangle. 
In this case we have $\alpha_i+\alpha_{i+1}<\pi$ for all $i$, which is equivalent to $\gamma_i +\gamma_{i+1}>\pi$.
Setting $s=\tan (\gamma_i/4),t=\tan (\gamma_{i+1}/4)$ as above, we have
\begin{equation*}
1=\tan\frac{\pi}{4}<\tan\left(\frac{\gamma_i}{4}+\frac{\gamma_{i+1}}{4}\right) =\frac{s+t}{1-st},
\end{equation*}
so that 
\begin{equation}\label{ineq:1-u<v}
1-u<v,\quad\text{where we set }u=st,v=s+t.
\end{equation}
We see from $\eqref{kappa_ave_conv}$ that $\kappa_\ave^\perp /\kappa_\conv^\perp$ can be written in terms of $u$ and $v$ as
\begin{equation}\label{kappa_ave_conv_uv}
\frac{\kappa_\ave^\perp}{\kappa_\conv^\perp}=\frac{1}{2}+\frac{1}{4}\cdot\frac{(1+u)^2-v^2}{1-u}.
\end{equation}
Now let $U=\set{(u,v)\in\R^2|(s,t)\in S,u=st,v=s+t,1-u<v}$ be the domain of $\kappa_\ave^\perp /\kappa_\conv^\perp$.
\begin{lemma}
The domain $U$ is given by
\begin{equation}\label{expr:U}
U=\set{(u,v)\in\R^2|1-u<v<1+u,4u\leqslant v^2}.
\end{equation}
Also, $\kappa_\ave^\perp /\kappa_\conv^\perp$ attains its supremum value $1/\sqrt{2}$ as $u\to 3-2\sqrt{2}+0$ along the graph $v=2\sqrt{u}$.
\end{lemma}
\begin{proof}
For a fixed value of $u\in (0,1)$, $s$ moves in the range $u<s<1$. 
Then we have $v=v(s)=s+u/s$ with its derivative $v'(s)=1-u/s^2$, which gives the following table.
\begin{center}
\begin{tabular}{c|ccccc}
$s$&$u$&$\cdots$&$\sqrt{u}$&$\cdots$&$1$\\ \hline
$v'(s)$&&$-$&$0$&$+$&\\
$v(s)$&$1+u$&$\searrow$&$2$$\sqrt{u}$&$\nearrow$&$1+u$
\end{tabular}
\end{center}
Thus $v$ moves in the range $2\sqrt{u}\leqslant v<1+u$.
Considering $\eqref{ineq:1-u<v}$ also, we obtain $\eqref{expr:U}$.

For the supremum value of $\kappa_\ave^\perp /\kappa_\conv^\perp$ on $U$, we consider the set $\underline{U}$ bounding $U$ from below with respect to $v$,
because for a fixed value of $u$, $\eqref{kappa_ave_conv_uv}$ attains supremum for the infimum value of $v$.
By a straightforward calculation, $\underline{U}$ is given by
$v=1-u$ for $0<u\leqslant 3-2\sqrt{2}$ and $v=2\sqrt{u}$ for $3-2\sqrt{2}\leqslant u<1$. 
Thus we have
\begin{equation*}
\sup_{\substack{(u,v)\in U\\ u:\text{fixed}}}\frac{\kappa_\ave^\perp}{\kappa_\conv^\perp}=
\begin{cases}
\displaystyle -\frac{1}{2}+\frac{1}{1-u}&\text{if }0<u\leqslant 3-2\sqrt{2}\\
\displaystyle \frac{3}{4}-\frac{u}{4}&\text{if }3-2\sqrt{2}\leqslant u<1.
\end{cases}
\end{equation*}
Hence $\kappa_\ave^\perp /\kappa_\conv^\perp$ attains its spremum value $1/\sqrt{2}$ as $(u,v)\to (3-2\sqrt{2},2-2\sqrt{2})$ along $v=2\sqrt{u}$,
which corresponds to $s=t\to 1-\sqrt{2}$.
\end{proof}
Note that the set $\set{v=2\sqrt{u}}$ bounding $U$ corresponds to the subset $\Delta =\set{(s,s)}$ of $S$, on which 
$\kappa_\ave^\perp /\kappa_\conv^\perp =\kappa_\new^\perp /\kappa_\conv^\perp$. 
Hence by a simular argument as above, $\kappa_\new^\perp /\kappa_\conv^\perp$ attains supremum on
$S\cap\set{1+st<s+t}$ as $s,t\to\sqrt{2}-1+0$, which corresponds to $\alpha_i ,\alpha_{i+1}\to\pi /2-0$.
This completes the proof of Theorem $\ref{thm:comparison_kappa_perp}$.
\end{proof}
\begin{corollary}\label{cor:comparison_height_prism}
For any convex polygon $P=B_1\dots B_N$, we have the following inequality 
for the maximal heights $h_\conv^\perp ,h_\new^\perp$ of the prisms of $P$ extruded with the conventional and the new gadgets:
\begin{equation*}
h_\new^\perp >\frac{4}{3}h_\conv^\perp .
\end{equation*}
In particular, if $N=3$, i.e., if $P$ is a triangle, then we have
\begin{equation*}
h_\new^\perp >\sqrt{2}h_\conv^\perp .
\end{equation*}
\end{corollary}
We show the $3$D plots of $\kappa_\new^\perp /\kappa_\conv^\perp$ and $K_\ave^\perp /K_\conv^\perp$ 
as fuctions of $s=\tan (\gamma_i /4)$ and $t=\tan (\gamma_{i+1}/4)$.
Observe that the $3$D plot of $\kappa_\new^\perp /\kappa_\conv^\perp$ has `ridges' over the curves $\set{s=t}$ and $\set{3s^2 t^2-3s^2-3t^2-2st+1=0}$
in the $st$-plane, 
where $\kappa_\new^\perp /\kappa_\conv^\perp =\kappa_\ave^\perp /\kappa_\conv^\perp$. 
\begin{figure}[htbp]
  \begin{center}
    \begin{tabular}{c}
\addtocounter{theorem}{1}
      \begin{minipage}{0.5\hsize}
        \begin{center}
          \includegraphics[width=\hsize]{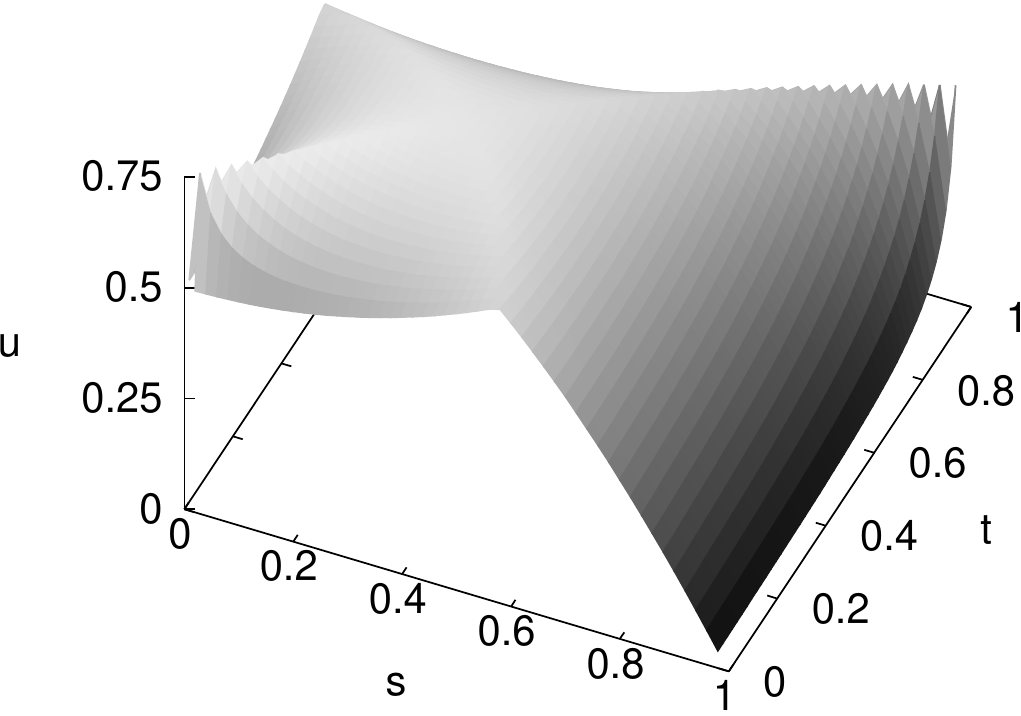}
        \end{center}
    \caption{$3$D plot of $u=\kappa_\new^\perp /\kappa_\conv^\perp$}
    \label{fig:Knew_Kconv}
      \end{minipage}
\addtocounter{theorem}{1}
      \begin{minipage}{0.5\hsize}
        \begin{center}
          \includegraphics[width=\hsize]{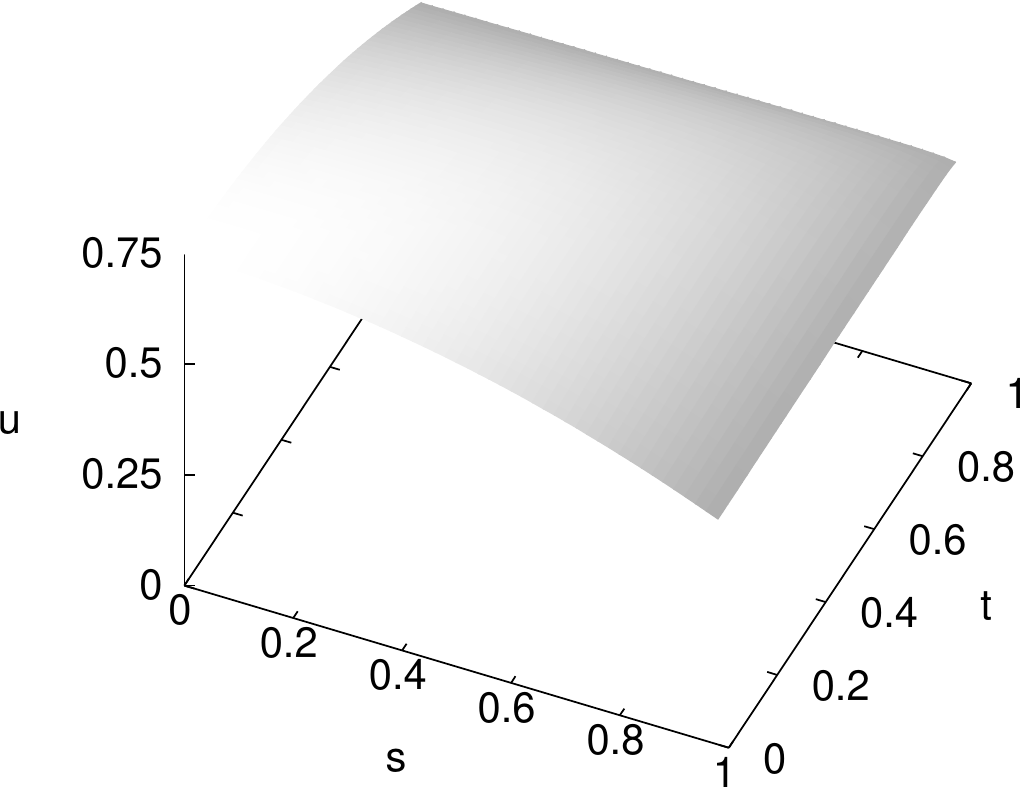}
        \end{center}
    \caption{$3$D plot of $u=\kappa_\ave^\perp /\kappa_\conv^\perp$}
    \label{fig:Kave_Kconv}
      \end{minipage}
    \end{tabular}
  \end{center}
\end{figure}
\section{Repetition and division of $3$D gadgets}\label{sec:8}
Since the paper is flat outside the extruded polyhedron in both of Construcions $\ref{const:conv}$ and $\ref{const:new}$, 
we can extrude another polyhedron so that its top or side face includes the bottom face of the previous polyhedron. 
In particular, we can repeat the same gadgets to make the extrusion higher as long as no interference occurs.
Changing the viewpoint, we can divide a gadget into smaller gadgets which extrudes the same height in total. 
However, naive repetition of our new gadgets does not keep the back sides flat. 
In this section we present without proof a modification of our new gadgets for repetition which keeps the back sides flat. 
We shall deal with proportional division of a gadget into $d$ gadgets in the ratio $p_1:\dots :p_d$ with $p_1+\dots +p_d=d$.
Throught this section, we will suppose $\delta_\Lt =\delta_\Rt =0$.
\begin{construction}\label{const:division_new}\rm
Consider a development as shown in Figure $\ref{fig:division_development}$, for which we require the following conditions.
\begin{enumerate}[(i)]
\item $\alpha <\beta_\Lt+ \beta_\Rt$, $\beta_\Lt <\alpha+ \beta_\Rt$ and $\beta_\Rt <\alpha+ \beta_\Lt$.
\item $\alpha +\beta_\Lt +\beta_\Rt <2\pi$.
\item $\alpha +\beta_\Lt +\beta_\Rt >\pi$.
\item $\beta_\sigma+\gamma /4\geqslant\pi /2\quad\text{for }\sigma =\Lt ,\Rt$.
\end{enumerate}
The crease pattern of the proportional division of our new $3$D gadget into $d$ gadgets in the ratio $p_1:\dots :p_d$ from the bottom with $p_1+\dots +p_d=d$
is constructed as follows, where all procedures are done for both $\sigma =\Lt ,\Rt$.
\begin{enumerate}
\item Draw a perpendicular to $\ell_\sigma$ through $B_\sigma$ for both $\sigma =\Lt ,\Rt$, letting $C$ be the intersection point.
\item Divide segment $CA$ (not $AC$) into $d$ parts proportionally in the ratio $p_1:\dots :p_d$, 
letting $A^{(1)},\dots ,A^{(d-1)}$ to be the divided points in order from the side of $C$.
\item For $n=1,\dots ,d-1$, draw a ray $\ell_\sigma^{(n)}$ parallel to $\ell_\sigma$, starting from $A^{(n)}$ and going in the same direction as $\ell_\sigma$,
letting $B_\sigma^{(n)}$ be the intersection point with $B_\sigma C$.
\item For $n=1,\dots ,d$, let $E_\sigma^{(n)}$ be the intersection point of 
a bisector of $\angle B_\sigma^{(n)}A_\sigma^{(n)}C$ and a perpendicular bisector of $B_\sigma^{(n-1)}B_\sigma^{(n)}$, 
where we set $A^{(n)}=A,B_\sigma^{(0)}=C$ and $B_\sigma^{(n)}=B_\sigma$. 
Also, draw a ray $m_\sigma^{(n)}$ parallel to $\ell_\sigma$, starting from $E_\sigma^{(n)}$ and going in the same direction as $\ell_\sigma$.
Thus we have $2d$ parallel rays $m_\sigma^{(1)},\ell_\sigma^{(1)},\dots ,m_\sigma^{(n)},\ell_\sigma^{(n)}=\ell_\sigma$ 
in order from the side of $C$.
\item For $n=1,\dots ,d$, let $F^{(n)}$ be the intersection point of segments $AC$ and $E_\Lt^{(n)}E_\Rt^{(n)}$.
\item For $n=1,\dots ,d-1$, draw a ray $k_\sigma^{(n)}$ parallel to $k_\sigma$ from $B_\sigma^{(n)}$ to the side of $m_\sigma^{(n+1)}$.
If $k_\sigma^{(n)}$ intersects segment $A^{(n)}E_\sigma^{(n+1)}$, then let ${G'_\sigma}^{\! (n)}$ be the intersection point.
If not, then let $J_\sigma^{(n+1)}$ be the intersection point with $m_\sigma^{(n)}$. 
\item For $n=1,\dots ,d$, draw a ray ${k'_\sigma}^{\! (n)}$ starting from $B_\sigma^{(n)}$ 
which is a reflection of $k_\sigma^{(n)}$ across $\ell_\sigma^{(n)}$, where we set $k_\sigma^{(d)}=k_\sigma$.
If ${k'_\sigma}^{\! (n)}$ intersects segment $A^{(n)}E_\sigma^{(n)}$, then let $G_\sigma^{(n)}$ be the intersection point.
If not, then it ${k'_\sigma}^{\! (n)}$ intersects $m_\sigma^{(n)}$ at $J_\sigma^{(n)}$ given in $(6)$.
Note that $G_\sigma^{(1)}$ always exists, and for $2\leqslant n\leqslant d$, $G_\sigma^{(n)}$ exists if and only if ${G'_\sigma}^{\! (n-1)}$ exists.
\item For $n=1,\dots ,d$, draw a circle with center $A^{(n)}$ through $B_\Lt^{(n)}$ and $B_\Rt^{(n)}$.
If the circle intersects segment $A^{(n)}F^{(n)}$ excluding endpoint $F^{(n)}$, then let $D^{(n)}$ be the intersection point, 
and if $n\leqslant d-1$ and it intersects segment $A^{(n)}F^{(n+1)}$ excluding endpoint $F^{(n+1)}$, then let ${D'}^{(n)}$ be the intersection point.
Note that $D^{(1)}$ always exists. Also, as we will see in $\eqref{AF-AB}$ in Lemma $\ref{lem:length_AF}$, we have
\begin{equation*}
\norm{A^{(n)}F^{(n)}}-\norm{A^{(n)}B_\sigma^{(n)}}=\norm{A^{(n-1)}F^{(n)}}-\norm{A^{(n-1)}B_\sigma^{(n-1)}}\quad\text{for }2\leqslant n\leqslant d,
\end{equation*}
which implies that $D^{(n)}$ exists if and only if ${D'}^{(n-1)}$ exists.
\item For $n=2,\dots ,d$ such that $D^{(n)}$ and ${D'}^{(n-1)}$ in $(8)$ exist, draw a parallel line to segment $D^{(1)}E_\sigma^{(1)}$ through $D^{(n)}$ 
and a parallel line to segment $D^{(1)}E_{\sigma'}^{(1)}$ through ${D'}^{(n-1)}$, 
letting $K_\sigma^{(n)}$ be the common intersection point with segment $E_\Lt^{(n)}E_\Rt^{(n)}$.
\item For $n=2,\dots ,d$ 
such that $D^{(n)}$ and ${D'}^{(n-1)}$ in $(8)$ exist, 
draw a ray starting from $D^{(n)}$ which is a reflection of ${k'_\sigma}^{\! (n)}$ across $A^{(n)}E_\sigma^{(n)}$.
If the ray intersects $A^{(n)}E_\sigma^{(n)}$, then $G_\sigma^{(n)}$ is the intersection point.
If not, then the ray intersects $E_\Lt^{(n)}E_\Rt^{(n)}$, letting $M_\sigma^{(n)}$ be the intersection point.
Also, draw a ray starting from ${D'_\sigma}^{\! (n-1)}$ which is a reflection of $k_\sigma^{(n-1)}$ across $A^{(n-1)}E_\sigma^{(n)}$.
Then the ray passes through ${G'_\sigma}^{\! (n-1)}$ (resp. $M_\sigma^{(n)}$) if ${G'_\sigma}^{\! (n-1)}$ exists (resp. does not exist). 
\item For $n=2,\dots ,d$ such that $G_\sigma^{(n)},{G'_\sigma}^{\! (n-1)}$ exist and $D^{(n)},{D'}^{(n-1)}$ do not, 
draw a line through $G_\sigma^{(n)}$ which is a reflection of ${k'_\sigma}^{\! (n)}$ across $A^{(n)}E_\sigma^{(n)}$,
and a line through ${G'_\sigma}^{\! (n-1)}$ which is a reflection of ${k'_\sigma}^{\! (n)}$ across $A^{(n)}E_\sigma^{(n)}$, 
letting $M_\sigma^{(n)}$ be their common intersection point with segment $E_\Lt^{(n)} E_\Rt^{(n)}$.
(If the two lines overlap $E_\Lt^{(n)} E_\Rt^{(n)}$, 
which happens only if $\norm{A^{(n)}F^{(n)}}=\norm{A^{(n)}B_\sigma^{(n)}}$ and $G_\sigma^{(n)}={G'_\sigma}^{\! (n-1)}=E_\sigma^{(n)}$, 
then let $M_\sigma^{(n)}=E_\sigma^{(n)}$.)
\item The crease pattern is shown as the solid lines in Figure $\ref{fig:division_new}$, 
and the assignment of mountain folds and valley folds is given in Table $\ref{tbl:assignment_new_division}$ if $\beta_\sigma +\gamma /4>\pi /2$,
and Table $\ref{tbl:assignment_new_division_phi=0}$ if $\beta_\sigma +\gamma /4=\pi /2$.
\end{enumerate}
\end{construction}
\begin{figure}[htbp]
  \begin{center}
\addtocounter{theorem}{1}
          \includegraphics[width=\hsize]{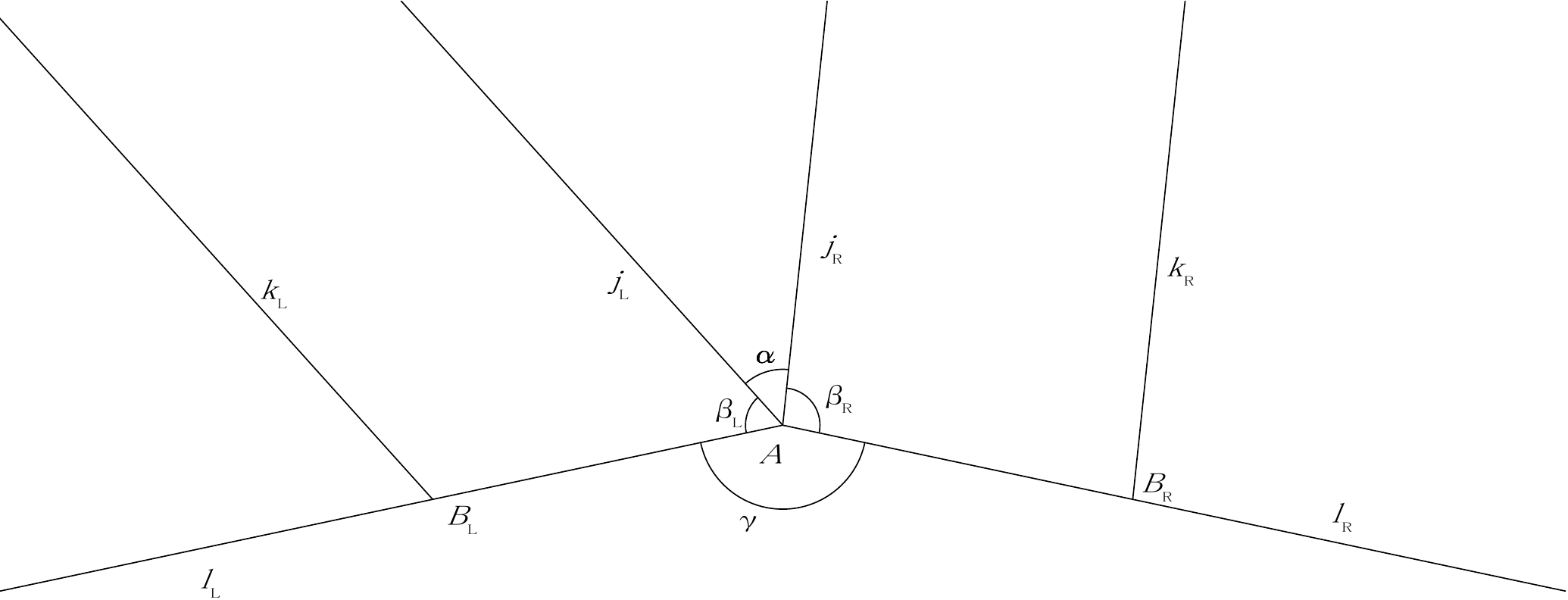}
    \caption{Development to which we apply $d$ $3$D gadgets successively}
    \label{fig:division_development}
\end{center}
\end{figure}
\begin{figure}[htbp]
  \begin{center}
\addtocounter{theorem}{1}
          \includegraphics[width=\hsize]{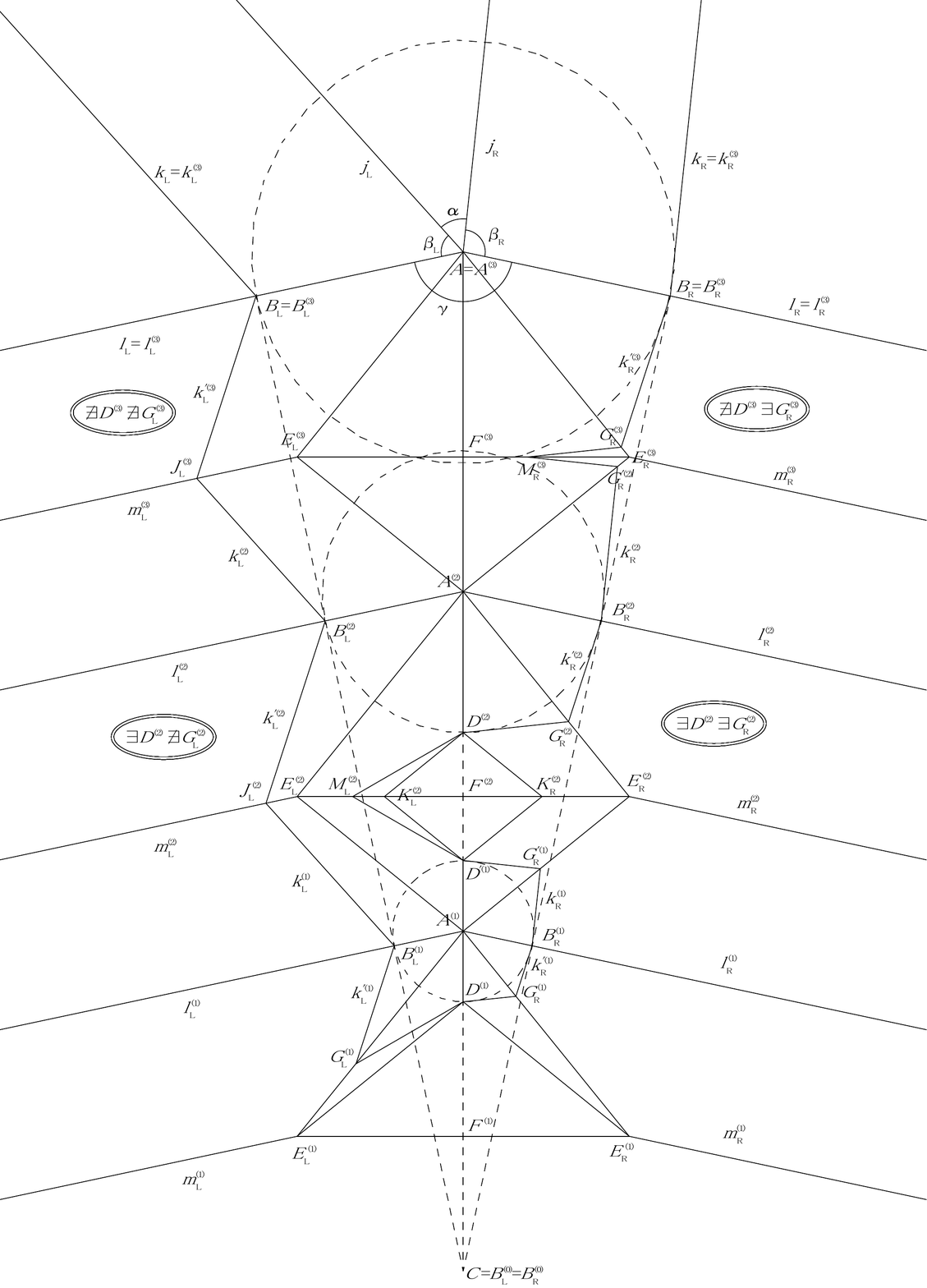}
    \caption{CP of the division of the new gadget}
    \label{fig:division_new}
\end{center}
\end{figure}
\addtocounter{theorem}{1}
\begin{table}[h]
\begin{tabular}{c|c|c|c|c|c}
&common&\multicolumn{4}{c}{$n$ with $2\leqslant n\leqslant d$ such that}\\
\cline{3-6}
&creases&$\exists D^{(n)}\exists G_\sigma^{(n)}$&$\exists D^{(n)}\nexists G_\sigma^{(n)}$
&$\nexists D^{(n)}\exists G_\sigma^{(n)}$&$\nexists D^{(n)}\nexists G_\sigma^{(n)}$\\
\hline
&$j_\sigma ,\ell_\sigma^{(n)},$&\multicolumn{4}{c}{$A^{(n-1)}E_\sigma^{(n)}$}\\
\cline{3-6}
mountain&$AB_\sigma ,A^{(1)}D^{(1)},$&\multicolumn{2}{c|}{$A^{(n)}D^{(n)},D^{(n)}K_\sigma^{(n)}$}
&\multicolumn{2}{c}{$A^{(n)}F^{(n)}$}\\
\cline{3-6}
folds&$B_\sigma^{(1)}G_\sigma^{(1)},D^{(1)}E_\sigma^{(1)}$&$B_\sigma^{(n-1)}{G'_\sigma}^{\! (n-1)},$
&$B_\sigma^{(n)}J_\sigma^{(n)},$&$B_\sigma^{(n-1)}{G'_\sigma}^{\! (n-1)},$&$B_\sigma^{(n)}J_\sigma^{(n)}$\\
&&${D'_\sigma}^{\! (n-1)}{G'_\sigma}^{\! (n-1)}$&${D'_\sigma}^{\! (n-1)}M_\sigma^{(n)}$
&${G'_\sigma}^{\! (n-1)}M_\sigma^{(n)}$& \\ \hline
valley&$k_\sigma ,m_\sigma^{(n)},$&\multicolumn{2}{c|}{$A^{(n-1)}{D'}^{(n-1)},{D'}^{(n-1)}K_\sigma^{(n)}$}
&\multicolumn{2}{c}{$A^{(n-1)}F^{(n)}$}\\
\cline{3-6}
folds&$A^{(n)}E_\sigma^{(n)},E_\Lt^{(n)}E_\Rt^{(n)},$&$B_\sigma^{(n-1)}G_\sigma^{(n)},$
&$B_\sigma^{(n-1)}J_\sigma^{(n)},$&$B_\sigma^{(n-1)}G_\sigma^{(n)},$&$B_\sigma^{(n-1)}J_\sigma^{(n)}$\\
&$D^{(1)}G_\sigma^{(1)}$&$D^{(n)}G_\sigma^{(n)}$&$D^{(n)}M_\sigma^{(n)}$&$G_\sigma^{(n)}M_\sigma^{(n)}$&
\end{tabular}\vspace{0.5cm}
\caption{Assignment of mountain folds and valley folds for the division of the new gadget for $\beta_\sigma +\gamma /4>\pi /2$}
\label{tbl:assignment_new_division}
\end{table}
\begin{lemma}\label{lem:length_AF}
Suppose $\norm{AB}=d$ and let $q_n=p_1+\dots +p_n$, so that $\norm{A^{(n)}B_\sigma^{(n)}}=q_n$. 
Then we have
\begin{equation*}
\norm{A^{(n-1)}F^{(n)}}=\frac{p_n}{2}\left(\frac{1}{\cos (\gamma /2)}-1\right) ,\quad
\norm{A^{(n)}F^{(n)}}=\frac{p_n}{2}\left(\frac{1}{\cos (\gamma /2)}+1\right) \quad\text{for }2\leqslant n\leqslant d.
\end{equation*}
Also, we have
\begin{equation}\label{AF-AB}
\norm{A^{(n-1)}F^{(n)}}-\norm{A^{(n-1)}B_\sigma^{(n-1)}}=\norm{A^{(n)}F^{(n)}}-\norm{A^{(n)}B_\sigma^{(n)}}
=\frac{p_n}{2}\left(\frac{1}{\cos (\gamma /2)}-1\right) -q_{n-1}.
\end{equation}
\end{lemma}
\begin{proof}
Note that $\angle A^{(n)}E_\sigma^{(n)}A^{(n-1)}=\pi /2$ for $n\geqslant 2$, so that $\angle A^{(n)}E_\sigma^{(n)}F^{(n)}=\pi /2-\gamma /4$.
Thus we see from $\eqref{length_new_delta=0}$ that
\begin{align*}
\norm{A^{(n-1)}F^{(n)}}&=\norm{A^{(n-1)}A^{(n)}}\cos^2\left(\frac{\pi}{2}-\frac{\gamma}{4}\right) 
=\frac{p_n}{\cos (\gamma /2)}\cdot\frac{1-\cos (\gamma /2)}{2}=\frac{p_n}{2}\left(\frac{1}{\cos (\gamma /2)}-1\right) ,\\
\norm{A^{(n)}F^{(n)}}&=\norm{A^{(n-1)}A^{(n)}}\cos^2\frac{\gamma}{4}
=\frac{p_n}{\cos (\gamma /2)}\cdot\frac{1+\cos (\gamma /2)}{2}=\frac{p_n}{2}\left(\frac{1}{\cos (\gamma /2)}+1\right) 
\end{align*}
as desired. 
The latter equation $\eqref{AF-AB}$ is then trivial.
\end{proof}
\begin{remark}\rm
In Construction $\ref{const:division_new}$, $(8)$,
if ${D'}^{(n)}$ exists, then ${D'}^{(n)}$ overlaps with $D^{(n)}$ because of $\norm{A^{(n)}D^{(n)}}=\norm{A^{(n)}{D'}^{(n)}}=p_1+\dots +p_n$.
Also, if ${D'}^{(n)}$ exists, then $D^{(n+1)}$ overlaps with ${D'}^{(n)}$ by $\eqref{AF-AB}$.
Thus $D^{(n)}$ and ${D'}^{(n-1)}$ overlap with $D^{(1)}$ as long as they exist.
In other words, if $D^{(n)}$ exists, then $\triangle A^{(n)}E_\Lt^{(n)}E_\Rt^{(n)}$ 
overlaps with the innermost point $D^{(1)}$ of the tongue of the lowest gadget. 
Also, if ${D'}^{(n-1)}$ exists, then $\triangle A^{(n-1)}E_\Lt^{(n)}E_\Rt^{(n)}$ overlaps with $D^{(1)}$.
\end{remark}
\begin{proposition}\label{prop:existence_D}
Suppose $n\geqslant 2$ and let $q_n=p_1+\dots +p_n$. 
Then points $D^{(n)}$ and ${D'}^{(n-1)}$ given in Construction $\ref{const:division_new}$, $(8)$ exist if and only if
\begin{equation}\label{existence_D}
q_{n-1}<\frac{p_n}{2}\left(\frac{1}{\cos (\gamma /2)}-1\right) . 
\end{equation}
In particular, if $p_1=\dots =p_d=1$ (equal division) and 
$0<\gamma \leqslant2\cos^{-1}(1/3)\approx 2.4619\approx 141.06^\circ$, then $D^{(n)}$ and ${D'}^{(n-1)}$ do not exist for $n\geqslant 2$.
\end{proposition}
\begin{proof}
The first assertion is clear 
because it follows from Construction $\ref{const:division_new}$, $(8)$ 
that $\eqref{AF-AB}$ in Lemma $\ref{lem:length_AF}$ is positive if and only if $D^{(n)}$ and ${D'}^{(n-1)}$ exist.
The second assertion is derived by using $p_n=1,q_n=n$ and setting the right-hand side (RHS) of $\eqref{existence_D}$ to be $2$.
\end{proof}
In the case of equal division,bwe show in Table $\ref{tbl:n_D}$ the number $n_D$ for various values of $\gamma$ 
such that $D_\sigma^{(n)}$ and ${D'_\sigma}^{\! (n-1)}$ for $n\geqslant 2$ exist if and only if $n\leqslant n_D$.
\addtocounter{theorem}{1}
\begin{table}[h]
\begin{tabular}{r|ccc|cc|cc|cc|cc|cc|c}
$\gamma\backslash n_D$&\multicolumn{3}{c|}{none}&\multicolumn{2}{c|}{$2$}&\multicolumn{2}{c|}{$3$}&\multicolumn{2}{c|}{$4$}&\multicolumn{2}{c|}{$5$}&$\cdots$\\ 
\hline
$\gamma (\text{rad})$&$0$&$\cdots$&$2.4619$&$\cdots$&$2.7389$&$\cdots$&$2.8549$&$\cdots$&$2.9189$&$\cdots$&$2.9595$&$\cdots$\\
$\gamma ({}^\circ )$&$0$&$\cdots$&$141.06$&$\cdots$&$156.93$&$\cdots$&$163.57$&$\cdots$&$167.24$&$\cdots$&$169.57$&$\cdots$\\
$\cos (\gamma /2)$&$1$&$\cdots$&$1/3$&$\cdots$&$1/5$&$\cdots$&$1/7$&$\cdots$&$1/9$&$\cdots$&$1/11$&$\cdots$
\end{tabular}
\caption{The number $n_D$ for various $\gamma$ such that $D^{(n)}$ and ${D'}^{(n-1)}$ 
exist for $n\geqslant 2$ if and only if $n\leqslant n_D$ in the case of equal division}
\label{tbl:n_D}
\end{table}
\begin{proposition}\label{prop:existence_G}
Suppose $n\geqslant 2$ and let $q_n=p_1+\dots +p_n$ as before.
In Construction $\ref{const:division_new}$, points $G_\sigma^{(n)}$ and ${G'_\sigma}^{\! (n-1)}$ given in $(7)$ exist if and only if 
\begin{equation}\label{existence_G}
q_{n-1}<\left( 1-\frac{1}{1+\tan (\gamma /4)\tan\phi_\sigma}\right)\frac{p_n}{\cos (\gamma /2)}, 
\end{equation}
where we set $\phi_\sigma =\beta_\sigma +\gamma /4-\pi /2$.
In particular, if $\phi_\sigma =0$, 
then $G_\sigma^{(n)}$ and ${G'_\sigma}^{\! (n-1)}$ do not exist.
Also, in the case of equal division, if $\gamma\leqslant 2\pi /3$, then $G_\sigma^{(n)}$ and ${G'_\sigma}^{\! (n-1)}$ do not exist,
and if $2\pi /3<\gamma\leqslant 2\cos^{-1}(1/4)\approx 2.6362\approx 151.04^\circ$, 
then $G_\sigma^{(n)}$ and ${G'_\sigma}^{\! (n-1)}$ can exist only for $n=2$. 
\end{proposition}
\begin{proof}
Suppose $\norm{AB}=d$, so that $\norm{A^{(n)}B_\sigma^{(n)}}=q_n$. 
Let $N_\sigma^{(n)}$ be the intersection point of $k_\sigma^{(n)}$ and (possibly an extension of) $m_\sigma^{(n)}$.
Then we can define the \emph{signed distance} $sd(E_\sigma^{(n)},N_\sigma^{(n)})$
between $E_\sigma^{(n)}$ and $N_\sigma^{(n)}$ on $m_\sigma^{(n)}$ by
\begin{equation}\label{sd_EN}
sd(E_\sigma^{(n)},N_\sigma^{(n)})=q_n-\frac{p_n\tan (\gamma /2)}{2}\left(\frac{1}{\tan (\gamma /4)}+\frac{1}{\tan (\pi -\beta_\sigma )}\right) ,
\end{equation}
where $\tan (\gamma /2)/2$ is the distance between $\ell_\sigma^{(n)}$ and $m_\sigma^{(n)}$.
We see that ${k'_\sigma}^{\! (n)}$ intersects segment $A^{(n)}E_\sigma^{(n)}$, that is, $G_\sigma^{(n)}$ exists, 
if and only if $sd(E_\sigma^{(n)},N_\sigma^{(n)})<0$.
Thus to prove the first assertion of the proposition, it is sufficient to prove the following result.
\begin{lemma}\label{lem:calc_sd_EN_2}
We have
\begin{equation}\label{ineq:sd_EN_2}
\begin{aligned}
\frac{\tan (\gamma /2)}{2}&\left(\frac{1}{\tan (\gamma /4)}+\frac{1}{\tan (\pi -\beta_\sigma )}\right)\\
=1+&\left( 1-\frac{1}{1+\tan (\gamma /2)\tan\phi_\sigma}\right)\frac{1}{\cos (\gamma /2)}<1+\frac{1}{2\cos (\gamma /2)}.
\end{aligned}
\end{equation}
\end{lemma}
\begin{proof}
As in the proof of Proposition $\ref{prop:calc_length_new}$, $\theta =\gamma /2$ and $\phi_\sigma =\beta_\sigma +\gamma /4 -\pi /2$ move in the range 
given by $\eqref{range_theta_phi}$.
If we set $t=\tan\theta ,u=\tan\phi_\sigma$ also, then the left-hand side (LHS) of $\eqref{ineq:sd_EN_2}$ becomes 
\begin{align*}
\text{LHS of $\eqref{ineq:sd_EN_2}$}&=\frac{\tan 2\theta}{2}\left(\frac{1}{\tan\theta}-\tan (\theta -\phi_\sigma )\right)
=\frac{t}{1-t^2}\left(\frac{1}{t}-\frac{t-u}{1+tu}\right)\\
&=\frac{1}{1+tu}\left( 1+\frac{2tu}{1-t^2}\right)
=\frac{1}{1+tu}\left\{ 1+tu\left( 1+\frac{1+t^2}{1-t^2}\right)\right\}\\
&=1+\frac{tu}{1+tu}\cdot\frac{1+t^2}{1-t^2}
=1+\left( 1-\frac{1}{1+\tan\theta\tan\phi_\sigma}\right)\frac{1}{\cos 2\theta},
\end{align*}
which gives the equality in $\eqref{ineq:sd_EN_2}$.
Since $\tan\theta\tan\phi_\sigma$ is monotonically increasing with respect to both $\theta$ and $\phi_\sigma$ in the range $\eqref{range_theta_phi}$, we have
\begin{equation*}
0\leqslant\tan\theta\tan\phi_\sigma <\tan\theta\tan (\pi /2-\theta )=1,
\end{equation*}
which gives the inequality in $\eqref{ineq:sd_EN_2}$.
\end{proof}
It follows from Lemma $\ref{lem:calc_sd_EN_2}$ that $sd(E_\sigma^{(n)},N_\sigma^{(n)})$ is bounded from below as
\begin{equation}\label{bound_sd}
sd(E_\sigma^{(n)},N_\sigma^{(n)})>q_{n-1}-\frac{p_n}{2\cos (\gamma /2)}.
\end{equation}
If $\phi_\sigma =0$, then the RHS of $\eqref{existence_G}$ vanishes, so that no $G_\sigma^{(n)}$ and ${G'_\sigma}^{\! (n-1)}$ exist.
The assertion of the proposition in the case of equal division is derived by setting $p_n=1,q_n=n$ and 
considering the condition that the RHS of $\eqref{bound_sd}$ is positive or zero for $n=2,3$.
This completes the proof of Proposition $\ref{prop:existence_G}$.
\end{proof}
In the case of equal division, we show in Table $\ref{tbl:n_G}$ the smallest number $n_G$ 
that makes the RHS of $\eqref{bound_sd}$ positive or zero for various values of $\gamma$.
In other words, $n_G$ is the smallest number such that
if $G_\sigma^{(n)}$ and ${G'_\sigma}^{\! (n-1)}$ exist for $n\geqslant 2$, then $n$ satisfies $n\leqslant n_G$ for any $\beta_\sigma$.
\addtocounter{theorem}{1}
\begin{table}[h]
\begin{tabular}{c|ccc|cc|cc|cc|cc|cc|c}
$\gamma\backslash n_G$&\multicolumn{3}{c|}{none}&\multicolumn{2}{c|}{$2$}&\multicolumn{2}{c|}{$3$}&\multicolumn{2}{c|}{$4$}&\multicolumn{2}{c|}{$5$}&$\cdots$\\ 
\hline
$\gamma (\text{rad})$&$0$&$\cdots$&$2\pi /3$&$\cdots$&$2.6362$&$\cdots$&$2.8067$&$\cdots$&$2.8909$&$\cdots$&$2.9413$&$\cdots$\\
$\gamma ({}^\circ )\phantom{ad}$&$0$&$\cdots$&$120$&$\cdots$&$151.04$&$\cdots$&$160.81$&$\cdots$&$165.64$&$\cdots$&$168.52$&$\cdots$\\
$\cos (\gamma /2)$&$1$&$\cdots$&$1/2$&$\cdots$&$1/4$&$\cdots$&$1/6$&$\cdots$&$1/8$&$\cdots$&$1/10$&$\cdots$
\end{tabular}
\caption{The smallest number $n_G$ for various $\gamma$ such that if $G^{(n)}$ and ${G'}^{(n-1)}$ exist for $n\geqslant 2$, then $n\leqslant n_G$
for any $\beta_\sigma$ in the case of equal division}
\label{tbl:n_G}
\end{table}
\begin{remark}\rm
We chose $\alpha =48^\circ ,\beta_\Lt =60^\circ ,\beta_\Rt =96^\circ$ and $\gamma =156^\circ$ in Figure $\ref{fig:division_development}$
so that all cases in Table $\ref{tbl:assignment_new_division}$ appear in Figure $\ref{fig:division_new}$. 
Indeed, we see from Proposition $\ref{prop:existence_D}$ and Table $\ref{tbl:n_D}$ that $D^{(n)}$ exists only for $n=2$.
Also, we have
\begin{equation*}
\phi_\Lt =\beta_\Lt +\frac{\gamma}{4}-\frac{\pi}{2}=9^\circ ,\quad\phi_\Rt =\beta_\Rt +\frac{\gamma}{4}-\frac{\pi}{2}=45^\circ ,
\end{equation*}
for which the RHS of $\eqref{existence_G}$ in Proposition $\ref{prop:existence_G}$ becomes
\begin{align*}
1+\left( 1-\frac{1}{1+\tan 39^\circ \tan\phi_\sigma}\right)\frac{1}{\cos 78^\circ}\approx
\begin{cases}
1.5468&\text{for }\sigma =\Lt ,\\
3.1521&\text{for }\sigma =\Rt .
\end{cases}
\end{align*}
Thus $G_\Lt^{(n)}$ exists for $n=1,2$, while no $G_\Rt^{(n)}$ exists.
\end{remark}
\begin{proposition}\label{prop:K=M}
Suppose $\beta_\sigma +\gamma /4=\pi /2$ in Construction $\ref{const:division_new}$. 
If $D^{(n)}$ and ${D'}^{(n-1)}$ exist in $(8)$, then $K_\sigma^{(n)}$ given in $(9)$ and $M_\sigma^{(n)}$ given in $(10)$ are identical points.
\end{proposition}
\begin{proof}
Recall that $D^{(n)}K_\sigma^{(n)}$ is parallel to $D^{(1)}E_\sigma^{(1)}$, 
while $D^{(n)}M_\sigma^{(n)}$ is contained in the reflection of ${k'_\sigma}^{\! (n)}$ across $A^{(n)}E_\sigma^{(n)}$.
Also recall that $D^{(1)}G_\sigma^{(1)}$ is contained in the reflection of ${k'_\sigma}^{\! (1)}$.
Since ${k'_\sigma}^{(n)}, n=1,\dots ,d$ are parallel to each other, $D^{(n)}M_\sigma^{(n)}$ is parallel to $D^{(1)}G_\sigma^{(1)}$.
But it follows from $\beta_\sigma +\gamma /4=\pi /2$ that $E_\sigma^{(1)}=G_\sigma^{(1)}$, 
and therefore both $D^{(n)}K_\sigma^{(n)}$ and $D^{(n)}M_\sigma^{(n)}$ are parallel to $D^{(1)}E_\sigma^{(1)}$, 
which implies that $K_\sigma^{(n)}=M_\sigma^{(n)}$. 
This completes the proof of Proposition $\ref{prop:K=M}$.
\end{proof}
Thus from Propositions $\ref{prop:existence_D}$, $\ref{prop:existence_G}$ and $\ref{prop:K=M}$, we obtain Table $\ref{tbl:assignment_new_division_phi=0}$ 
which gives the assignment of mountain folds and valley folds in the case of $\beta_\sigma +\gamma /4 =\pi /2$
in the resulting crease pattern in Construction $\ref{const:division_new}$.
\addtocounter{theorem}{1}
\begin{table}[h]
\begin{tabular}{c|c|c|c}
&common&\multicolumn{2}{c}{$n$ with $2\leqslant n\leqslant d$ such that}\\
\cline{3-4}
&creases&$\exists D^{(n)}$&$\nexists D^{(n)}$\\
\hline
mountain&$j_\sigma ,\ell_\sigma^{(n)},A^{(1)}D^{(1)},$&\multicolumn{2}{c}{$A^{(n-1)}E_\sigma^{(n)},B_\sigma^{(n)}J_\sigma^{(n)}$}\\
\cline{3-4}
folds&$B_\sigma^{(1)}E_\sigma^{(1)},D^{(1)}E_\sigma^{(1)}$&$A^{(n)}D^{(n)},D^{(n)}K_\sigma^{(n)}$&$A^{(n)}F^{(n)}$\\
\hline
valley&$k_\sigma ,m_\sigma^{(n)},$&\multicolumn{2}{c}{$B_\sigma^{(n-1)}J_\sigma^{(n)}$}\\
\cline{3-4}
folds&$A^{(n)}E_\sigma^{(n)},E_\Lt^{(n)}E_\Rt^{(n)}$&$A^{(n-1)}{D'}^{(n-1)},{D'}^{(n)}K_\sigma^{(n)}$&$A^{(n-1)}F^{(n)}$\\
\end{tabular}\vspace{0.5cm}
\caption{Assignment of mountain folds and valley folds for the division of the new gadget for $\beta_\sigma +\gamma /4=\pi /2$}
\label{tbl:assignment_new_division_phi=0}
\end{table}
\begin{remark}\rm
We may be able to assign another combination of of mountain folds and valley folds to the gadgets
corresponding to the folding order of the gadgets.
Note that in Construction $\ref{const:division_new}$, if we assign mountain folds and valley folds  
to the extrusion as in Table $\ref{tbl:assignment_new_division}$ or $\ref{tbl:assignment_new_division_phi=0}$, 
then the layer of each gadget lies under that of the subsequent gadget. 
For $n=2,\dots d$, if points $D^{(n)}$ and ${D'}^{(n-1)}$ do not exist, then we can exchange the order of 
the layer including $A^{(n)}F^{(n)}$ in the $n$-th gadget and the layer including $A^{(n-1)}F^{(n-1)}$ or $A^{(n-1)}D^{(n-1)}$ in the $(n-1)$-th gadget.
Correspondingly, creases $A^{(n-1)}E_\Lt^{(n)},A^{(n-1)}E_\Rt^{(n)}$ change to valley folds
and $E_\Lt^{(n)}E_\Rt^{(n)}$ to a mountain fold
from the assignment given in Table $\ref{tbl:assignment_new_division}$ or $\ref{tbl:assignment_new_division_phi=0}$,
while the assignment to the $(n-1)$-th gadget is unchanged.
This exchange of the layers makes the appearance of the ridge better, but the folding becomes a little more difficult in exchange.
\end{remark}
\begin{remark}\rm
More or less surprisingly, upper gadgets do not interfere with adjacent gadgets 
in the sense that we can widen the outgoing pleats as we like at least theoretically,
although we have to fold back their pleats appropriately so that they do not collide with adjacent pleats.
Therefore, only the sizes of the lowest gadgets are restricted by their interference coefficients.
As an example, we give in Figrure $\ref{fig:widening_pleats}$
two crease patterns of an extruded square prism five times as high as a cube with the same base,
where we choose $d=2,p_1=2/5,p_2=8/5$ in the crease pattern on the left side, and $d=3,p_1=3/5,p_2=p_3=6/5$ on the right side.
Observe that each of the double pleats of the upper gadgets in the crease pattern on the left side behaves as a single pleat with double thickness,
and does not engage with other pleats in between its layers, 
while all layers of the pleats in the crease pattern on the right side are independent and engage with other pleats.
As another example, we can extrude a square pyramid with a sharp end which cannot be extruded without repetition of new gadgets,
whose crease pattern is given in Figure $\ref{fig:sharp_pyramid}$.
\end{remark}
\begin{figure}[htbp]
  \begin{center}
\addtocounter{theorem}{1}
          \includegraphics[width=\hsize]{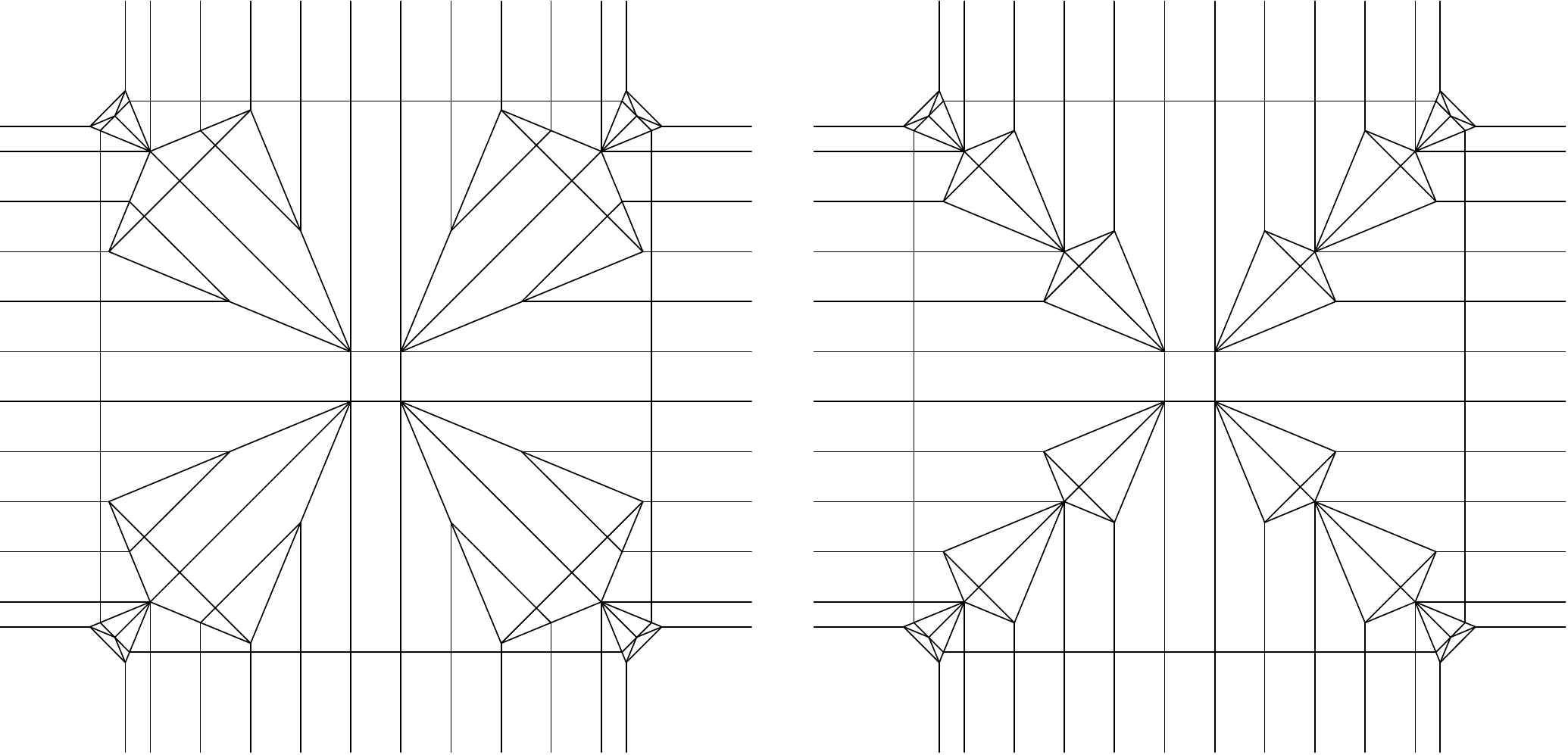}
    \caption{CPs of an extruded square prism of dimensions $1\times 1\times 5$}
    \label{fig:widening_pleats}
\end{center}
\end{figure}
\begin{figure}[htbp]
  \begin{center}
\addtocounter{theorem}{1}
          \includegraphics[width=0.8\hsize]{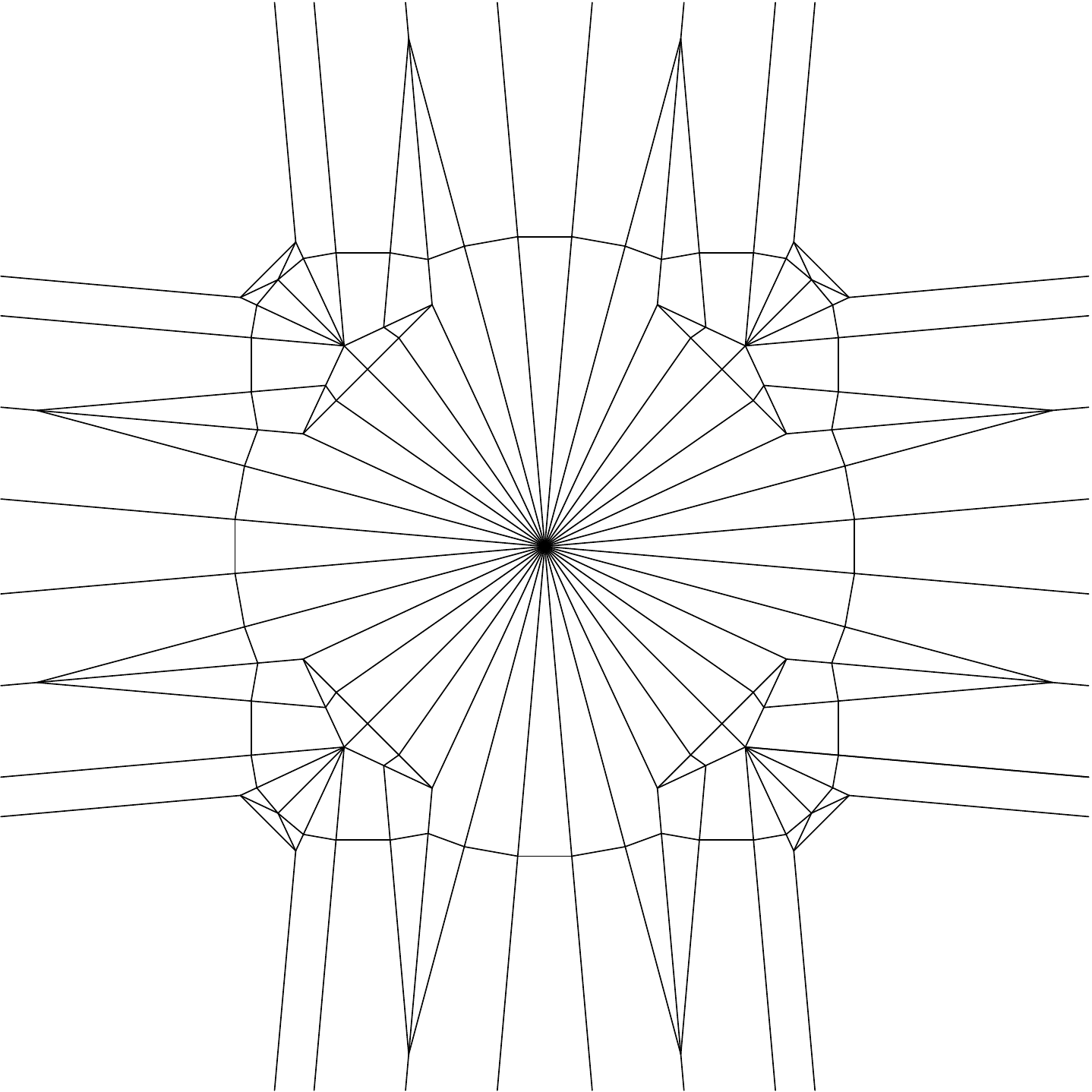}
    \caption{CP of an extruded square pyramid with a sharp end, where each side triangular face has an apex angle of $10^\circ$}
    \label{fig:sharp_pyramid}
\end{center}
\end{figure}

We end this section with the proportional division of a conventional gadget into $d$ gadgets. 
We use the same development as in the case of our new gadgets for comparison.
\begin{construction}\label{const:division_conv}\rm
Consider a development as shown in Figure $\ref{fig:division_development}$, for which we require the following conditions.
\begin{enumerate}[(i)]
\item $\alpha <\beta_\Lt+ \beta_\Rt$, $\beta_\Lt <\alpha+ \beta_\Rt$ and $\beta_\Rt <\alpha+ \beta_\Lt$.
\item $\alpha +\beta_\Lt +\beta_\Rt <2\pi$.
\item $\alpha +\beta_\Lt +\beta_\Rt >\pi$.
\end{enumerate}
The crease pattern of the proportional division of the conventional $3$D gadget into $d$ gadgets
in the ratio $p_1:\dots :p_d$ from the bottom 
is constructed as follows, where all procedures are done for both $\sigma =\Lt ,\Rt$.
\begin{enumerate}
\item Draw a perpendicular to $\ell_\sigma$ through $B_\sigma$ for both $\sigma =\Lt ,\Rt$, letting $C$ be the intersection point.
\item Divide segment $CA$ into $d$ parts proportionally in the ratio $p_1:\dots :p_d$, 
letting $A^{(1)},\dots ,A^{(d-1)}$ to be the divided points in order from the side of $C$.
\item For $n=1,\dots ,d-1$, draw a ray $\ell_\sigma^{(n)}$ parallel to $\ell_\sigma$, starting from $A^{(n)}$ and going in the same direction as $\ell_\sigma$,
letting $B_\sigma^{(n)}$ be the intersection point with $B_\sigma C$.
\item For $n=1,\dots ,d-1$, draw a ray $k_\sigma^{(n)}$ parallel to $k_\sigma$ from $B_\sigma^{(n)}$ to the side of $\ell_\sigma^{(n+1)}$,
where we set $\ell_\sigma^{(d)}=\ell_\sigma$.
Also, draw a ray ${k'_\sigma}^{\! (n)}$ starting from $B_\sigma^{(n)}$ 
which is a reflection of $k_\sigma^{(n)}$ across $\ell_\sigma^{(n)}$, where we set $k_\sigma^{(d)}=k_\sigma$ and $B_\sigma^{(d)}=B_\sigma$.
\item For $n=2,\dots ,d$, let $J_\sigma^{(n)}$ be the intersection point of $k_\sigma^{(n-1)}$ and ${k'_\sigma}^{(n)}$.
\item Draw a perpendicular bisector $m_\sigma^{(1)}$ of $B_\sigma^{(1)}C$, letting $D_\sigma^{(1)}$ be the intersection point with ${k'_\sigma}^{(1)}$.
Also, restrict $m_\sigma^{(1)}$ to the ray starting from $D_\sigma^{(1)}$ and going to the same direction as $\ell_\sigma$.
\item For $n=2,\dots ,d$, draw a perpendicular bisector $m_\sigma^{(n)}$ of $B_\sigma^{(n-1)}B_\sigma^{(n)}$ 
and a line through $A^{(n)}$ parallel to $A^{(1)}D_\sigma^{(1)}$, letting $D_\sigma^{(n)}$ be the intersection point of the two lines.
Also, restrict $m_\sigma^{(n)}$ to the ray starting from $D_\sigma^{(n)}$ and going to the same direction as $\ell_\sigma$, 
which passes through $J_\sigma^{(n)}$.
\item The desired crease pattern is shown as the solid lines in Figure $\ref{fig:division_conv}$, 
and the assignment of mountain folds and valley folds is given in Table $\ref{tbl:assignment_conv_division}$.
\end{enumerate}
\end{construction}
\begin{figure}[htbp]
  \begin{center}
\addtocounter{theorem}{1}
          \includegraphics[width=\hsize]{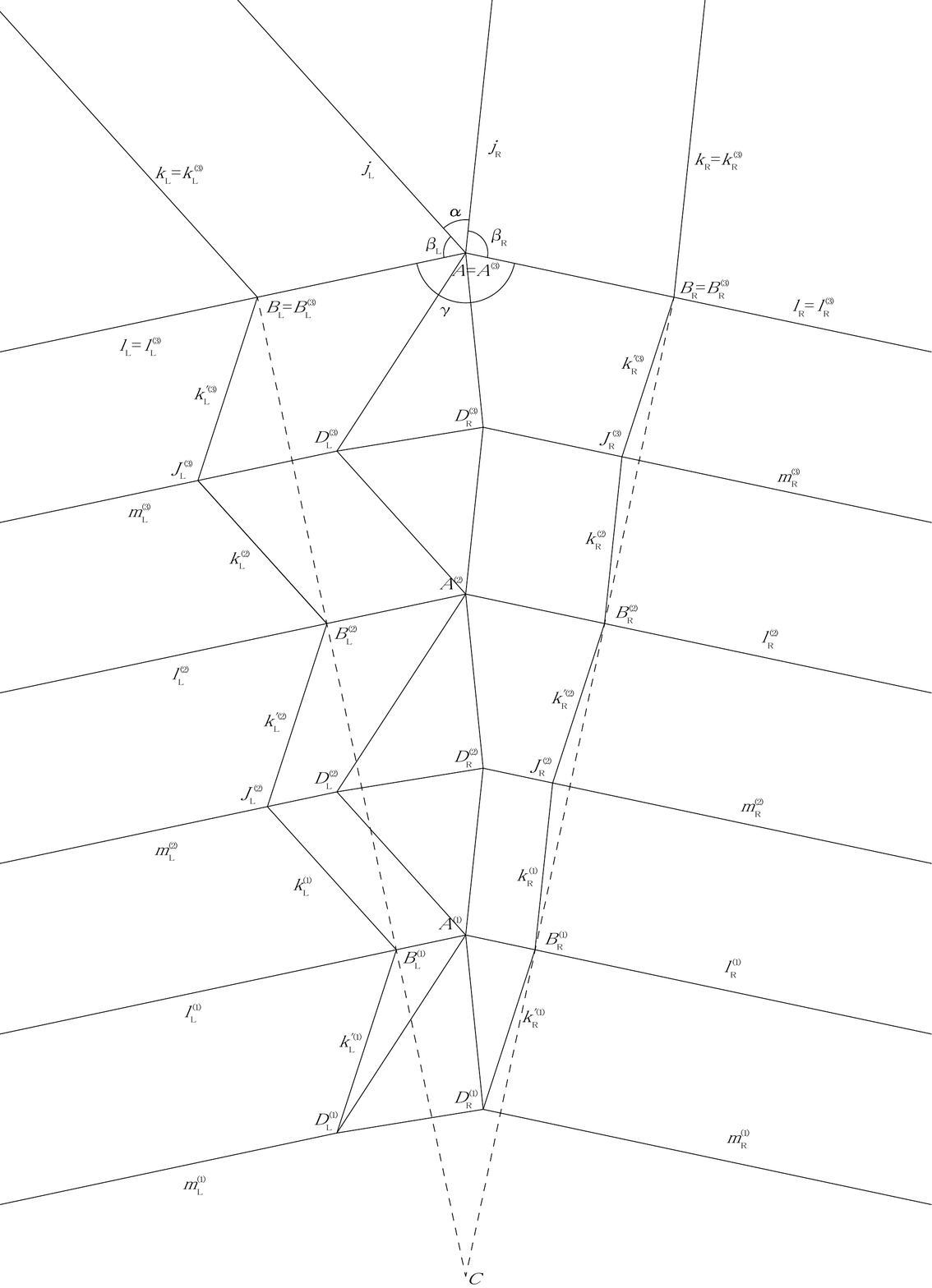}
    \caption{CP of the division of the conventional gadget}
    \label{fig:division_conv}
\end{center}
\end{figure}
\addtocounter{theorem}{1}
\begin{table}[h]
\begin{tabular}{r|c|c}
&common creases&$n$ with $2\leqslant n\leqslant d$\\ \hline
mountain folds&$j_\sigma ,\ell_\sigma^{(n)},A^{(n)}B_\sigma^{(n)},B_\sigma^{(1)}D_\sigma^{(1)}$
&$A^{(n-1)}D_\sigma^{(n)},B_\sigma^{(n)}J_\sigma^{(n)}$\\ \hline
valley folds&$k_\sigma ,m_\sigma^{(n)},A^{(n)}D_\sigma^{(n)},D_\Lt^{(n)}D_\Rt^{(n)}$
&$B_\sigma^{(n-1)}J_\sigma^{(n)}$\\
\end{tabular}
\caption{Assignment of mountain folds and valley folds for the division of the conventional gadget}
\label{tbl:assignment_conv_division}
\end{table}
\section{Remarks and further applications}\label{sec:9}
In this section, we give some remarks on the constructions of the new $3$D gadgets.
We also briefly review further applications of the new gadgets, some of which will be studied in more detail in a sequel to this paper.\\

\noindent$\bullet$ {\it Modification of the crease pattern.}
Here consider the resulting crease pattern in Construction $\ref{const:new}$ 
and let us forget the crease patten on the side of $\sigma'$, i.e., the other side of $\sigma$.
If $\alpha_\sigma >0$, then we can duplicate the crease patten on the side of $\sigma$ by reflection across $AC$, to obtain a new $3$D gadget.
There is another way of constructing a new gadget. 
Draw a circle with center $A$ through $B_\sigma$.
Draw a tangent line to the circle to the side of $\sigma'$, letting $T_{\sigma'}$ to be the tangency point.
Choose a point $B_{\sigma'}$ on the circular arc $DT_{\sigma'}$ excluding $D$ so that the region on the side of $\sigma$ bounded 
by the bisector of $\angle B_\Lt A B_\Rt$ includes $j_\sigma\setminus\{A\}$ properly, 
which is necessary for condition $\mathrm{(i)}$ of Construction $\ref{const:new}$ to hold. 
In other words, we choose $B_{\sigma'}$ on the circular arc $DT_{\sigma'}$ excluding $D$ so that $\gamma_{\sigma'}=\angle DAB_{\sigma'}$ satisfies
$(\gamma_{\sigma'}-\gamma_\sigma )/2<\alpha_\sigma$. 
Then $\ell_{\sigma'}$ is determined as a ray starting from $B_{\sigma'}$ and perpendicular to $CB_{\sigma'}$.
If we choose $\beta_{\sigma'}$ so that $(\pi -\gamma_{\sigma'})/2\leqslant\beta_{\sigma'}<\pi -(\gamma_\sigma +\gamma_{\sigma'})/2$, 
then conditions $\mathrm{(i)}$--$\mathrm{(iv)}$ of Construction $\ref{const:new}$ hold, and thus we obtain another $3$D gadget.\\

\noindent$\bullet$ {\it Insertion of a face.}
Consider the resulting crease pattern in Construction $\ref{const:new}$. 
Then we can insert a rectangular face of any width between ridge $AB$
as in Figure $\ref{fig:insertion}$, where the newly added creases $A_\Lt A_\Rt$ and $D_\Lt D_\Rt$ are mountain folds. 
When we insert a face between the ridge of the new cube gadget, we reproduce a gadget created by Natan in \cite{Natan}.
Using this operation, we can extrude a prism of a convex polygon with $2N(>4)$ sides with $N$ new gadgets.
Also, we can repeat each new gadget to make the extrusion higher, although there may appear `horizontal streaks' on inserted faces. 
As an example,
we show in Figure $\ref{fig:hexagonal_prism}$ the crease pattern of an extruded hexagonal prism made by inserting three side faces to a triangular prism. \\
\begin{figure}[htbp]
  \begin{center}
\addtocounter{theorem}{1}
          \includegraphics[width=0.8\hsize]{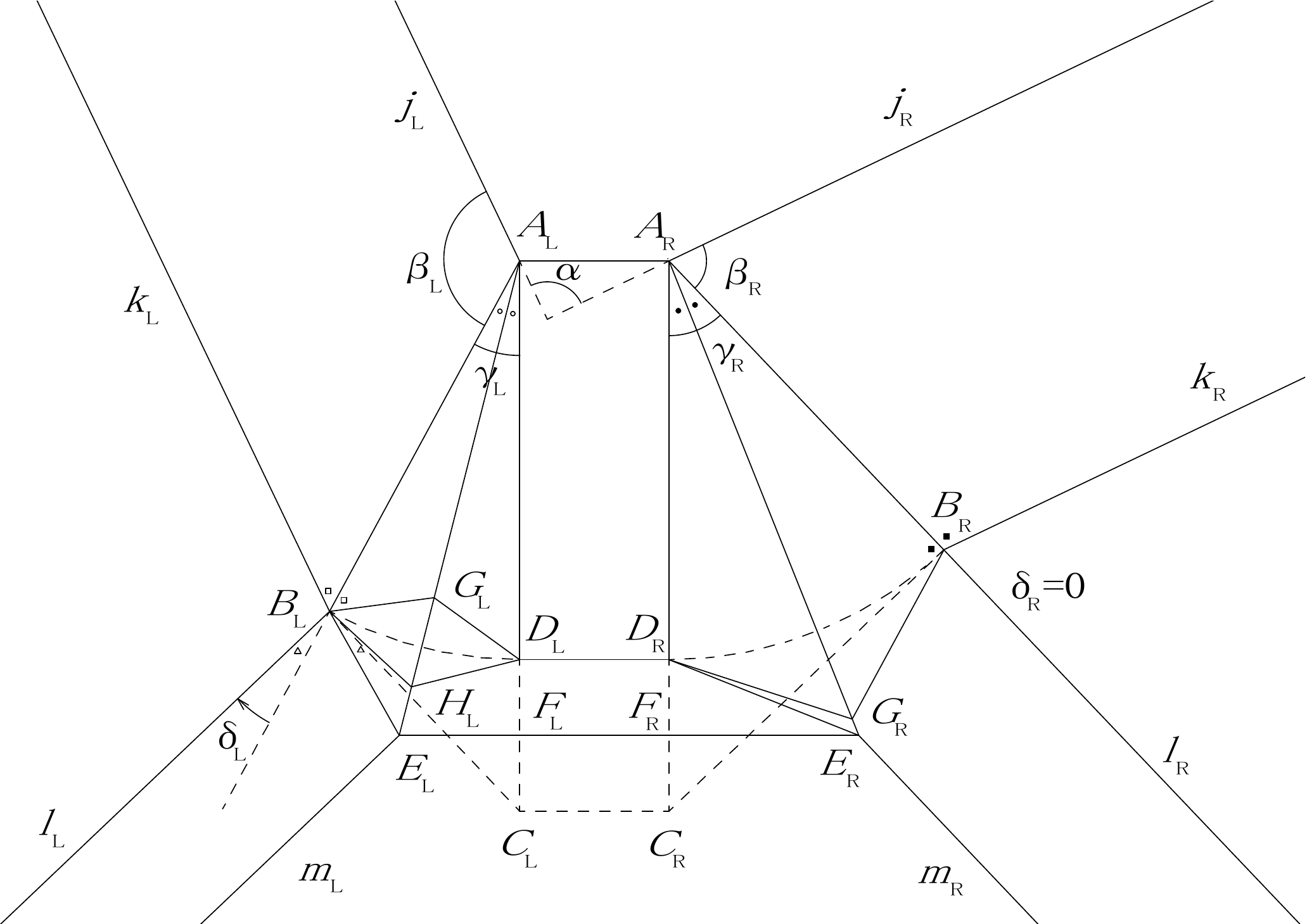}
    \caption{Insertion of a face between the ridge in a new gadget}
    \label{fig:insertion}
\end{center}
\end{figure}
\begin{figure}[htbp]
  \begin{center}
\addtocounter{theorem}{1}
          \includegraphics[width=0.8\hsize]{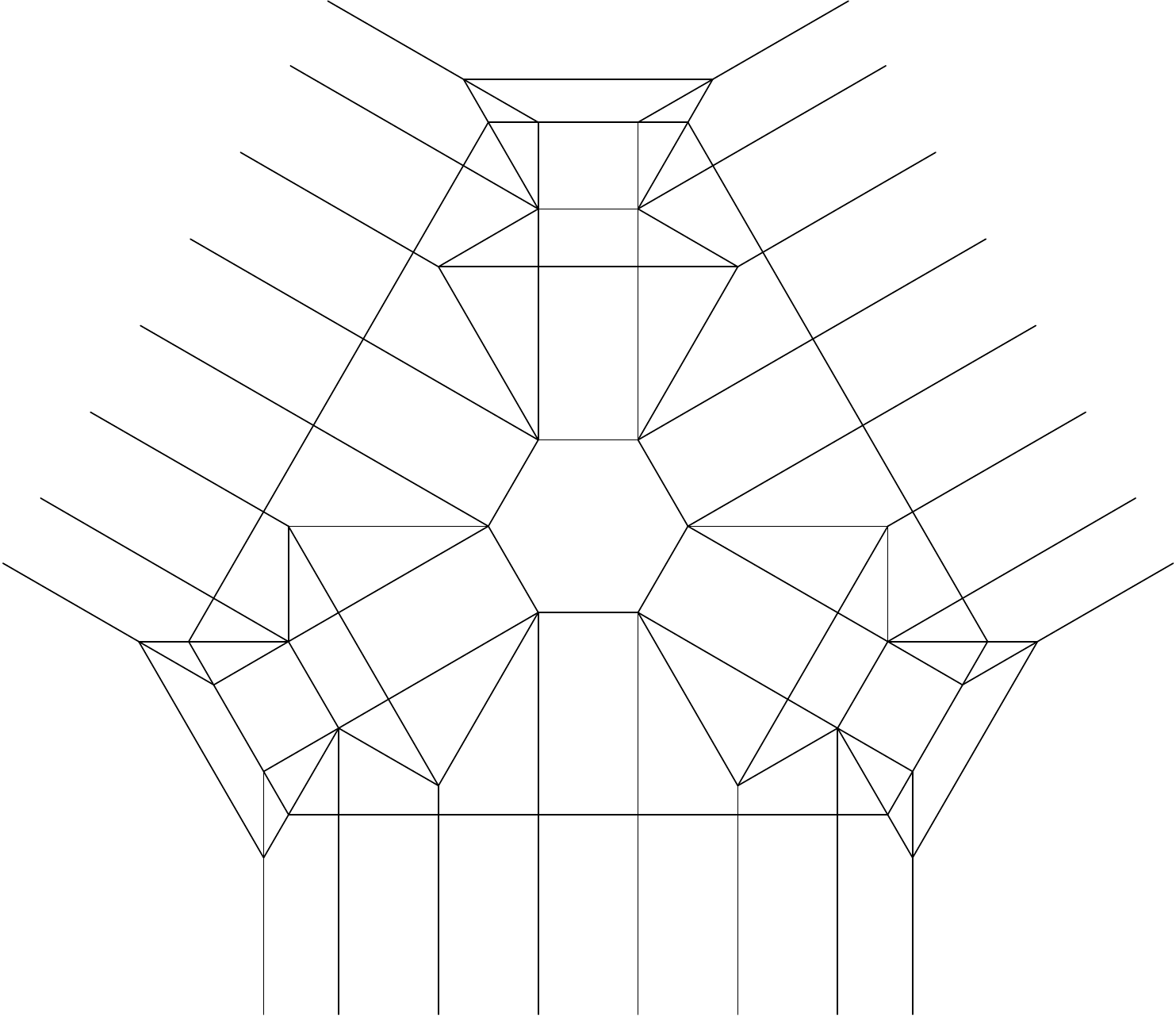}
    \caption{CP of an extruded hexagonal prism made from a triangular prism by inserting three side faces}
    \label{fig:hexagonal_prism}
\end{center}
\end{figure}

\noindent$\bullet$ {\it Negative $3$D gadgets.}
In the developments in Figures $\ref{fig:development_conv}$ and $\ref{fig:development_new}$ we have considered the face bounded by $j_\Lt ,j_\Rt$ 
with inner angle $\alpha$ as lying in $H_h=\set{z=h}$ for some $h>0$, and the face bounded by $k_\Lt ,k_\Rt$ as lying in $H_0$.
Alternatively, even if we suppose the faces bounded by $j_\Lt ,j_\Rt$ and $k_\Lt ,k_\Rt$ lie in $H_0$ and $H_h$ respectively,
the development is the same.
Thus we expect naturally that we can construct a `negative' $3$D gadget by modifying the resulting crease pattern in Construction $\ref{const:new}$.
Here we shall give two constructions of negative $3$D gadgets in the case $\delta_\Lt =\delta_\Rt =0$, which differ in scope of application.

The first construction applies to cases where $\gamma <2\pi /3$ holds and
$\beta_\sigma\leqslant\pi /2+\gamma /4$ holds for either $\sigma =\Lt$ or $\sigma =\Rt$.
We may assume $\beta_\Lt\leqslant\pi /2+\gamma /4$. 
Let $G'_\Rt$ be the intersection point of segment $AE_\Rt$ and the ray starting from $G_\Lt$ through $D$.
Also, let $P_\Rt$ be the intersection point of ray $m_\Rt$ and the ray starting from $E_\Lt$ through $C$.
These points $G'_\Rt$ and $P_\Rt$ exist because we have that
\begin{align*}
\angle ADG'_\Rt&=\beta_\Lt\leqslant\frac{\pi}{2}+\frac{\gamma}{4}=\angle ADE_\Rt ,\quad\text{and}\\
\angle P_\Rt E_\Lt E_\Rt +\angle P_\Rt E_\Rt E_\Lt&=2\angle DE_\Lt E_\Rt +\angle AB_\Rt E_\Rt\\
&=2\cdot\frac{\gamma}{4}+\left(\frac{\pi}{2}+\frac{\gamma}{4}\right) =\frac{\pi}{2}+\frac{3}{4}\gamma <\pi 
\end{align*}
by the above conditions.
Then the desired crease pattern is constructed as in Figure $\ref{fig:negative_gadget_1}$, 
and the assignment of mountain folds and valley folds is given in Table $\ref{tbl:assignment_negative_1}$,
where the subscript $\sigma$ is taken for both $\sigma =\Lt ,\Rt$.
Note that in this construction, $\triangle AE_\Lt E_\Rt$ is rotated to the side of $k_\Lt$ 
so that $\triangle AG_\Lt G'_\Rt$ overlap with the left side face.

\begin{figure}[htbp]
  \begin{center}
\addtocounter{theorem}{1}
          \includegraphics[width=0.8\hsize]{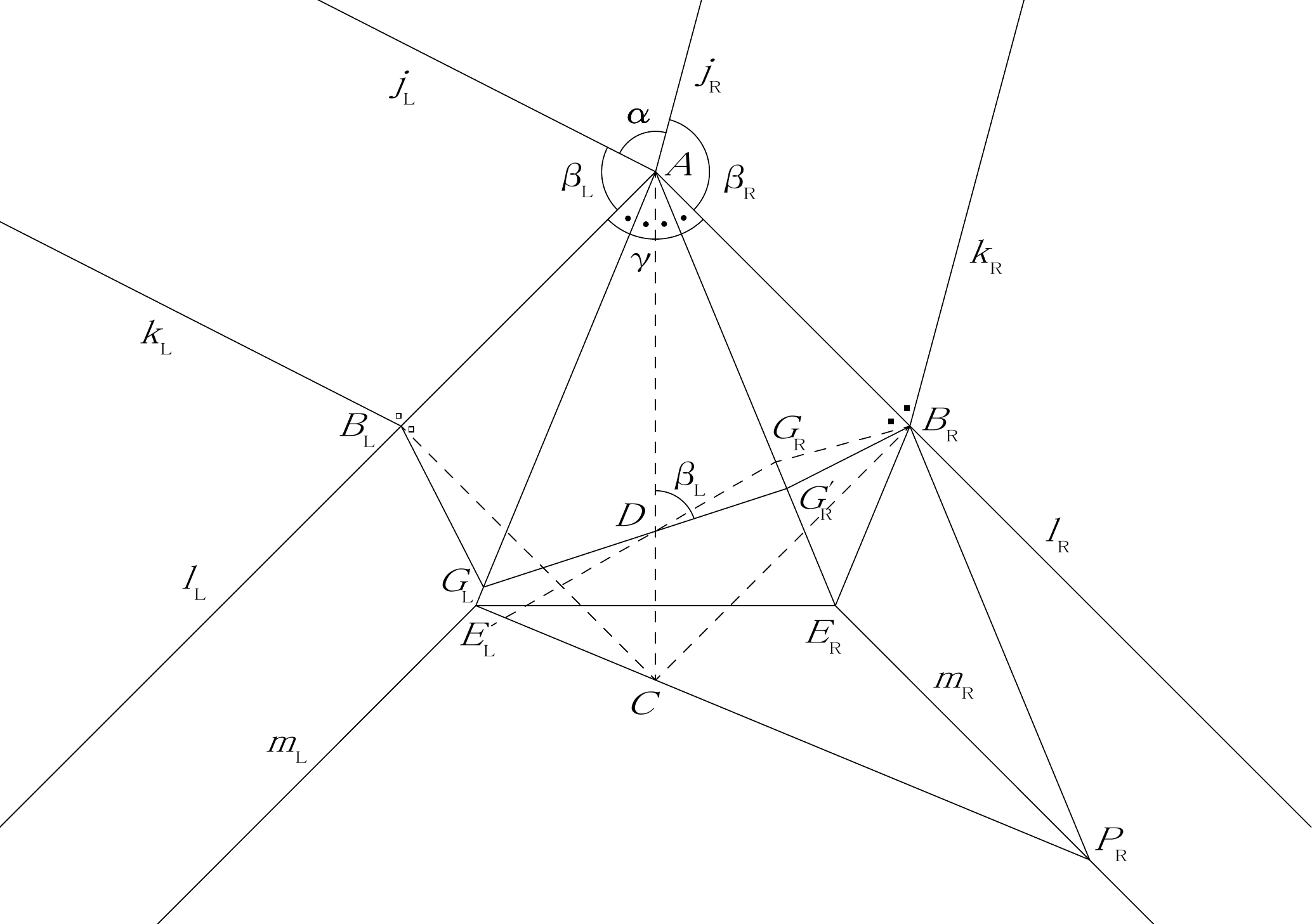}
    \caption{CP of the negative $3$D gadget by the first construction}
    \label{fig:negative_gadget_1}
\end{center}
\end{figure}
\addtocounter{theorem}{1}
\begin{table}[h]
\begin{tabular}{c|c}
mountain folds&$k_\sigma ,\ell_\sigma ,AB_\sigma ,B_\Rt E_\Rt ,E_\Lt P_\Rt ,G_\Lt G'_\Rt$\\ \hline
valley folds&$j_\sigma ,m_\sigma ,AE_\sigma ,B_\Lt G_\Lt ,B_\Rt G'_\Rt ,B_\Rt P_\Rt ,E_\Lt E_\Rt$
\end{tabular}\vspace{0.5cm}
\caption{Assignment of mountain folds and valley folds to the negative gadget by the first construction}
\label{tbl:assignment_negative_1}
\end{table}
For the second construction, suppose either $\beta_\Lt ,\beta_\Rt\leqslant\pi /2$ or $\beta_\Lt ,\beta_\Rt\geqslant\pi /2$ holds. 
Let us consider the crease pattern in Figure $\ref{fig:wrong_negative_CP}$,
where we take $G'_\Lt ,G'_\Rt ,P_\Rt ,P_\Lt$ so that segments $G'_\Lt G'_\Rt$ and $P_\Lt P_\Rt$ are parallel to segment $E_\Lt E_\Rt$.
This is \emph{not} the correct crease pattern because $\triangle AG'_\Lt G'_\Rt$ does not face the correct direction.
\begin{figure}[htbp]
  \begin{center}
\addtocounter{theorem}{1}
          \includegraphics[width=0.8\hsize]{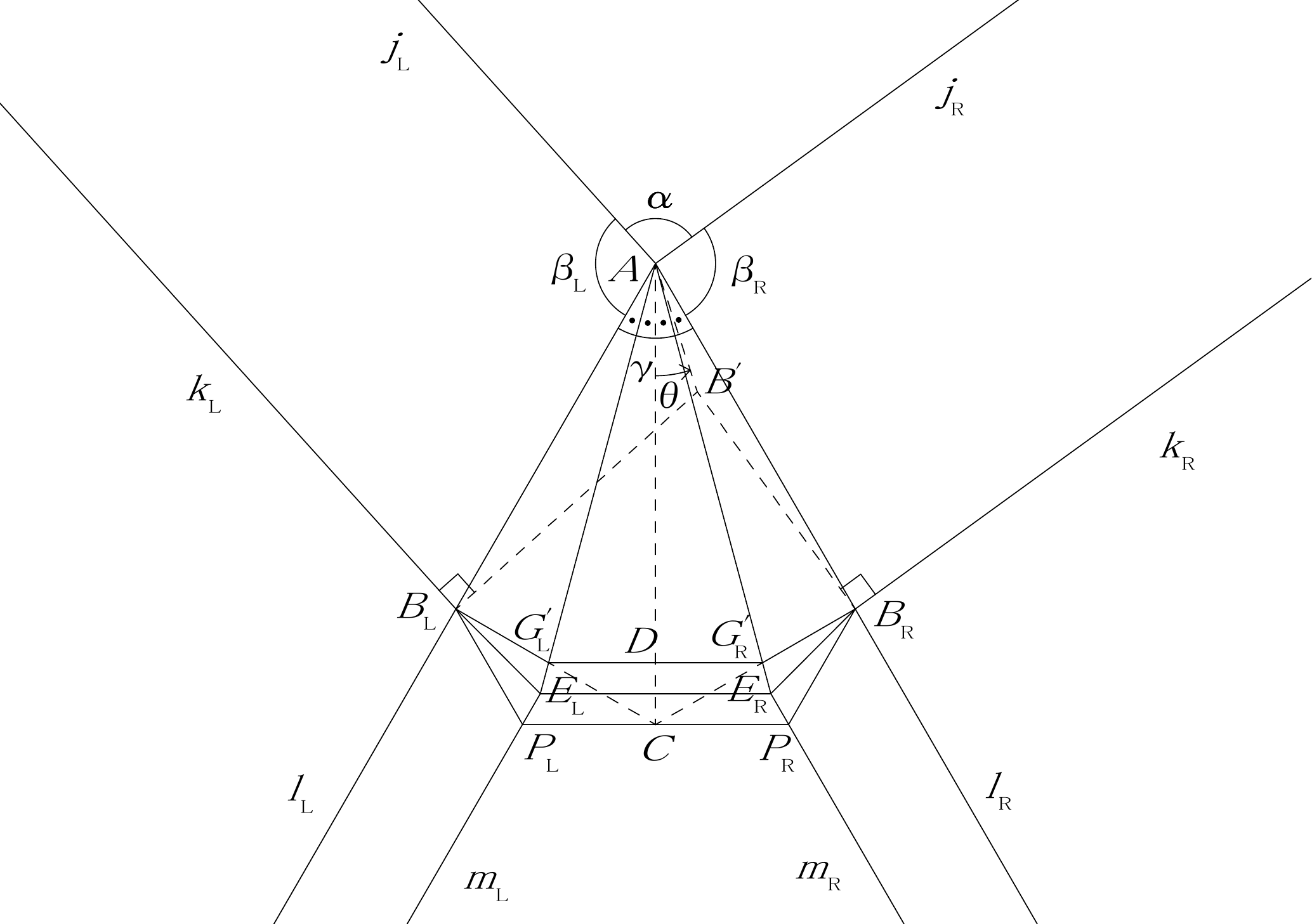}
    \caption{Wrong CP of the negative $3$D gadget which is to be revised}
    \label{fig:wrong_negative_CP}
\end{center}
\end{figure}
Let $B'$ be the intersection point of a perpendicular to $k_\Lt$ and a perpendicular to $k_\Rt$ in the crease pattern.
Also, let $\theta$ be the angle such that if $AD$ rotates counterclockwise around $A$ by $\theta\in (-\pi/2 ,\pi /2]$, 
then $AD$ overlaps with an extension of $AB'$.
Then $\theta$ is calculated as
\begin{equation}\label{eq:theta}
\theta =\tan^{-1}\left\{\frac{\tan\beta_\Rt -\tan\beta_\Lt}{2+(\tan\beta_\Lt +\tan\beta_\Rt )/\tan (\gamma /2)}\right\} .
\end{equation}
Let us define $\proj_{A,0}$ to be the projection of the resulting extrusion to the crease pattern in the bottom plane $\set{z=0}$ with $\proj_{A,0}(A)=A$.
Then we have $\proj_{A,0}(B)=B'$.
Since $G'_\Lt G'_\Rt$ in the bottom plane $\set{z=0}$ must be perpendicular to $AB$ in the resulting extrusion, and also perpendicular to $AB'$,
we have to rotate $G'_\Lt G'_\Rt$ couterclockwise around $D$ by $\theta$.

However, the supremum angle $\theta_\Rt$ (resp. $\theta_\Lt$) by which we \emph{can} rotate 
$G'_\Lt G'_\Rt$ counterclockwise (resp. clockwise) is given by
\begin{equation}\label{eq:theta_sigma}
\theta_\sigma =\min\left\{2\angle E_\sigma CP_\sigma ,\angle P_\sigma B_\sigma \ell_\sigma ,2\angle E_{\sigma'}P_{\sigma'}C\right\}
=\min\left\{\frac{\gamma}{2},\beta_\sigma +\frac{\gamma}{2}-\frac{\pi}{2},\pi -\gamma\right\} ,
\end{equation}
where we exclude $\theta_\sigma =\pi -\gamma$ for which we have to fold back the pleat formed by $\ell_\sigma$ and $m_\sigma$.
Hence this construction is possible if $\theta\in [-\theta_\Lt ,\theta_\Rt ]\cap (\gamma -\pi ,\pi -\gamma )$, 
where $\theta$ and $\theta_\sigma$ are given by $\eqref{eq:theta}$ and $\eqref{eq:theta_sigma}$ respectively,
and the crease pattern is shown in Figure $\ref{fig:negative_gadget_2}$.
Also, the assignment of mountain folds and valley folds is given in Table $\ref{tbl:assignment_negative_2}$,
where the subscript $\sigma$ is taken for both $\sigma =\Lt ,\Rt$.
\begin{figure}[htbp]
  \begin{center}
\addtocounter{theorem}{1}
          \includegraphics[width=0.8\hsize]{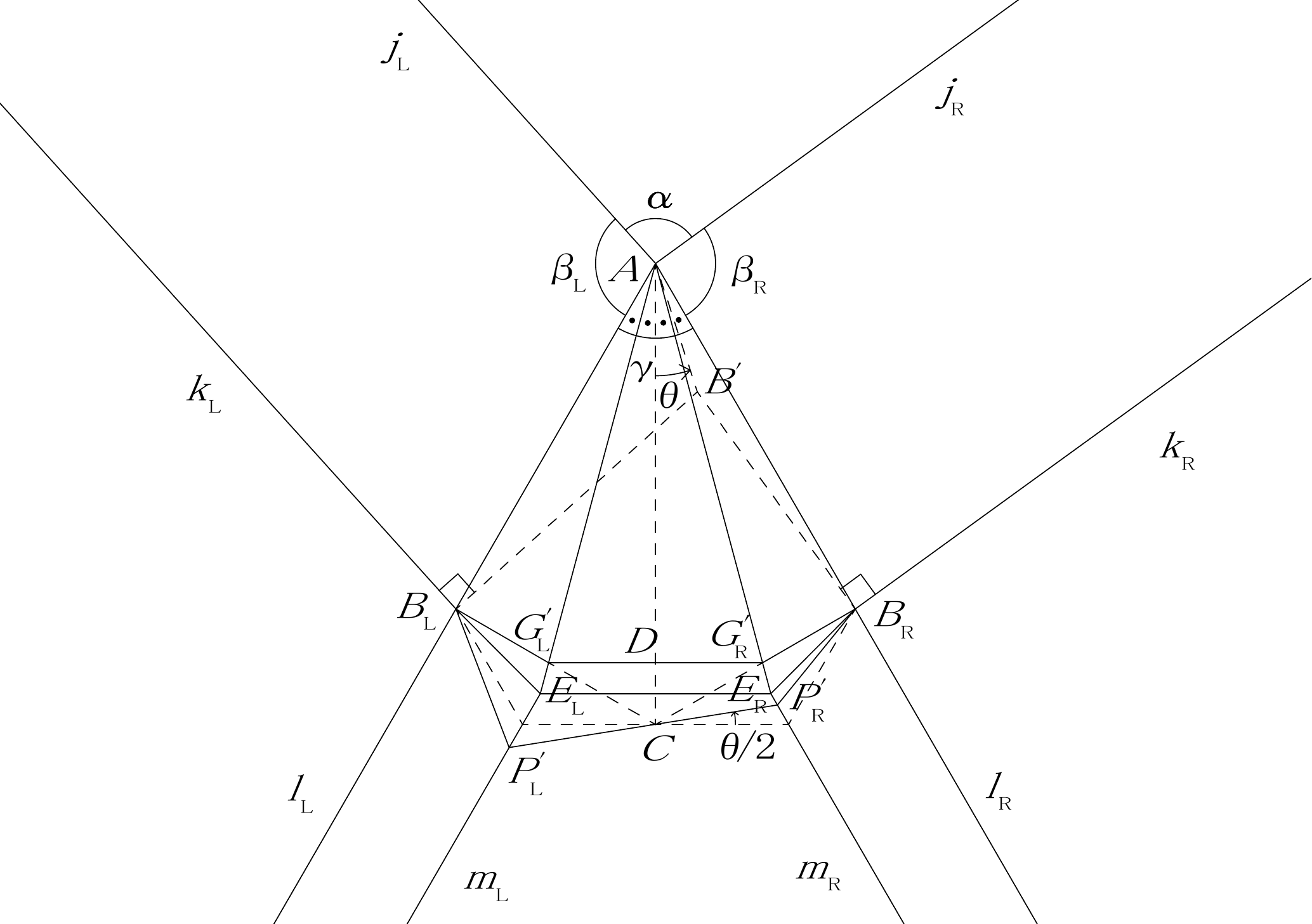}
    \caption{CP of the negative $3$D gadget by the second construction}
    \label{fig:negative_gadget_2}
\end{center}
\end{figure}
\addtocounter{theorem}{1}
\begin{table}[h]
\begin{tabular}{c|c}
mountain folds&$k_\sigma ,\ell_\sigma ,AB_\sigma ,B_\sigma E_\sigma (\text{for }\beta_\sigma\neq\pi /2),G'_\Lt G'_\Rt ,P'_\Lt P'_\Rt$\\ \hline
valley folds&$j_\sigma ,m_\sigma ,AE_\sigma ,B_\sigma G'_\sigma ,B_\sigma P'_\sigma (\text{for }\beta_\sigma\neq\pi /2),E_\Lt E_\Rt$
\end{tabular}\vspace{0.5cm}
\caption{Assignment of mountain folds and valley folds to the negative gadget by the second construction}
\label{tbl:assignment_negative_2}
\end{table}

\begin{proposition}\label{prop:appl_scope_negative_2}
Suppose either $\beta_\Lt ,\beta_\Rt\leqslant\pi /2$ or $\beta_\Lt ,\beta_\Rt\geqslant\pi /2$ holds.
If $\gamma <2\pi /3$, then the above second construction of negative $3$D gadgets is always possible.
\end{proposition}
\begin{proof}
We may assume $\beta_\Lt\leqslant\beta_\Rt$.
Also, we set $b_\sigma =\tan\beta_\sigma$ for $\sigma =\Lt ,\Rt$  and $c=\tan (\gamma /2)$. 
Note that $\beta_\Lt =\beta_\Rt =\pi /2$ we can choose any $\theta\in [-\theta_\Lt ,\theta_\Rt]\cap (\gamma -\pi ,\pi -\gamma )$
because $\ora{BA}=(0,0,1)$ is normal to any vector in the bottom plane.

First suppose $\beta_\Lt ,\beta_\Rt <\pi /2$.
Then we have $\theta\geqslant 0$ and $\theta_\Rt =\beta_\Rt +\gamma /2-\pi /2$.
Also, it follows from Lemma $\ref{lem:condition_gamma}$ that
\begin{equation*}
0<\pi -\beta_\Lt -\beta_\Rt <\gamma /2<\pi /2,
\end{equation*}
so that taking the tangents gives 
\begin{equation*}
0<\frac{b_\Lt +b_\Rt}{b_\Lt b_\Rt -1}<c.
\end{equation*}
Thus we see that $b_\Lt b_\Rt -1>0$ and 
\begin{equation}\label{ineq:b_L+b_R}
0<b_\Lt +b_\Rt <(b_\Lt b_\Rt -1)c.
\end{equation}
Hence we calculate as
\begin{align*}
\tan\theta_\Rt -\tan\theta&=\frac{b_\Rt c-1}{c+b_\Rt}-\frac{(b_\Rt -b_\Lt )c}{2c+b_\Lt +b_\Rt}\\
&=\frac{(b_\Lt +b_\Rt )(c^2-1)+2(b_\Lt b_\Rt -1)c}{(c+b_\Rt )(2c+b_\Lt +b_\Rt)}>\frac{(b_\Lt +b_\Rt )(c^2+1)}{(c+b_\Rt )(2c+b_\Lt +b_\Rt)}>0,
\end{align*}
where we used $\eqref{ineq:b_L+b_R}$ in the last line.
This gives $-\theta_\Lt <0\leqslant\theta <\theta_\Rt$.

Next suppose $\beta_\Lt ,\beta_\Rt >\pi /2$.
Then it follows from Lemma $\ref{lem:condition_gamma}$ that
\begin{equation*}
\pi /2<\beta_\Lt\leqslant\beta_\Rt <\pi -\gamma /2,
\end{equation*}
so that
\begin{equation*}
c+b_\Lt\leqslant c+b_\Rt <0,
\end{equation*}
which gives $\theta\leqslant 0$.
Thus we calculate as
\begin{equation*}
\tan\theta +\tan\theta_\Lt =\frac{(b_\Rt -b_\Lt )c}{2c+b_\Lt +b_\Rt}+c=\frac{2c(c+b_\Rt )}{2c+b_\Lt +b_\Rt}>0,
\end{equation*}
which gives $-\theta_\Lt <\theta\leqslant 0<\theta_\Rt$.

Lastly, suppose $\beta_\Lt <\pi /2=\beta_\Rt$ (resp. $\beta_\Lt =\pi /2<\beta_\Rt$).
Then we see from $\eqref{eq:theta}$ that $\theta =\gamma /2=\theta_\Rt$ (resp. $\theta =-\gamma /2=-\theta_\Lt$) as desired.
This completes the proof of Proposition $\ref{prop:appl_scope_negative_2}$.
\end{proof}
Note that if $\beta_\Lt =\pi /2$, $\beta_\Rt\neq\pi /2$ and $\gamma <2\pi /3$, the first and the second constructions give the same result.

However, the second construction of negative $3$D gadgets is not always possible for $\gamma\geqslant 2\pi /3$.
Indeed, if $\gamma\geqslant 2\pi /3$ and $\beta_\Lt <\pi /2=\beta_\Rt$ (resp. $\beta_\Lt =\pi /2<\beta_\Rt$),
then we have $\theta =\gamma /2$ (resp. $\theta =-\gamma /2$) and $\theta_\sigma =\pi -\gamma\leqslant\gamma /2$,
and consequently $\theta\notin [-\theta_\Lt ,\theta_\Rt ]\cap (\gamma -\pi ,\pi -\gamma )$.\\

\noindent$\bullet$ {\it Treatment of intersections of outgoing pleats.}
If $\beta_{i,\Rt}-\delta_{i,\Rt}+\beta_{i+1,\Lt}-\delta_{i+1,\Lt}<\pi /2$, then the two pairs of 
the simple outgoing pleats $\ell_{i,\Rt},m_{i,\Rt}$ and $\ell_{i+1,\Lt},m_{i+1,\Lt}$ intersect.
Although there are many ways of treating the outgoing pleats, let us describe a useful way of changing the outgoing pleats into one direction.

Let $P$ be the intersection point of $\ell =\ell_{i,\Rt}$ and $\ell'=\ell_{i+1,\Lt}$.
Let $m=m_{i,\Rt}$ and $m'=m_{i+1,\Rt}$, and let $Q=Q_{i,\Rt}$ (resp. $Q'=Q_{i+1,\Lt}$) be the intersection point of 
$m$ (resp. $m'$) and its perpendicular through $P$.
Also, let $R$ be the intersection point of $m$ and $m'$.
Without loss of generality, we may assume that $Q=(-1,0),Q'=(1,0)$ 
and $P=(x,y)$ lies in the range $\set{x\geqslant 0,y>0}$ by interchanging the left and the right sides if necessary.
Then the $x$-coordinate of $P$ satisfies one of the following: $\mathrm{(i)}$ $0\leqslant x<1$, $\mathrm{(ii)}$ $x=1$, or $\mathrm{(iii)}$ $x>1$.
Examples of the crease patterns of the outgoing pleats in these cases are given as solid lines in Figure $\ref{fig:outgoing_pleats}$,
where $s,s',t,t',r$ are the rays starting from $S,S',T,T',R$ respectively and perpendicular to $QQ'$, and
we can choose any line $u$ parallel to and between $s$ and $s'$ as a reference line by which $t$ (resp. $t'$) is determined 
as a bisector of $u$ and $t$ (resp. $u$ and $t'$).
Also, the assignment of mountain folds and valley folds in all cases is commonly given in Table $\ref{tbl:assignment_outgoing}$.
\begin{figure}[htbp]
  \begin{center}
\addtocounter{theorem}{1}
          \includegraphics[width=\hsize]{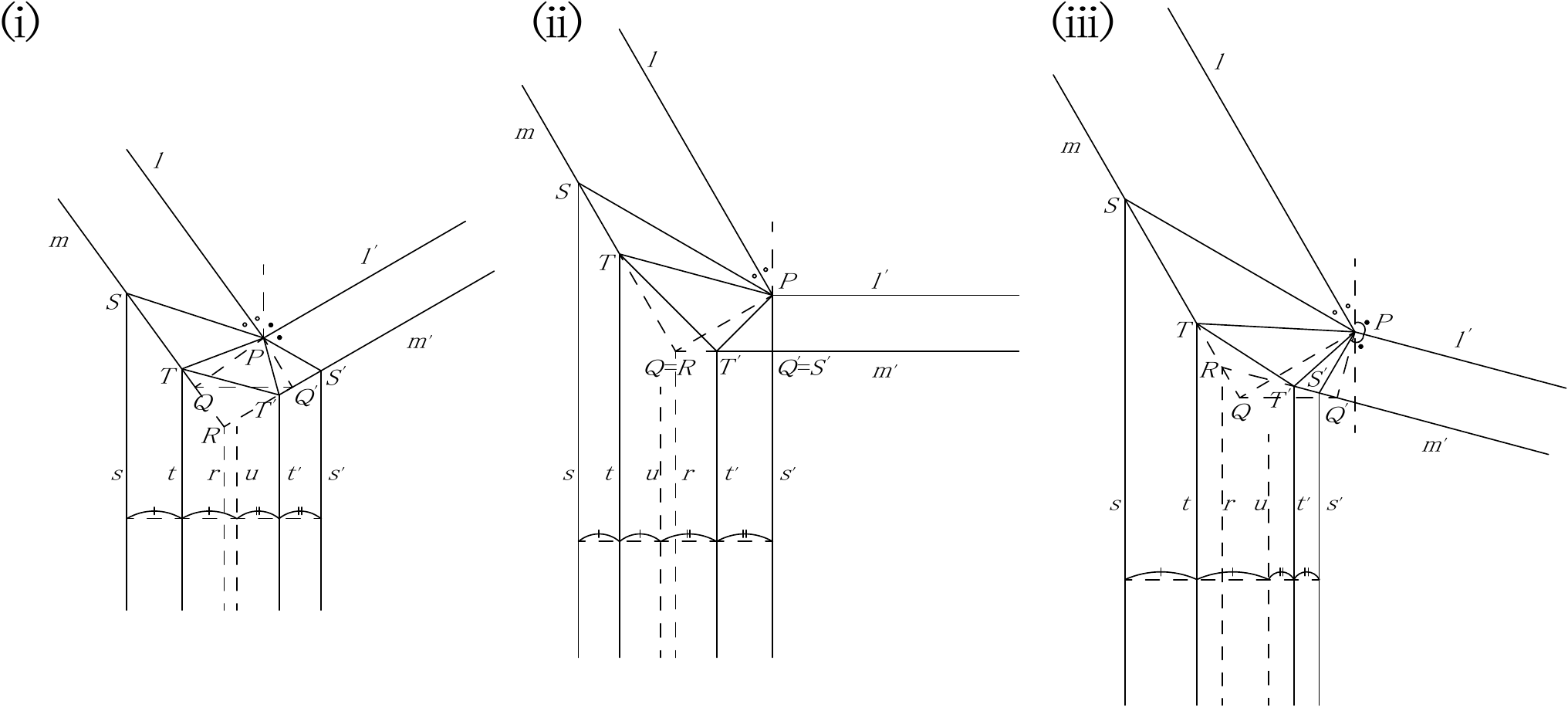}
    \caption{Treatment of outgoing pleats}
    \label{fig:outgoing_pleats}
\end{center}
\end{figure}
\addtocounter{theorem}{1}
\begin{table}[h]
\begin{tabular}{c|c}
mountain folds&$\ell ,\ell' ,s,s',PT,PT'$\\ \hline
valley folds&$m$ (incl. $ST$)$,m'$ (incl. $S'T'$)$,t,t',PS,PS',TT'$
\end{tabular}\vspace{0.5cm}
\caption{Assignment of mountain folds and valley folds to outgoing pleats}
\label{tbl:assignment_outgoing}
\end{table}

Note that the sum of the widths of the outgoing pleats is given by $\norm{QQ'}$ in all cases, 
and the distance between $s$ and $s'$ is given by $2\norm{QQ'}$.
Note also that the ray $r$ starting from $R$ bisects $s$ and $s'$. 
Thus instead of using reflections across $\ell$ and $\ell'$,
we can find $s$ and $s'$ as the lines parallel to $r$ at a distance of $\norm{QQ'}$, and then find $S$ and $S'$ as their intersection points
with $m$ and $m'$ respectively. 

One advantage of our treatment of outgoing pleats is that 
we can shift the positions of intersecting pleats as in Figure $\ref{fig:shift_outgoing}$,
which is useful when we make origami $3$D tessellations.
As another advantage, the intersecting pleats tuck the altered pleats inside so as to prevent them from fluttering easily, 
leaving only Y-shaped streaks on the front side.\\
\begin{figure}[htbp]
  \begin{center}
\addtocounter{theorem}{1}
          \includegraphics[width=\hsize]{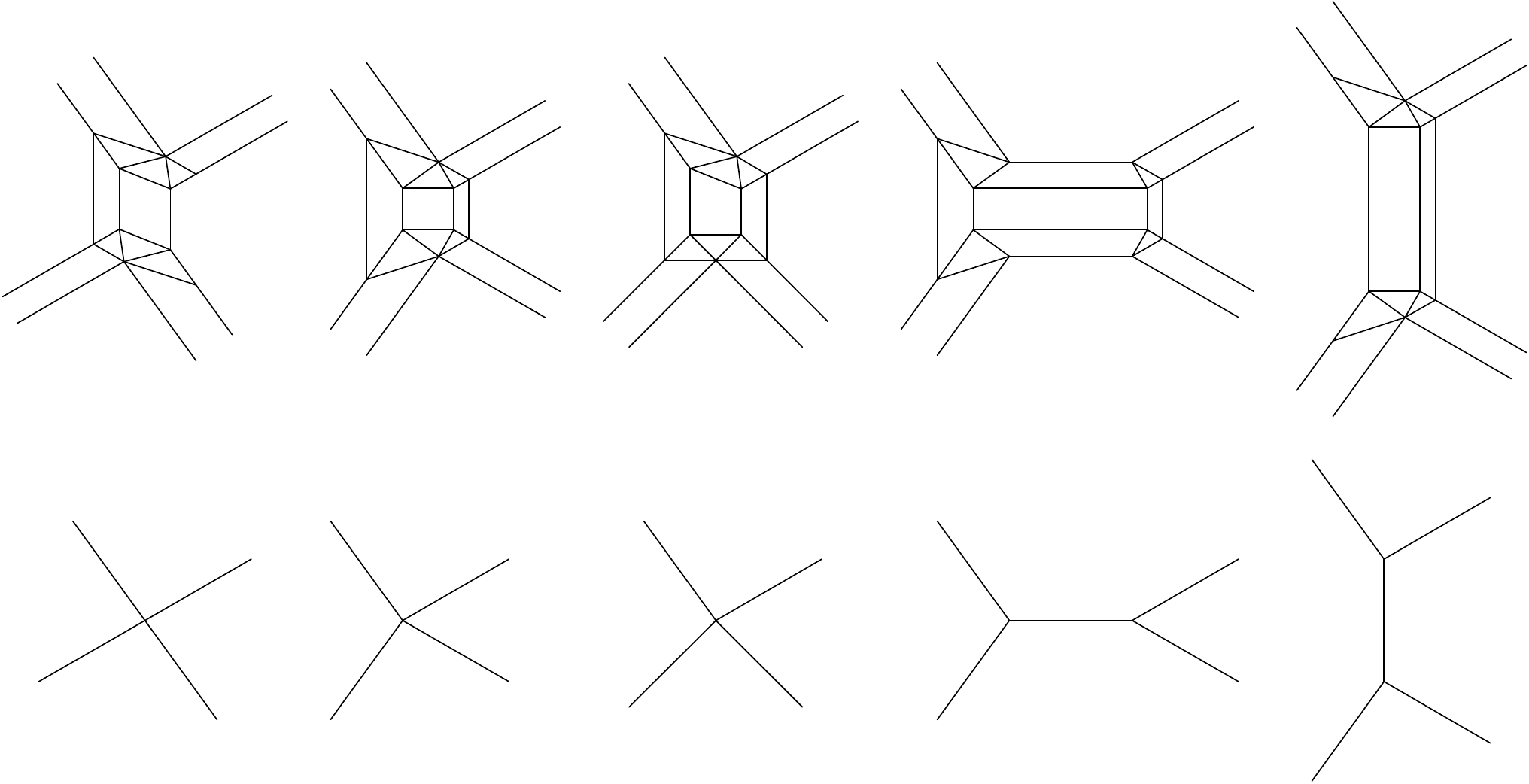}
    \caption{Upper: CPs of various shifts of outgoing pleats, Lower: Streaks appearing on the front side}
    \label{fig:shift_outgoing}
\end{center}
\end{figure}

\noindent$\bullet$ {\it Curved creases.}
Since the extrusion with our new $3$D gadgets has flat back sides, we can add new creases, which may be even curved, to the existing creases,
and deform the extruded object with the new creases. 
For example, as described in Introduction, we can extrude a flat-foldable cube using the crease pattern shown in Figure $\ref{fig:flat-foldable}$.
As another example, we can make an extrusion with contours consisting of semicircles by adding sine curves to the crease pattern of a cuboid.
The crease pattern of the extrusion and the side views are shown in Figure $\ref{fig:extrusion_curved}$, 
where the dotted lines are unfolded when we deform the cuboid with the curved creases.\\
\begin{figure}[htbp]
  \begin{center}
\addtocounter{theorem}{1}
          \includegraphics[width=0.8\hsize]{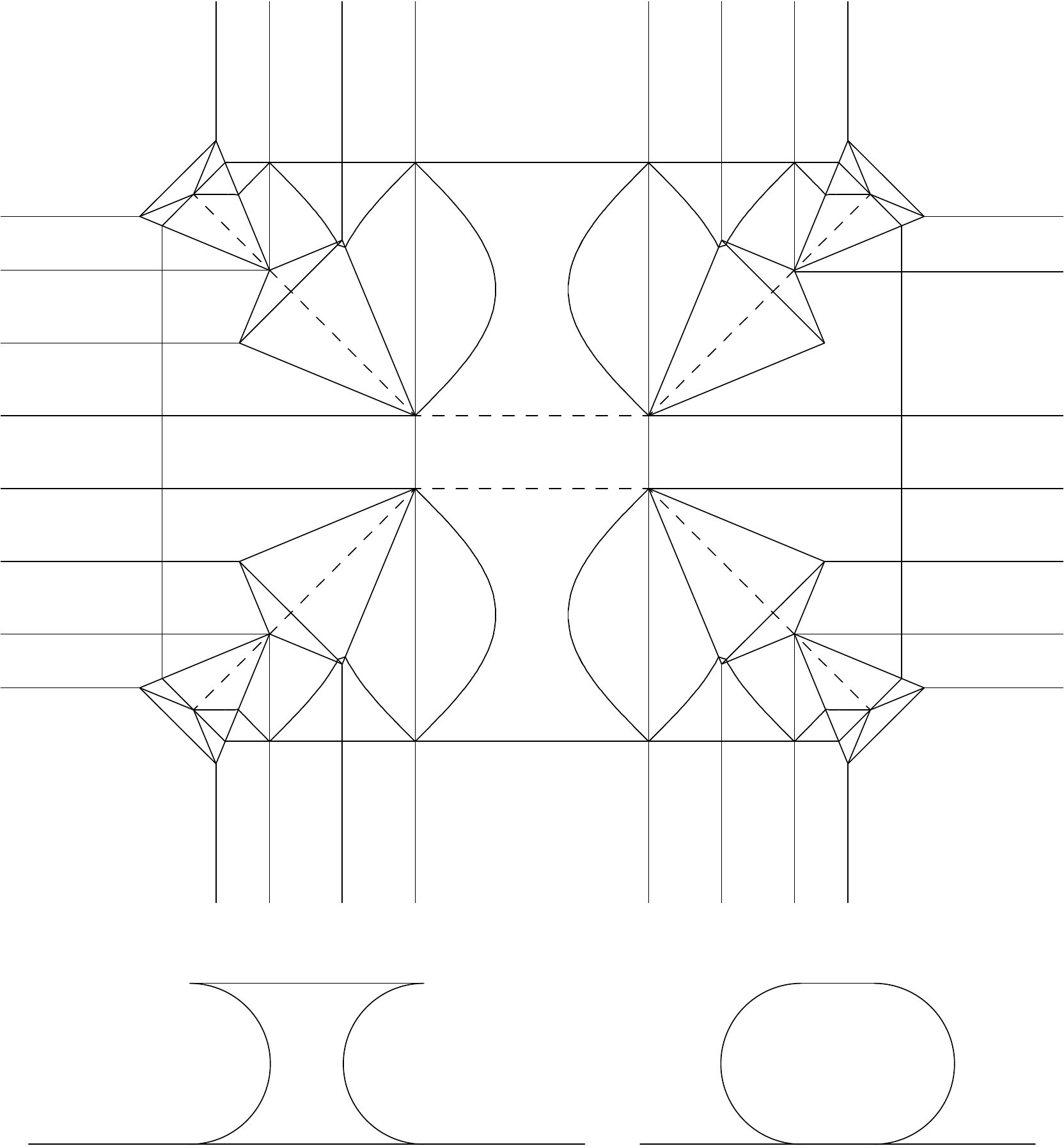}
    \caption{Upper: CP of an extrusion whose contour consists of semicircles, Lower: Side views of the resulting extrusion}
    \label{fig:extrusion_curved}
\end{center}
\end{figure}

\section{Conclusion}\label{sec:conclusion}
In the preceding sections, we presented the construction algorithms of our new $3$D gadgets and their variantions in origami extrusions
and described how to compute the interference coefficients to maximize the efficiency when extruding a polyhedron with the $3$D gadgets.
Using the interference coefficients, we proved the downward compatibility theorem of the new $3$D gadgets with the conventional ones,
and particularly we compared numerically the maximal height of the prism of a given convex polygon which can be extruded with the new $3$D gadgets
with that which can be extruded with the conventional ones.

Since for a given polyhedron there are many ways of designing the crease pattern of its extrusion,
selecting the optimal folding method and controlling the angles of the outgoing pleats will be of some interest in computer-aided design.
The new $3$D gadgets are theoretically always more efficient than the conventional ones when both kinds of gadgets are applicable, 
but the conventional gadgets are practically more suited when $\gamma$ is very small
because the new gadgets need finer creases and are more difficult to fold.
Also, if $\beta_\sigma +\gamma /4<\pi /2$, then we can only apply the conventional $3$D gadgets. 
On the other hand, when $\alpha$ is very small, the new gadgets may be the only practical solution because of the efficiency.
Thus we should use the two kinds of $3$D gadgets selectively in accordance with the situation, using the compatibility.

Althouth our new $3$D gadgets enable us to extrude a polyhedron which was impossible or very difficult to extrude with the conventional methods,
our methods are still very limited to extrude a general polyhedron such that many edges assemble at its vertex.
Studying the examples of a regular octahedron and a regular icosahedron extruded by Natan \cite{Natan} and the author \cite{Doi} respectively
may be a key to overcoming this difficulty.
There are still a lot to be studied including the above, and we expect further developments in future research on the basis of our results in this paper.

\end{document}